\documentclass[onecolumn,journal,noadjust]{IEEEtran}
\usepackage{cite}

\ifCLASSINFOpdf
\else
\fi
\usepackage{amssymb}
\usepackage{dsfont}
\usepackage{float}
\usepackage{algorithm}
\usepackage{algorithmic}
\usepackage{graphicx}
\newtheorem{lemma}{Lemma}

\newtheorem{definition}{Definition}

\newtheorem{theorem}{Theorem}
\newtheorem{remark}{Remark}
\newtheorem{corollary}{Corollary}
\usepackage{xcolor}
\usepackage{mathrsfs}
\usepackage{mathtools}
\usepackage{caption}
\usepackage{mathrsfs}
\usepackage[ colorlinks = true,              linkcolor = blue, urlcolor  =
blue,              citecolor = red,              anchorcolor = green,
]{hyperref}

\newcommand{\ind}{\perp \!\!\! \perp}

\usepackage [autostyle, english = american]{csquotes}
\MakeOuterQuote{"}

\usepackage{amsmath}
\DeclareMathOperator*{\argmax}{arg\,max}
\DeclareMathOperator*{\argmin}{arg\,min}
\begin{document}
%

\title{Weighted Unequal Error Protection over a\\Rayleigh Fading Channel}
%
%
%

\author{%
    Adeel Mahmood \\
    Radio Systems Research\\
    Nokia Bell Labs  
}

\maketitle

\begin{abstract}
We study a variant of unequal error protection in channel coding, where the message bit string is divided into a finite number of blocks and the maximization objective is a weighted sum of per-block decoding success probabilities. The channel model is quasi-static Rayleigh fading with channel state information available to the receiver but unavailable to the transmitter. We analyze the asymptotic and finite blocklength performance of two achievability schemes, one based on power-domain superposition (PDS) and another based on orthogonal resource allocation (ORA), also known as time-sharing. Upper bounds on the optimal number of blocks to transmit are derived. Algorithms to compute the optimal power and time splits for the two schemes are given. Simplified algorithms to compute locally optimal power and time splits are also given. Our results show that PDS outperforms ORA, but the performance differential is less than $2\%$ in both the asymptotic and finite blocklength regimes (Figures \ref{asymptotic_PDS_vs_ORA} - \ref{finite_blocklength_PDS_vs_ORA_5000}). For both PDS and ORA, numerical results also upper bound the gap between the asymptotic and finite blocklength performance by approximately $10\%$ for $n = 1000$ and $3\%$ for $n = 5000$ (Figures \ref{finite_blocklength_within_PDS_1000} - \ref{finite_blocklength_within_ORA_5000}).                  
\end{abstract}

\begin{IEEEkeywords}
Fading, AWGN, power allocation, channel coding, finite blocklength regime. 
\end{IEEEkeywords}
\section{Introduction}

Shannon's separation-based architecture proposed to split the task of transmitting an information source through a noisy channel into two separate stages: source coding of the information source into bits followed by transmitting the bits through a noisy channel via channel coding. While separate source-channel coding (SSCC) is first-order asymptotically optimal, joint source-channel coding (JSCC) has been shown to outperform SSCC in the finite blocklength regime via second-order asymptotic analysis based on the normal approximation \cite{kostina_JSCC}. Furthermore, while SSCC abstracts away what the bits mean and why they are being exchanged, communication performance in many modern applications is increasingly judged by whether the receiver can complete a task or recover the critical meaning, not by whether every payload bit is correct. This is encapsulated by semantic and task-oriented communications \cite{gunduz2022transmittingbitscontextsemantics}. Therefore, semantic communication through joint source-channel coding has received significant interest in recent years. Deep learning-based JSCC systems, for example, can exhibit graceful degradation, avoiding the sharp “cliff effect” of many separation-based pipelines when channel quality varies \cite{deepJSCC}.

However, deploying fully joint designs is difficult because real communication networks are inherently modular; application providers responsible for source coding and network providers responsible for channel coding are typically distinct entities connected through standardized packet interfaces \cite{homaspaper1}. It is impractical to  relay the real-time channel state information to the application provider to enable channel-aware source coding. Likewise, conveying the raw information source to the network providers is difficult.

Recent works \cite{homaspaper1}, \cite{homapaper2} proposed a new standard interface, called the multi-level reliability interface, between the application provider and network provider in order to (i) abstract the channel into multiple reliability levels for the application provider and (ii) simultaneously allow the application provider to indicate multiple levels of importance within the compressed information source to the network provider. With this multi-level reliability interface with $K$ reliability levels, the source coder is trained as a JSCC with respect to independent multi-level block erasure channels \cite{homapaper2} with erasure probabilities $\epsilon_1 < \epsilon_2 < \cdots < \epsilon_K$ advised by the multi-level reliability interface. Subsequently, the encoded bits given to the network provider can be partitioned into bit blocks, each block having a semantic importance assigned to it. A channel coder at the network provider can then optimize the transmission of these bits with respect to such a partition.  

Consider a $K$-level reliability interface and let $d_1 > d_2 > \cdots > d_K > 0$ denote the importance levels (implicitly specified by the interface through the erasure probabilities) for the bit blocks that are input to the network provider. We characterize the achievable channel coding performance for any given $K$ and parameters $d_1, \ldots, d_K$. This will be the main problem studied in this paper. Since the channel is known only to the receiver in many cases, the transmitter optimizes its transmission by assigning a Rayleigh distribution to the channel amplitude $|H|$. This problem then becomes a variant of bit-wise unequal error protection \cite{UEP_paper}; key aspects of our model/results include 
\begin{itemize}
    \item a random channel coefficient $H$ known only to the receiver,
    \item asymptotic analysis and finite blocklength analysis based on error exponent and normal approximation bounds, 
    \item explicit upper bounds on the optimal number of bit blocks to transmit as a function of $d_1, \ldots, d_K$, per-block coding rate $R$ and expected channel power gain (Theorems \ref{active_layer_theorem} and \ref{active_layer_theoremORA}),
    \item first-order asymptotic analysis of two achievability schemes: power-domain superposition (PDS) and orthogonal resource allocation (ORA),
    \item analytical and numerical comparisons between PDS and ORA in the asymptotic and finite blocklength regimes (Section \ref{finalcompnumerical}),
    \item showing that the ORA scheme achieves comparable performance to the PDS scheme (Figures \ref{asymptotic_PDS_vs_ORA}, \ref{finite_blocklength_PDS_vs_ORA_1000} and \ref{finite_blocklength_PDS_vs_ORA_5000}), 
    \item the achievable performance comparison for different values of $K$ (Figures \ref{partition_asymptotic} and \ref{partition_finite}).   
\end{itemize}

We now rigorously set up the problem. Define a channel encoder $\operatorname{f}$ and channel decoder $\operatorname{g}$ as 
\begin{align}
    &\operatorname{f} : \{0, 1 \}^k \to \mathbb{C}^n\\
    &\operatorname{g} : \mathbb{C}^n \to \{0, 1 \}^k. 
\end{align}
Let $B^k = (B_1^{m}, \ldots, B_K^m)$ be a binary string that represents a message to be transmitted, where $m = \frac{k}{K}$ is an integer. We assume that $B^k$ is uniformly distributed over $\{0,1 \}^k$, which implies that the bit blocks $B_i^m$'s are independent and each $B_i^m$ is uniformly distributed over the set $\{0,1 \}^m$. We will interchangeably use the terminology "layers" or "packets" to refer to the bit blocks. We use $\mathbf{X}$ and $\mathbf{Y}$ for the channel input and output vectors. We have 
\begin{align*}
    B^k \stackrel{\operatorname{f}}{\longrightarrow} \mathbf{X} \longrightarrow \operatorname{ channel } \longrightarrow \mathbf{Y} \stackrel{\operatorname{g}}{\longrightarrow} \widehat{B}^k,
\end{align*}
where $\mathbf{X}, \mathbf{Y} \in \mathbb{C}^n$. We define $R \coloneqq m/n$ as the per-block coding rate. For a given importance vector $\overrightarrow d = (d_1, \ldots, d_K)$, where $d_1 > d_2 > \cdots > d_K > 0$ and $\sum_{i=1}^K d_i = 1$, the goal is to design the channel code $(\operatorname{f}, \operatorname{g})$ to maximize 
\begin{align}
    \sum_{i=1}^K \mathbb{P}(B_i^{m} = \widehat{B}_i^{m}) d_i.  \label{operational_obj}
\end{align}
The maximization objective in $(\ref{operational_obj})$ is suitable for applications where certain packets of information take higher priority but each packet has a standalone value. For example, the packets $B_1^m, \ldots, B_K^m$ could represent distinct attributes or features in decreasing order of importance. Hence, any subset of blocks is useful to the end-receiver, e.g., for approximate reconstruction of the original source message \cite[Section III]{homapaper2}. 

Our channel model is a single-input, single-output quasi-static flat fading channel with input–output relation 
\begin{align}
    \mathbf{Y} = H \mathbf{X} + \mathbf{Z}, \label{channel_model} 
\end{align}
where $H \sim \mathcal{C} \mathcal{N}(0, \sigma^2)$ is the random channel coefficient which is constant over the whole length-$n$ block, $\mathbf{Z} \sim \mathcal{CN}(\mathbf{0}, \mathbf{I}_n)$, $\mathbf{X}$ and $\mathbf{Z}$ are independent and
\begin{align}
    \mathbb{E}\left [ ||\mathbf{X}||^2 \right] =  \sum_{i=1}^n \mathbb{E}\left [ |X_i|^2 \right] \leq nP. \label{expected_cost_constraint}
\end{align}
We assume that the channel state information (CSI) is available to the receiver only, but both the transmitter and receiver know the distribution of $H$. Our model is delay-limited and non-ergodic, which means that the receiver must decode by the end of each block and the performance is governed by outage-type behavior rather than ergodic averaging.

The first achievability scheme, which we call PDS, is based on a power-domain superposition (PDS) encoder $\operatorname{f}$ with a successive interference cancellation (SIC) decoder $\operatorname{g}$. Let $\alpha_1, \ldots, \alpha_K$ be nonnegative fractions such that $\sum_{i=1}^K \alpha_i = 1$, and $\alpha_i$ is the power fraction used for transmitting the $i$th bit block. For each bit block $B_i^m$, where $1 \leq i \leq K$, we construct a triple $(\mathcal{C}_i, \operatorname{f}_i, \operatorname{g}_i)$ where  
\begin{itemize}
    \item $\mathcal{C}_i \subset \mathbb{C}^n$ is a random codebook of size $2^m$ generated by drawing all codeword entries i.i.d.\ according to $\mathcal{CN}(0,1)$ and
independently across $i$.
\item $\operatorname{f}_i : \{0,1\}^m \to \mathcal{C}_i$ is the encoder mapping. In particular, $\mathbf{X}_i = \operatorname{f}_i(B_i^m)$ so that $\mathbf{X}_i \sim \mathcal{C N}(\mathbf{0}, \mathbf{I}_n)$ and $\mathbb{E}[||\mathbf{X}_i||^2] = n$ since $B_i^m \sim \operatorname{Unif}\left(\{0,1 \}^m \right)$.
\item $\operatorname{g}_i : \mathbb{C}^n \to \{0,1 \}^m$ is a maximum-likelihood decoder with known $H$ so that 
\begin{align*}
    \operatorname{g}_i(\mathbf{y}) &= \argmin_{B_i^m} \| \mathbf{y} - H \sqrt{\alpha_i P} \operatorname{f}_i(B_i^m) \|^2. 
\end{align*}

\end{itemize}
The PDS encoder transmits  
\begin{align}
    \mathbf{X} = \operatorname{f}(B^k) &= \sum_{i=1}^K  \sqrt{ \alpha_i P}\, \operatorname{f}_i(B_i^m)   \\
    &= \sum_{i=1}^K  \sqrt{ \alpha_i P}\, \mathbf{X}_i.  \label{xcons}
\end{align}

\begin{lemma}
    The construction of $\mathbf{X}$ in $(\ref{xcons})$ guarantees that $\mathbb{E}[||\mathbf{X}||^2] = nP$. 
    \label{power_cons_satisfied}
\end{lemma}
\textit{Proof:} The proof of Lemma \ref{power_cons_satisfied} is given in Appendix \ref{power_cons_satisfied_proof}.  
  
The SIC decoder $\operatorname{g}(\mathbf{Y}) = (\widehat{B}_1^m, \ldots, \widehat{B}_K^m)$ decodes in the order $1, \ldots, K$. Specifically, let $\mathbf{Y}^{(1)} = \mathbf{Y}$. Then for $i = 1, \ldots, K$, 
\begin{align*}
    \widehat{B}_i^m &= \operatorname{g}_i(\mathbf{Y}^{(i)})\\
    \mathbf{Y}^{(i + 1)} &= \mathbf{Y}^{(i)} - \sqrt{\alpha_i P} \operatorname{f}_i(\widehat{B}_i^m). 
\end{align*}

Using the PDS strategy is a layered broadcast coding approach that is often used to transmit layers of a successively refinable source \cite{SR_broadcast_andrea}, \cite{SR_via_broadcast}, \cite{utility_max}. Transmission of a Gaussian source, represented as multiple layers using successive refinement, is considered in both \cite{SR_broadcast_andrea} and \cite{SR_via_broadcast}, where the objective is to minimize the expected distortion. In \cite{utility_max}, the optimization objective is a finite weighted sum of marginal utilities, which is very similar to $(\ref{operational_obj})$. However, the $i$th marginal utility gain in \cite[(5)]{utility_max} requires cumulative success in decoding all the layers from $1$ to $i$, whereas $(\ref{operational_obj})$ rewards marginal correctness of each layer independently. We also assume that $B^k = (B_1^{m}, \ldots, B_K^m)$ is uniformly distributed and, therefore, each bit block $B_i^m$ does not necessarily represent a layer of a successively refinable source.

The use of PDS with SIC is optimal for a degraded Gaussian broadcast channel \cite{1055184}. Transmission across a fading AWGN channel with unknown CSIT (which is the channel model in our paper) admits a degraded Gaussian broadcast channel viewpoint \cite{1237140} for which PDS-SIC is optimal. Indeed, by using the PDS-SIC achievability scheme, we show that  $(\ref{operational_obj})$ can be asymptotically lower bounded by 
\begin{align}
    \sum_{i=1}^K  \exp \left(- \frac{\tau_i}{\sigma^2}  \right) d_i, \label{5-?}
\end{align}
where $\tau_1 \leq \cdots \leq \tau_K$ are the channel gains of a $K$-receiver degraded Gaussian broadcast channel with $K$ message sets of equal size $2^{nR}$, such that the equal rate-tuple $(R, \ldots, R)$ lies in the capacity region of the said $K$-receiver degraded Gaussian broadcast channel. For a fixed $R > 0$, this is a constraint on the set of feasible $\tau_1, \ldots, \tau_K$, and maximizing $(\ref{5-?})$ over the feasible $\tau_1, \ldots, \tau_K$ yields the best asymptotic lower bound to $(\ref{operational_obj})$ over the class of PDS schemes; see Section \ref{PDS_results} for a rigorous analysis of the PDS scheme without resorting to the degraded broadcast channel viewpoint.  

A key characteristic of the degraded Gaussian broadcast channel is that stronger receivers can decode everything intended for weaker receivers. However, the objective $(\ref{operational_obj})$ involves \emph{marginal} success probabilities and hence, does not mandate that the $i$th receiver decodes $B_i^m$ as well as $B_1^m, \ldots, B_{i-1}^m$ to obtain a "reward" of $d_i$. In other words, the success events may not be nested. Therefore, a converse establishing the optimality of PDS-SIC for the objective $(\ref{operational_obj})$ does not automatically follow from the known converse results for the degraded Gaussian broadcast channel \cite{1055184}. An investigation into whether $d_1 > \cdots > d_K$ in $(\ref{operational_obj})$ results in a nested structure without loss of optimality is left for future research. We instead compare the PDS scheme with another achievability scheme described next.

The second achievability scheme we consider is based on orthogonal resource allocation (ORA) encoding with maximum likelihood (ML) decoding of the orthogonal segments. This is also called time-sharing. Let $n_1, \ldots, n_K$ be nonnegative integers such that $\sum_{i=1}^K n_i = n$, and $n_i$ is the number of channel uses used for transmitting the $i$th bit block. For each bit block $B_i^m$, where $1 \leq i \leq K$, we construct a triple $(\mathcal{C}_i, \operatorname{f}_i, \operatorname{g}_i)$ where
\begin{itemize}
    \item $\mathcal{C}_i \subset \mathbb{C}^{n_i}$ is a random codebook of size $2^m$ generated by drawing all codeword entries i.i.d.\ according to $\mathcal{CN}(0,1)$ and
independently across $i$.
\item $\operatorname{f}_i : \{0,1\}^m \to \mathcal{C}_i$ is the encoder mapping. In particular, $\mathbf{X}_i = \operatorname{f}_i(B_i^m)$ so that $\mathbf{X}_i \sim \mathcal{C N}(\mathbf{0}, \mathbf{I}_{n_i})$ and $\mathbb{E}[||\mathbf{X}_i||^2] = n_i$ since $B_i^m \sim \operatorname{Unif}\left(\{0,1 \}^m \right)$.
\item $\operatorname{g}_i : \mathbb{C}^{n_i} \to \{0,1 \}^m$ is a maximum-likelihood decoder so that 
\begin{align*}
    \operatorname{g}_i(\mathbf{y}) &= \argmin_{B_i^m} \| \mathbf{y} - H \sqrt{P} \operatorname{f}_i(B_i^m) \|^2. 
\end{align*}
\end{itemize}
The ORA encoder then transmits
\begin{align*}
    \mathbf{X} = \sqrt{P}(\mathbf{X}_1, \ldots, \mathbf{X}_K).     
\end{align*}
The channel output can be written as $\mathbf{Y} = (\mathbf{Y}_1, \ldots, \mathbf{Y}_K)$ so that the ML decoder is simply 
\begin{align*}
    \operatorname{g}(\mathbf{Y}) &= ( \operatorname{g}_1(\mathbf{Y}_1), \ldots, \operatorname{g}_K(\mathbf{Y}_K)). 
\end{align*}
Unlike the PDS-SIC scheme, the ORA scheme does not impose a nested decoding structure and is therefore structurally better aligned with the objective $(\ref{operational_obj})$, while also allowing for a simpler practical implementation. Our results show that ORA achieves very close performance to the PDS both asymptotically and in the finite blocklength approximation.

In both PDS and ORA, we construct the random codebook for each bit block with codeword entries that are i.i.d. standard complex Gaussian. One reason for focusing only on an i.i.d. Gaussian ensemble is to simplify the analysis, especially in the finite blocklength regime. For the PDS scheme with SIC decoder, the i.i.d. Gaussian codewords make the "interference noise" Gaussian so that conditioned on the channel state $H = h$ known to the receiver and the correct decoding of prior bit blocks, the error analysis at each step $i$, where $1 \leq i \leq K$, is that of a simple additive white Gaussian noise channel with an explicitly determined SNR. Another reason for focusing only on an i.i.d. Gaussian codebook distribution is that the channel state information is not available to the transmitter (CSIT) so the encoder cannot fully optimize the selection of codewords. In the presence of CSIT, an optimal coding scheme is a variable-power scheme called truncated channel inversion \cite{7156144}. In our case, unavailable CSIT together with the delay constraint precludes fancier coding involving power optimization across different fading blocks.        

The rest of this paper is organized as follows. Section \ref{prelim_section} provides some notation, definitions and lemmas on the upper bounds to the average error probability across an AWGN channel. Section \ref{PDS_results} provides all the main results on the PDS scheme. Section \ref{ORA_results} provides all the main results on the ORA scheme. Section \ref{finalcompnumerical} provides a numerical comparison of the PDS and ORA schemes. Section \ref{combined_proof} provides the proof of Theorems \ref{active_layer_theorem}, \ref{global_maximizer_thm} and \ref{local_maximizer_thm} on the PDS scheme. Section \ref{combined_proofORA} provides the proof of Theorems \ref{active_layer_theoremORA}, \ref{global_maximizer_thmORA} and \ref{local_maximizer_thmORA} on the ORA scheme. All the remaining proofs are given in the appendices.     

\section{Preliminaries \label{prelim_section}}

We write $\log$ to denote logarithm to the base $2$ and $\ln$ to denote logarithm to the base $e$. Define the standard $(k-1)$-dimensional simplex as 
\begin{align*}
    \Delta^{K - 1} \coloneqq \left \{ \overrightarrow \alpha \in \mathbb{R}^K : \alpha_i \geq 0, \sum_{i=1}^K \alpha_i = 1 \right \}.
\end{align*}
For integers $a > b$, we adopt the convention that 
\begin{align}
    \sum_{i=a}^b x_i = 0. \label{sum_convention} 
\end{align}
Define $\psi : (0, \infty] \to [0, 4 e^{-2}]$ as $\psi(y) \coloneqq y^2 e^{-y}$. Key facts about $\psi(y)$:
\begin{itemize}
    \item $\psi(0^+) = 0$ and $\psi(y) \to 0$ as $y \to \infty$.
    \item $\psi(y)$ is unimodal, increasing for $y \leq 2$ and decreasing for $y \geq 2$.
    \item $\psi'(y) = y e^{-y}(2-y)$; therefore, $\psi(y)$ attains a maximum at $y = 2$ with $\psi(2) = 4e^{-2}$. 
\end{itemize}
Consider the equation 
\begin{align}
    \psi(y) = c  \label{fh} 
\end{align}
for some constant $c \in (0, 4 e^{-2})$. It is clear that $(\ref{fh})$ has two solutions $y^- \in (0, 2)$ and $y^+ \in (2, \infty)$. These two solutions can be written in terms of the Lambert function as described next. With a change of variable $w = -y/2$, equation $(\ref{fh})$ becomes 
\begin{align}
 w e^w = - \frac{\sqrt{c}}{2},   
\end{align}
where
\begin{align}
    -\frac{\sqrt{c}}{2} \in \left(-\frac{1}{e} ,0 \right). 
\end{align}
The two solutions of $(\ref{fh})$ are thus given by the two real brances of the Lambert function, the principal branch $W_0$ and the secondary branch $W_{-1}$ (see \cite[Section 4.13]{NISTHandbook}). Specifically,  
\begin{align}
    y^- &= -2 W_0\left(- \frac{\sqrt{c}}{2} \right) \in (0, 2)\\
    y^+ &= -2 W_{-1}\left( -\frac{\sqrt{c}}{2} \right) \in (2, \infty). 
\end{align}
In this paper, we consider the functions $W_0$ and $W_{-1}$ with domains $[-1/e, 0)$ so that 
\begin{itemize}
    \item $W_0 : [-1/e, 0) \to [-1, 0)$ and $W_0(x)$ is increasing, 
    \item $W_{-1} : [-1/e, 0) \to (-\infty, -1]$ and $W_{-1}(x)$ is decreasing in $x$. 
\end{itemize}
\begin{definition} $\mathrm{BisectionSearch}(f,\,a,\,b)$ denotes the standard bisection method that returns the unique root of $f$ in $[a,b]$, under assumptions ensuring that such a unique root exists. 
\end{definition}
\begin{remark}
    Throughout this paper, we invoke $\mathrm{BisectionSearch}(f,\,a,\,b)$ only when $f : \mathbb{R} \to \mathbb{R}$ is continuous and monotone on
    $[a,b]$ and has a unique root in the interval $[a, b]$. Hence, all instances of $\mathrm{BisectionSearch}(f,\,a,\,b)$ in the paper can be replaced by any comparable root-finding routine.    
\end{remark}

For $\mathbf{X} \sim \overline{P}$ and  channel $W$, we use $\overline{P} \circ W$ to denote the joint probability distribution and $\overline{P}W$ to denote the induced output distribution, i.e., $(\mathbf{X}, \mathbf{Y}) \sim \overline{P} \circ W$ and $\mathbf{Y} \sim \overline{P}W$.

The capacity $C(\rho)$ and dispersion $V(\rho)$ of a complex additive white Gaussian noise channel with SNR $\rho$ are given by
\begin{align*}
    C(\rho) &= \log(1 + \rho),\\
    V(\rho) &= \log^2(e) \frac{\rho(\rho + 2)}{(\rho + 1)^2}.
\end{align*}
Mathematically, 
\begin{align*}
    V(\rho) &= \frac{\log^2(e)}{n} \mathbb{E} \left [ \operatorname{Var}\left(  \ln \left(\frac{W(\mathbf{Y} | \mathbf{X})}{\overline{P}W(\mathbf{Y})} \right) \Bigg | \mathbf{X} \right) \right],
\end{align*}
where $(\mathbf{X}, \mathbf{Y}) \sim \overline{P} \circ W$, $\overline{P}= \mathcal{CN}(\mathbf{0}, \rho \mathbf{I}_n)$ and $W(\cdot|\mathbf{x}) = \mathcal{CN}(\mathbf{x}, \mathbf{I}_n)$. The dispersion $V(\rho)$ appears in the characterization of the optimal second- and higher-order coding performance under a maximal (per codeword) cost constraint. 
However, in this paper, we focus only on the achievability analysis for an i.i.d. complex Gaussian codebook where the channel input $\mathbf{X} \sim \mathcal{CN}(\mathbf{0}, \rho \mathbf{I}_n)$. In this case, the directly relevant quantity is the total information density variance $V_{\operatorname{tot}}(\rho)$ instead of the dispersion $V(\rho)$:  
\begin{align*}
    V_{\operatorname{tot}}(\rho) &\coloneqq \frac{\log^2(e)}{n} \operatorname{Var}\left(  \ln \left(\frac{W(\mathbf{Y} | \mathbf{X})}{\overline{P}W(\mathbf{Y})} \right) \right).  
\end{align*}
For random channel codes with equal-power channel input ($||\mathbf{X}||^2 = n \rho$ almost surely), we have $V(\rho) = V_{\operatorname{tot}}(\rho)$ (cf. \cite[Lemma 2]{timidboldadeel} for DMCs). But for $\mathbf{X} \sim \mathcal{CN}(\mathbf{0}, \rho \mathbf{I}_n)$, we have 
\begin{align*}
    V_{\operatorname{tot}}(\rho) &= V(\rho) + \log^2(e) \frac{\rho^2}{(1 + \rho)^2} = \log^2(e) \frac{2\rho}{1 + \rho}. 
\end{align*}

Finite blocklength analysis of the maximization objective $(\ref{operational_obj})$ will require non-asymptotic bounds on the average error probability $\mathbb{P}(B_i^m \neq \widehat{B}_i^m)$. We now state the error exponent \cite{gallager_paper_65}, \cite{gallager1968} and normal approximation bounds specialized to our i.i.d. codebook design. 

\begin{lemma}[Error Exponent]
    Consider a channel $W(\cdot|\mathbf{x}) = \mathcal{CN}(\mathbf{x}, \mathbf{I}_n)$. Consider a random channel code $(\operatorname{f}_n, \operatorname{g}_n)$ with codebook size $2^{nR}$ such that the channel input $\mathbf{X} \sim \mathcal{CN}(\mathbf{0}, \rho \mathbf{I}_n)$.  Then for any $n$ and $R$, the ensemble average error probability of this code is upper bounded by
    \begin{align}
        \exp \left(-n \max_{\lambda \in [0, 1]}\left [ \lambda \ln \left(1 + \frac{\rho}{1 + \lambda} \right) - \lambda R \ln(2) \right] \right). \label{error_exp_bound}
    \end{align}
    \label{error_exp_iid}
\end{lemma}
\textit{Proof:} See \cite[Theorem 10]{gallager_paper_65}. For completeness and to keep the paper self-contained, the proof is given in Appendix \ref{error_exp_iid_proof}.

\begin{lemma}[Normal Approximation]
    Consider a channel $W(\cdot|\mathbf{x}) = \mathcal{CN}(\mathbf{x}, \mathbf{I}_n)$. Consider a random channel code $(\operatorname{f}_n, \operatorname{g}_n)$ with codebook size $2^{nR}$ such that the channel input $\mathbf{X} \sim \mathcal{CN}(\mathbf{0}, \rho \mathbf{I}_n)$.  Then for any $n$ and $R$, the ensemble average error probability of this code is upper bounded by
    \begin{align}
        \min \left \{1, \Phi\left( \frac{\sqrt{n}\left(R - C(\rho) \right) + \frac{\log n}{2\sqrt{n}}}{\sqrt{V_{\operatorname{tot}}(\rho)}} \right) + \frac{2}{\sqrt{n}} \right \}. \label{normal_approximation}
    \end{align}
\label{lemma_happy_CLT}
\end{lemma}
\textit{Proof:} The proof is given in Appendix $\ref{lemma_happy_CLT_proof}$.

Error exponent bounds are most useful when the rate $R$ is fixed below capacity and the error probability is allowed to converge to zero as $n \to \infty$. Normal approximation bounds are useful when the desired error probability is fixed and the rate is allowed to converge to the capacity as $n \to \infty$. We will use the bounds $(\ref{error_exp_bound})$ and $(\ref{normal_approximation})$ in deriving finite blocklength achievability results for the objective $(\ref{operational_obj})$. However, in these achievability results (Theorems \ref{PDSfiniteblocklengthbound} and \ref{ORAfiniteblocklengthbound}), neither the rate nor the error probability can be assumed to be fixed and independent of $n$; even the effective capacity for each of the $K$ bit blocks can depend on $n$. This is because the achievability results involve an optimization over power allocations in the case of PDS (Theorem \ref{PDSfiniteblocklengthbound}) and resource allocations in the case of ORA (Theorem \ref{ORAfiniteblocklengthbound}). An optimal power allocation $\overrightarrow \alpha^\star$ and an optimal resource allocation $\overrightarrow v^\star$ will have dependence on $n$, which will cause the effective rate and SNR experienced by each of the $K$ bit blocks to also depend on $n$. Therefore, to obtain a tight upper bound for any given blocklength $n$, rate $R$ and SNR $\rho$, we take the minimum of the two upper bounds from Lemmas \ref{error_exp_iid} and \ref{lemma_happy_CLT}. Specifically, define 
\begin{align}
    \mathcal{E}_{\operatorname{nor}}(n, R, \rho) &\coloneqq \min \left \{1, \Phi\left( \frac{\sqrt{n}\left(R - C(\rho) \right) + \frac{\log n}{2\sqrt{n}}}{\sqrt{V_{\operatorname{tot}}(\rho)}} \right) + \frac{2}{\sqrt{n}} \right \}, \label{noreps}\\
    \mathcal{E}_{\operatorname{exp}}(n, R, \rho) &\coloneqq \exp \left(-n \max_{\lambda \in [0, 1]}\left [ \lambda \ln \left(1 + \frac{\rho}{1 + \lambda} \right) - \lambda R \ln(2) \right] \right),\\
    \mathcal{E}(n, R, \rho) &\coloneqq \min \left \{ \mathcal{E}_{\operatorname{nor}}(n, R, \rho), \mathcal{E}_{\operatorname{exp}}(n, R, \rho)  \right \}. \label{finitelengthupperbound}
\end{align}
By continuous extension, we define $\mathcal{E}(0, R, \rho) = 1$ for all $R > 0$ and $\rho \geq 0$. We define $\mathcal{E}(n, R, 0) = 1$ for all $R > 0$ and $n \geq 0$. Note that $\mathcal{E}(n, R, \rho)$ is well-defined even if $n \geq 0$ is not an integer. Furthermore, if $R$ and $\rho$ are constants independent of $n$, then  
\begin{align}
    \lim_{n \to \infty} \mathcal{E}(n, R, \rho) = \begin{cases}
        0 & R < \log(1 + \rho),\\
        1/2 & R = \log(1 + \rho),\\
        1 & R > \log(1 + \rho).
    \end{cases}
    \label{firstorderapproxRconstant}
\end{align}
Note that Lemmas \ref{error_exp_iid} and \ref{lemma_happy_CLT} do not require $R$ and $\rho$ to be constants independent of $n$. Figure \ref{error_bounds_comparison} shows a comparison of $\mathcal{E}_{\operatorname{nor}}(n, R, \rho)$ and $\mathcal{E}_{\operatorname{exp}}(n, R, \rho)$ for $\rho = 3$ and $n = 10000$. 

\begin{figure}[H]
    \centering
\includegraphics[width=10cm]{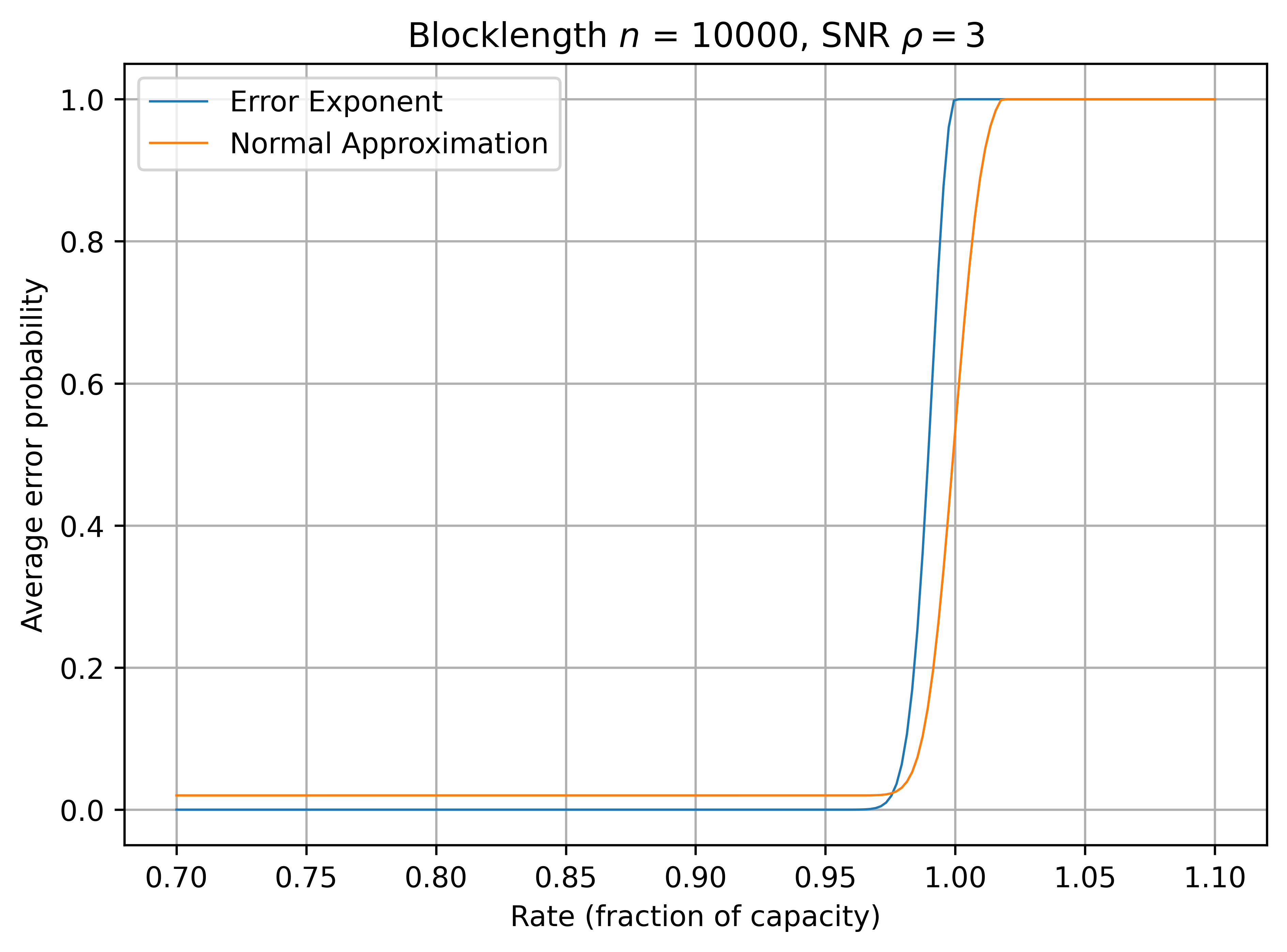}
\caption{The error exponent bound is tighter than the normal approximation bound in the low-rate regime, but the opposite is true in the high-rate/near-capacity regime. The horizontal axis is the coding rate represented as a fraction of the capacity $C(\rho) = \log(1 + \rho)$. }
\label{error_bounds_comparison}
\end{figure}

The proof of Lemma \ref{lemma_happy_CLT} is based on Shannon's achievability bound (see Appendix \ref{lemma_happy_CLT_proof}). For sufficiently large $n$ and
denoting the error probability by $\epsilon$, $(\ref{normal_approximation})$ implies \begin{align}
    \epsilon &\leq \Phi\left( \frac{\sqrt{n}\left(R - C(\rho) \right) + \frac{\log n}{2\sqrt{n}} }{\sqrt{V_{\operatorname{tot}}(\rho)}} \right) + O(n^{-1/2})\\
\implies nR &\geq n C(\rho) + \sqrt{n V_{\operatorname{tot}}(\rho)} \Phi^{-1}(\epsilon) - \frac{1}{2} \log n +O(1). \label{4x--}  
\end{align}
This is similar in form to the well-known third-order result
\begin{align}
    nR &\geq n C(\rho) + \sqrt{n V(\rho)} \Phi^{-1}(\epsilon) + \frac{1}{2} \log n +O(1) \label{herepoptmm}
\end{align}
shown by Tan and Tomamichel \cite{7056434} for the case when the channel input $\mathbf{X}$ is uniformly distributed on an $(n-1)$-sphere with radius $\sqrt{n \rho}$. The third-order term $1/2 \log n$ in $(\ref{herepoptmm})$ is shown to be optimal in the average error probability and maximal cost constraint framework \cite[Theorem 54]{ppv}. An improvement of the third-order $\log n$ term in $(\ref{4x--})$ to match that in $(\ref{herepoptmm})$ is possible by revising the proof of Lemma \ref{lemma_happy_CLT} to use Polyanskiy, Poor and Verdú’s random coding union bound as the starting point. Doing so and then using a similar proof technique\footnote{Although the argument from \cite{7056434} can be significantly simplified for our i.i.d. Gaussian case} as in \cite{7056434} would lead to a result of the form 
\begin{align}
        \min \left \{1, \Phi\left( \frac{\sqrt{n}\left(R - C(\rho) \right) - \frac{\log n}{2\sqrt{n}} + \frac{C_1}{\sqrt{n}}}{\sqrt{V_{\operatorname{tot}}(\rho)}} \right) + \frac{C_2}{\sqrt{n}} \right \} \label{norm2al_approximation}
    \end{align}
in lieu of $(\ref{normal_approximation})$, where $C_1$ and $C_2$ are some positive constants. While $(\ref{norm2al_approximation})$ is \emph{asymptotically} a tighter upper bound than $(\ref{normal_approximation})$, the constants $C_1$ and $C_2$ are hard to write down explicitly without assuming sufficiently large $n$, and the fourth-order $1/\sqrt{n}$ term in $(\ref{norm2al_approximation})$ comprising the constants $C_1$ and $C_2$ is worse than that in $(\ref{normal_approximation})$. Since we do not assume sufficiently large $n$, it is not clear that $(\ref{norm2al_approximation})$ is a tighter upper bound than $(\ref{normal_approximation})$ for small or even moderate values of $n$; hence, we will use the bound $(\ref{normal_approximation})$ owing to its simplicity and explicit constants.
        
We also emphasize that many finite blocklength results in channel coding \cite{7056434}, \cite{ppv}, \cite{7156144} are \emph{refined asymptotic} results such as those of the form $(\ref{herepoptmm})$. In such asymptotic expansions, the highest-order $O(1)$ term is usually discarded and the resulting truncated expression is used as a finite blocklength approximation (see, e.g., \cite[(296)]{ppv}). In this paper, however, all finite blocklength results (Lemmas \ref{error_exp_iid} and \ref{lemma_happy_CLT} and Theorems \ref{PDSfiniteblocklengthbound} and \ref{ORAfiniteblocklengthbound}) will be strictly non-asymptotic that hold for any finite $n \geq 1$.

\section{Main Results on the PDS scheme \label{PDS_results}}

\begin{theorem}
    Under the PDS scheme, we have   
    \begin{align}
        \sum_{i=1}^K \mathbb{P}(B_i^{m} = \widehat{B}_i^{m}) d_i \geq \max_{\overrightarrow \alpha \in \Delta^{K - 1}} \sum_{i=1}^K \mathbb{E}_{\gamma}\left [\prod_{j=1}^i \left(1 - \mathcal{E}\left(n, R, \frac{\gamma \alpha_j P}{1 + \gamma P \beta_j} \right) \right) \right] d_i \label{finitenPDSopt}
    \end{align}
    for any $n \geq 1$, where $\beta_i \coloneqq \sum_{j=i + 1}^K \alpha_j$ and the expectation is w.r.t. an exponential random variable $\gamma$ with mean $\sigma^2$.  
    \label{PDSfiniteblocklengthbound}
\end{theorem}

 \begin{IEEEproof}
     Conditioned on the fading state $|H|^2 = \gamma$ and the correct decoding and subtraction of the prior bit blocks $1, \ldots, i - 1$, the signal-to-interference-plus-noise ratio (SINR) at step $i$ is 
\begin{align}
    \rho_i(\gamma) &\coloneqq \frac{\gamma \alpha_i P}{1 + \gamma P \beta_i}, \label{sinrdef}
\end{align}
where $\beta_i = \sum_{j=i + 1}^K \alpha_j$ is the residual power\footnote{Note that $\beta_K = 0$ by our convention in $(\ref{sum_convention})$.} in the not-yet-decoded lower layers. Hence, conditioned on the fading state $|H|^2 = \gamma$ and the event, denoted as $E_i$, that the prior blocks $1, \ldots, i - 1$ are decoded correctly, layer $i$ effectively experiences a complex additive white Gaussian noise channel with SNR given by $(\ref{sinrdef})$. This argument relies on the fact that the random codebooks $\{\mathcal{C}_i \}_{i=1}^K$ for coding each bit block are generated i.i.d. $\mathcal{C} \mathcal{N}(0, 1)$, so that the interference "noise" from the undecoded lower layers is Gaussian. Hence, we can use the finite blocklength bound $(\ref{finitelengthupperbound})$ to write 
\begin{align*}
    \mathbb{P}\left(B_i^m \neq \widehat{B}_i^m \big | E_i, |H|^2 =  \gamma \right) \leq \mathcal{E}\left(n, R, \rho_i(\gamma) \right) 
\end{align*}
for $i \in \{1, \ldots, K \}$, where
\begin{align*}
    E_i &= \bigcap_{j=1}^{i-1} \{B_j^m = \widehat{B}_j^m \}
\end{align*}
and $E_1 = \Omega$ by convention. Then 
\begin{align*}
    \mathbb{P}(B_i^{m} = \widehat{B}_i^{m} \big | |H|^2 =  \gamma) &\geq \mathbb{P}\left( \bigcap_{j=1}^i \{ B_j^{m} = \widehat{B}_j^{m} \} \Big | |H|^2 = \gamma \right)\\
    &\geq \prod_{j=1}^i \left(1 - \mathcal{E}\left(n, R, \rho_j(\gamma) \right) \right).
\end{align*}
Taking the expectation over $\gamma \sim \operatorname{Exp}(\sigma^2)$ establishes the result. 
 \end{IEEEproof}

Define 
\begin{align}
    G_n(\overrightarrow \alpha) \coloneqq \sum_{i=1}^K \mathbb{E}_{\gamma}\left [\prod_{j=1}^i \left(1 - \mathcal{E}\left(n, R, \frac{\gamma \alpha_j P}{1 + \gamma P \beta_j} \right) \right) \right] d_i. \label{tt33}
\end{align}

\begin{theorem}
    Let $R > 0$ be any constant independent of $n$. Then 
    \begin{align}
        \lim_{n \to \infty}\,\, \max_{\overrightarrow \alpha \in \Delta^{K - 1}}\,G_n(\overrightarrow \alpha) &= \max_{\overrightarrow \alpha \in \Delta^{K - 1}} \sum_{i=1}^K  \exp \left(- \frac{\max \{\tau_1, \ldots, \tau_i \}}{\sigma^2}  \right) d_i, \label{nqu}
    \end{align}
    where 
    \begin{align}
    \tau_i =  \begin{cases}
         \frac{2^{R} - 1}{P(\alpha_i -(2^{R} - 1) \beta_i)} & \text{ if } \alpha_i -(2^{R} - 1) \beta_i > 0,\\
         +\infty & \text{ otherwise}. 
    \end{cases} 
    \label{deftau}
\end{align}
\label{finitetofirstorderreduction}    
\end{theorem}
\textit{Proof:} The proof of Theorem \ref{finitetofirstorderreduction} is given in Appendix \ref{finitetofirstorderreduction_proof}.

\begin{definition}
We define a first-order asymptotically optimal power split to be a solution to the following optimization problem:
\begin{align}
    \max_{\overrightarrow \alpha \in \Delta^{K - 1}} \sum_{i=1}^K  \exp \left(- \frac{\max \{\tau_1, \ldots, \tau_i \}}{\sigma^2}  \right) d_i. \label{a1}
\end{align}
\label{firstorderPDSsol}
\end{definition}

Since directly solving the RHS of $(\ref{finitenPDSopt})$ is both analytically and numerically difficult, we first focus on solving for the first-order asymptotically optimal solution $\overrightarrow \alpha^\star$ in $(\ref{a1})$, which gives us the asymptotic performance of the PDS scheme by Theorem \ref{finitetofirstorderreduction}. Most of our  analytical results will focus on characterizing the solution $\overrightarrow \alpha^\star$ in $(\ref{a1})$. In Section \ref{finalcompnumerical}, we make use of $\overrightarrow \alpha^\star$ in approximating the RHS of $(\ref{finitenPDSopt})$.

The optimization problem $(\ref{a1})$ is equivalent to   
\begin{align}
\max_{\overrightarrow \alpha \in \Delta^{K - 1}} \sum_{i=1}^K  \exp \left(- \frac{\tau_i}{\sigma^2}  \right) d_i. \label{c1}
\end{align}
This follows from Theorem \ref{opt_properties_theorem} and Corollary \ref{corl4} below.  

\begin{theorem}[Structural Properties of an optimal solution]
    Let $\overrightarrow \alpha$ be an optimal solution in $(\ref{c1})$. Then the following hold:  
    \begin{enumerate}
        \item\label{opt_prop_1} $\tau_1 \leq \tau_2 \leq \cdots \leq \tau_K$. 
        \item \label{opt_prop_2} $\tau_i = + \infty$ $\iff$ $\alpha_i = 0$.
        \item \label{opt_prop_3} $\tau_1 < \infty$ and $\alpha_1 > 0$.
     \item \label{opt_prop_4} $\alpha_i - (2^R - 1) \beta_i \geq 0$ for all $1 \leq i \leq K$.
     \item \label{opt_prop_5} For each $i \in \{2, \ldots, K \}$, we have  $\alpha_i - (2^R - 1) \beta_i = 0$ $\iff$ $\beta_{i-1} = 0$. 
     \item \label{opt_prop_6} Let $\ell \in \{1, \ldots, K \}$ be the largest integer such that $\tau_i < \infty$ for all $i \leq \ell$. Then $\tau_1 < \cdots < \tau_\ell < \infty$ or equivalently, $\alpha_i > 2^R \alpha_{i + 1}$ for all $i \in \{1, \ldots, \ell \}$. 
    \end{enumerate}
    \label{opt_properties_theorem}
\end{theorem} 

\textit{Proof:} The proof of Theorem $\ref{opt_properties_theorem}.(\ref{opt_prop_1})$ is given in Appendix \ref{opt_prop_1_proof}. The proof of Theorem $\ref{opt_properties_theorem}.(\ref{opt_prop_2})$ is given in Appendix \ref{opt_prop_2_proof}. The proof of the remaining points is given in Appendix \ref{3-6_proof}.

\begin{corollary}
    We have 
    \begin{align}
        \max_{\overrightarrow \alpha \in \Delta^{K - 1}} \sum_{i=1}^K  \exp \left(- \frac{\max \{\tau_1, \ldots, \tau_i \}}{\sigma^2}  \right) d_i &= \max_{\overrightarrow \alpha \in \Delta^{K - 1}} \sum_{i=1}^K  \exp \left(- \frac{\tau_i}{\sigma^2}  \right) d_i,\\
        \argmax_{\overrightarrow \alpha \in \Delta^{K - 1}} \sum_{i=1}^K  \exp \left(- \frac{\max \{\tau_1, \ldots, \tau_i \}}{\sigma^2}  \right) d_i &= \argmax_{\overrightarrow \alpha \in \Delta^{K - 1}} \sum_{i=1}^K  \exp \left(- \frac{\tau_i}{\sigma^2}  \right) d_i.
    \end{align}
    \label{corl4}
\end{corollary}
\textit{Proof:} The proof of Corollary \ref{corl4} is given in Appendix \ref{corl4_proof}.

Before stating subsequent results on the optimal solution in $(\ref{c1})$, we define 
\begin{align}
    \theta \coloneqq \frac{2^R - 1}{P \sigma^2}. \label{thetadef}
\end{align}
This definition will remain in effect throughout the paper. The numerator is the threshold SNR needed to successfully transmit one bit block of size $m = nR$. The denominator is the average SNR of the channel. Hence, $\theta$ is the threshold-to-average SNR ratio. All of our results will characterize an optimal solution $\overrightarrow \alpha^\star$ in terms of $\theta$. Results on the orthogonal resource allocation scheme will also be presented in terms of $\theta$.

\begin{theorem}
Let $K = 2$. Let $\overrightarrow \alpha^\star = (\alpha^\star, 1- \alpha^\star)$ denote an optimal solution in $(\ref{c1})$. Define 
\begin{align*}
    \xi \coloneqq \left( \frac{\theta 2^{R/2}}{1 +  \sqrt{1-2^R\theta^2 }}\right)^2  \exp \left( \theta(2^R - 1) +  2 \sqrt{1 - 2^R \theta^2} \right).
\end{align*}
Then the following hold:  
    \begin{itemize}
        \item If $\theta \geq 2^{-R/2}$, then $\alpha^\star = 1$.
        \item If $0 < \theta < 2^{-R/2}$ and $\frac{d_2}{d_1} \leq \xi$, then $\alpha^\star = 1$.
        \item If $0 < \theta < 2^{-R/2}$ and $\frac{d_2}{d_1} > \xi$, then there exists a unique solution $q_0$ to
        \begin{align}
            \frac{2^R}{q^2} \exp \left( \theta \left(q + 2^R - 1 -  \frac{2^R}{q} \right) \right) = \frac{d_2}{d_1} \label{equationone}
        \end{align}
        over the interval 
        \begin{align}
            q_0 \in \left( \frac{1}{\theta} - \sqrt{\frac{1}{\theta^2} - 2^R}, \frac{1}{\theta} + \sqrt{\frac{1}{\theta^2} - 2^R}  \right). \label{int_q}
        \end{align}
        Then 
        \begin{align*}
            \alpha^\star &= \begin{cases}
                \frac{q_0 + 2^R - 1}{q_0 + 2^R} & \text{ if } \exp\left(-\frac{\theta 2^R}{q_0} \right) \left( 1 + \frac{2^R}{q_0^2} \right) > 1\\
                1 & \text{ otherwise}.
            \end{cases}
        \end{align*}
    \end{itemize}
    \label{K=2theorem}
\end{theorem}
\textit{Proof:} The proof of Theorem \ref{K=2theorem} is given in Appendix \ref{K=2theorem_proof}.
\begin{remark}
    The LHS of $(\ref{equationone})$ is strictly decreasing over the interval given in $(\ref{int_q})$ so the unique solution $q_0$ can be obtained by a simple bisection search.  
\end{remark}

To solve for the general $K \geq 2$ case, we first perform a change of variable. Define 
\begin{align}
    \Delta_+^{K-1} &\coloneqq \left \{\overrightarrow \alpha \in \Delta^{K-1}:  \alpha_i - (2^R - 1) \beta_i \geq 0 \text{ for all } 1\leq i \leq K \right \},\\
    \mathcal{S}_K &\coloneqq \left  \{\overrightarrow x \in \mathbb{R}^K: \sum_{i=1}^{K} (2^R)^{i-1} x_i = 1 \text{ and } x_j \geq 0 \text{ for all } 1 \leq j \leq K   \right \}. \label{consonxx}
\end{align}

\begin{lemma}
Let $\overrightarrow x = M_B(\overrightarrow \alpha)$ be specified as 
$x_i = \alpha_i - (2^R - 1)\beta_i$ for all $1 \leq i \leq K$. Then $M_B$ is a bijection from $\Delta_+^{K-1}$ to $\mathcal{S}_K$ with the inverse mapping $\overrightarrow{\alpha} = M_B^{-1}(\overrightarrow x)$ given by $\alpha_K = x_K$ and 
\begin{align}
    \alpha_{i} &= x_{i} + (2^R - 1)\sum_{j=i + 1}^{K} 2^{R(j-i-1)} x_{j} \label{xtoa}
\end{align}
for all $1 \leq i \leq K - 1$. 
\label{atox}
\end{lemma}
\textit{Proof:} The proof of Lemma \ref{atox} is given in Appendix \ref{atox_proof}. 

\begin{remark}
    The mapping $\overrightarrow x \mapsto \overrightarrow \alpha$ in $(\ref{xtoa})$ can be implemented by the following backward recursion: $\alpha_K = x_K$ and for $i = K - 1, \ldots, 1$, 
    \begin{align}
    \begin{split}
        \beta_i &= \sum_{j = i + 1}^K \alpha_j\\
        \alpha_i &= x_i + (2^R - 1) \beta_i.
    \end{split}
    \label{xtoa2}
    \end{align}
\end{remark}

Define $g : [0, 1] \to [0, 1)$ as 
\begin{align}
  g(x) &\coloneqq \exp \left(- \frac{\theta}{x}  \right) \label{smallgdef}
\end{align}
with the convention $$g(0) = \lim_{x \downarrow 0} g(x) =  0.$$
We then have 
    \begin{align}
        \max_{\overrightarrow \alpha \in \Delta^{K-1}} \sum_{i=1}^K  \exp \left(- \frac{\tau_i}{\sigma^2}  \right) d_i = \max_{\overrightarrow \alpha \in \Delta^{K-1}_+} \sum_{i=1}^K  \exp \left(- \frac{\tau_i}{\sigma^2}  \right) d_i = \max_{\overrightarrow x \in \mathcal{S}_K} \sum_{i=1}^K g(x_i) d_i, \label{lhs=rhsopt} 
    \end{align}
where the first equality above follows from Theorem $\ref{opt_properties_theorem}.(\ref{opt_prop_4})$ and the second equality above follows from Lemma \ref{atox}. Hence, if $\overrightarrow{x}^\star = (x_1^\star, \ldots, x_K^\star)$ is a maximizer in the RHS of $(\ref{lhs=rhsopt})$, then an optimal power split $\overrightarrow \alpha^\star$ can be constructed from $\overrightarrow{x}^\star$ using $(\ref{xtoa})$ or $(\ref{xtoa2})$. We thus focus our attention on solving the optimization problem 
\begin{align}
    \max_{\overrightarrow x \in \mathcal{S}_K} G(\overrightarrow x),  \label{b3}
    \end{align}
where we define 
\begin{align*}
    G(\overrightarrow x) \coloneqq \sum_{i=1}^K g(x_i) d_i.
\end{align*}
A maximizer in $(\ref{b3})$ exists because the objective is continuous and the feasible set $\mathcal{S}_K$ is compact (note that $\overrightarrow x \in \mathcal{S}_K$ satisfies $x_i \leq 1$ for all $i$). 
Given the bijection in Lemma \ref{atox}, we can directly map the structural properties of an optimal solution $\overrightarrow \alpha^\star$ in $(\ref{c1})$ as established in Theorem \ref{opt_properties_theorem} to the optimality conditions for a solution in $(\ref{b3})$.  
\begin{corollary}
    Let $\overrightarrow x^\star$ be an optimal solution in $(\ref{b3})$. Then the following hold: 
    \begin{enumerate}
        \item $x_1^\star \geq x_2^\star \geq \cdots \geq x_K^\star$,
        \item $x_1^\star > 0$,
        \item \label{opt_prop_3x}  For each $i \in \{2, \ldots, K \}$, we have $x_i^\star = 0$ $\iff$ $x_j^\star = 0$ for all $i \leq j \leq K$.
        \item \label{opt_prop_4x} Let $\ell \in \{1, \ldots, K \}$ be the largest integer such that $x_i^\star > 0$ for all $i \leq \ell$. Then $x_1^\star > \cdots > x_\ell^\star > 0$. 
    \end{enumerate}
    \label{optxproperties}
\end{corollary}

Note that 
\begin{align*}
    \max\{1 \leq i \leq K: x_i > 0 \} = \max\{1 \leq i \leq K: \alpha_i > 0 \}, 
\end{align*}
so the parameter $\ell$ defined in Corollary $\ref{optxproperties}.(\ref{opt_prop_4x})$ denotes the optimal number of bit blocks to transmit. Theorem \ref{active_layer_theorem} gives an explicit upper bound on $\ell$ in terms of $\theta, \overrightarrow d$ and rate $R$. 

Using the optimality conditions in Corollary \ref{optxproperties}, the KKT conditions and the second-order necessary and sufficient conditions based on the Lagrangian Hessian, we prove the following series of theorems.

\begin{theorem}
\label{active_layer_theorem}
Let $\overrightarrow x^\star$ denote an optimal solution in $(\ref{b3})$. Then 
\begin{align}
    x_i^\star \begin{cases}
        > 0 & \text{ for } i \in \{1, \ldots, \ell \}\\
        = 0 & \text{ otherwise}, 
    \end{cases} 
    \label{skfds}
\end{align}
where
    \begin{align}
        \ell \leq \ell_{\operatorname{PDS}} \coloneqq \begin{cases}
            1 & \text{ if } \theta \geq 2\\
            \max \left \{1 \leq i \leq K: \frac{2^{iR}}{d_i} \leq \left(\frac{2^R}{d_1} \right) \frac{4e^{-2}}{\theta^2 e^{-\theta}}  \right \} & \text{ if } \theta < 2.
        \end{cases} \label{elllmdef}    
    \end{align}
    In particular, if
    \begin{align}
        \theta > -2W_0 \left(-\frac{1}{e} \sqrt{\frac{d_2}{2^R d_1}} \right), \label{particularthm}    
    \end{align}
     then $\ell = 1$ and $\overrightarrow x^\star = (1, 0, \ldots, 0)$. 
\end{theorem}

\begin{theorem}
\label{global_maximizer_thm}
    An optimal solution $\overrightarrow x^\star$ in $(\ref{b3})$ is given by Algorithm $\ref{algorithmglobalmax}$.
\end{theorem}

\begin{theorem}
\label{local_maximizer_thm}
A strict local maximizer $\overrightarrow x_{\operatorname{loc}}^\star$ in $(\ref{b3})$ is given by 
Algorithm \ref{algorithmlocalmax} such that $\overrightarrow x_{\operatorname{loc}}^\star$ satisfies the KKT conditions, the optimality conditions of Corollary \ref{optxproperties} and 
    \begin{align}
    \overrightarrow x_{\operatorname{loc}}^\star(i) 
    \begin{cases}
        > 0 & \text{ for } i \in \{1, \ldots, \ell \}\\
        = 0 & \text{ otherwise}, 
    \end{cases} 
    \label{gf3}
\end{align}
where 
    \begin{align}
    \ell \in \mathcal{L}_p\left(\theta, R, \overrightarrow d\right) &\coloneqq \left \{1 \leq j \leq \ell_{\operatorname{PDS}}:  H_j^-(\lambda_{\min}) \geq 1 \geq H_{j}^-(\lambda_{\max}(j))  \right \}, \label{lspecification} \\
    H_{j}^-(\lambda) &\coloneqq \sum_{i=1}^j 2^{R(i-1)} \frac{\theta}{-2W_0\left(-\frac{1}{2} \sqrt{c_i(\lambda)}  \right)} \quad \quad  \text{for each }  j = 1, \ldots, \ell_{\operatorname{PDS}},\\
    c_i(\lambda) &\coloneqq \frac{\lambda \theta 2^{R(i-1)}}{d_i} \quad \quad \quad \quad \quad \quad \quad \quad \quad \quad \quad    \text{for each }  i = 1, \ldots, \ell_{\operatorname{PDS}},\\
    \lambda_{\min} &\coloneqq d_1 \theta e^{-\theta},\\
    \lambda_{\max}(j) &\coloneqq \frac{4 e^{-2} d_{j}}{\theta 2^{R(j-1)}} \quad \quad\quad\quad\quad\quad\quad\quad\quad\quad\quad\,\,\,  \text{for each }  j = 1, \ldots, \ell_{\operatorname{PDS}},
\end{align}
and $\ell_{\operatorname{PDS}}$ is defined in $(\ref{elllmdef})$.
\end{theorem} 

The combined proof of Theorems $\ref{active_layer_theorem}, \ref{global_maximizer_thm}$ and $\ref{local_maximizer_thm}$ is given in Section \ref{combined_proof}. 

\subsection{Summary of Algorithms \ref{algorithmglobalmax} and \ref{algorithmlocalmax}}

Algorithm \ref{algorithmglobalmax} is based on the fact\footnote{proven in Section \ref{combined_proof}} that an optimal $\overrightarrow x^\star$ in $(\ref{b3})$ is either equal to $(1, 0, \ldots, 0)$ or satisfies  
\begin{align}
    x_i^\star = x_i^\star(\lambda) = \begin{cases}
        \frac{\theta}{-2W_0\left(-\frac{1}{2} \sqrt{c_i(\lambda)}  \right)} & \text{ for } i \in \{1, \ldots, \ell - 1 \}\\
        \frac{\theta}{-2W_0\left(-\frac{1}{2} \sqrt{c_\ell(\lambda)}  \right)} \text{ or } \frac{\theta}{-2W_{-1}\left(-\frac{1}{2} \sqrt{c_\ell(\lambda)}  \right)} & \text{ for } i = \ell \\
        0 & \text{ for } i \in \{\ell + 1, \ldots, K \}
    \end{cases}
    \label{maxi_candsPDS}
\end{align}
for some $2\leq \ell \leq \ell_{\operatorname{PDS}}$ and $\lambda \in [\lambda_{\min}, \lambda_{\max}(\ell)]$ such that
\begin{align}
    \sum_{i=1}^\ell 2^{R(i-1)} x_i^\star(\lambda) = 1. \label{Hconsgrk}
\end{align}
Algorithm \ref{algorithmglobalmax} first checks a sufficient (but not necessary) condition for $\overrightarrow x^\star = (1, 0, \ldots, 0)$. If the condition is not satisfied, then  Algorithm \ref{algorithmglobalmax} searches across the values $\ell \in \{1, \ldots, \ell_{\operatorname{PDS}} \}$ and, for each $\ell$, constructs the maximizer candidates of the form $(\ref{maxi_candsPDS})$ by solving for $\lambda \in [\lambda_{\min}, \lambda_{\max}(\ell)]$ satisfying $(\ref{Hconsgrk})$. For $\ell = 1$, the only maximizer candidate is $(1, 0, \ldots, 0)$. The constraints $\ell \leq \ell_{\operatorname{PDS}}$ and $\lambda \in [\lambda_{\min}, \lambda_{\max}(\ell)]$ guarantee that 
\begin{align*}
    - \frac{1}{e} \leq -\frac{1}{2} \sqrt{c_i(\lambda)} < 0
\end{align*}
for all $1 \leq i \leq \ell$ so that $W_0$ and $W_{-1}$ are well-defined. Algorithm \ref{algorithmlocalmax} uses the same logic but restricts itself to the principal Lambert branch $W_0$ for $i = \ell$. In Algorithm \ref{algorithmlocalmax}, for each $\ell$, the solution $\lambda$ to the equation $(\ref{Hconsgrk})$ is unique if it exists, and a necessary and sufficient condition for the existence is easily specified, making Algorithm \ref{algorithmlocalmax} simpler to implement. But when $x_\ell^*(\lambda)$ is chosen as the secondary Lambert branch $W_{-1}$, then we show in Section \ref{combined_proof} that there can be at most 2 values of $\lambda$ satisfying the equation $(\ref{Hconsgrk})$; this solution set is the output of Algorithm \ref{mixed_branch_roots}, which can be implemented using at most $3$ bisection searches. Algorithm \ref{mixed_branch_roots} is then used as a subroutine in Algorithm \ref{algorithmglobalmax}.

Although Algorithm \ref{algorithmlocalmax} does not do an exhaustive search by ignoring the secondary branch $W_{-1}$, the output of Algorithm \ref{algorithmlocalmax} is proven via a sufficient second-order condition to be a strict local maximizer in $(\ref{b3})$. Furthermore, the objective values based on the output of Algorithm \ref{algorithmlocalmax} match those of Algorithm \ref{algorithmglobalmax} in almost all cases in our numerical experiments over a wide range of parameters, thus establishing that the maximizer $\overrightarrow x^\star$ usually has all components given by the principal Lambert branch. However, Algorithm \ref{algorithmglobalmax} did outperform Algorithm \ref{algorithmlocalmax} in our numerical experiments for a few instances  with a large coding rate $R$, small $\theta$ (high average SNR), and nearly equal importance weights. For example, for $K = 2$, $R = 6$, $\theta = 0.025$ and $\overrightarrow d = \left(0.51, 0.49 \right)$, the second component $x_2^*$ of the maximizer takes values in the secondary branch $W_{-1}$. All numerical results for PDS in this paper will be generated using Algorithm \ref{algorithmglobalmax}. 

\subsection{Definitions for Algorithm \ref{mixed_branch_roots}}

For any $2 \leq \ell \leq \ell_{\operatorname{PDS}}$ and $0 < \theta \leq -2W_0 \left(-\frac{1}{e} \sqrt{\frac{d_2}{2^R d_1}} \right)$, define 
\begin{align*}
  \beta_{i, \ell}(s) &\coloneqq -\frac{1}{2} \sqrt{s^2 e^{-s}\frac{d_{\ell} a_i}{d_i a_{\ell}}} \quad \, \text{ for } i < \ell,\\
  a_i &\coloneqq 2^{R(i-1)} \quad\quad\quad\quad\quad \text{ for } i \leq \ell,\\
  t_i(s) &\coloneqq -2 W_0\left(\beta_{i, \ell}(s)  \right)\,\,\, \quad \text{ for } i < \ell,\\
  F_{\ell}(s) &\coloneqq \theta\left( \sum_{i=1}^{\ell-1}  \frac{a_i}{t_i(s)} + \frac{a_\ell}{s}\right),\\
  Q_{i, \ell}(s) &\coloneqq \frac{s(2-s)}{t_i(s)(t_i(s) - 2)} \quad\,\, \text{ for } i < \ell,\\
  M_\ell(s) &\coloneqq  \sum_{i=1}^{\ell-1}  a_i Q_{i, \ell}(s) - a_\ell,\\
  s_{\ell, \max} &\coloneqq  -2W_{-1}\left( - \frac{1}{2}\sqrt{\frac{d_1 \theta^2 e^{-\theta} 2^{R(\ell-1)}}{d_\ell}} \right). 
\end{align*}
Recall from $(\ref{particularthm})$ that when $\theta > -2W_0 \left(-\frac{1}{e} \sqrt{\frac{d_2}{2^R d_1}} \right)$, the optimal solution $\overrightarrow x^\star = (1, 0, \ldots, 0)$ and there is nothing to do.

\begin{algorithm}[H]
\caption{\textsc{ComputeRootsForIndex}$(\ell,R,\theta,\overrightarrow{d})$}
\label{mixed_branch_roots}
\begin{algorithmic}[1]
\REQUIRE Positive integer $\ell \geq 2$, $R>0$, $0 < \theta \leq -2W_0 \left(-\frac{1}{e} \sqrt{\frac{d_2}{2^R d_1}} \right)$, and $\overrightarrow{d}=(d_1,\dots,d_K)$ with $d_1>\cdots>d_K>0$
\ENSURE A set $\Lambda_\ell$ of roots in the $\lambda$-domain

\IF{$M_\ell(s_{\ell,\max}) \leq 0$}
    \IF{$F_\ell(2) \geq 1 \;\land\; F_\ell(s_{\ell,\max}) \leq 1$}
        \STATE $s_0 \leftarrow \mathrm{BisectionSearch}(F_\ell(s)-1,\,2,\,s_{\ell,\max})$
        \STATE $\lambda \leftarrow \dfrac{s_0^2 e^{-s_0} d_\ell}{\theta\,2^{R(\ell-1)}}$
        \STATE \textbf{return} $\{\lambda\}$
    \ELSE
        \STATE \textbf{return} $\emptyset$
    \ENDIF
\ELSE
    \STATE $s_0 \leftarrow \mathrm{BisectionSearch}(M_\ell(s),\,2,\,s_{\ell,\max})$
    \IF{$F_\ell(s_0) > 1$}
        \STATE \textbf{return} $\emptyset$
    \ELSE
        \IF{$F_\ell(2) \geq 1 \;\land\; F_\ell(s_{\ell,\max}) < 1$}
            \STATE $s_1 \leftarrow \mathrm{BisectionSearch}(F_\ell(s)-1,\,2,\,s_0)$
            \STATE $\lambda \leftarrow \dfrac{s_1^2 e^{-s_1} d_\ell}{\theta\,2^{R(\ell-1)}}$
            \STATE \textbf{return} $\{\lambda\}$
        \ELSIF{$F_\ell(2) < 1 \;\land\; F_\ell(s_{\ell,\max}) \geq 1$}
            \STATE $s_2 \leftarrow \mathrm{BisectionSearch}(F_\ell(s)-1,\,s_0,\,s_{\ell,\max})$
            \STATE $\lambda \leftarrow \dfrac{s_2^2 e^{-s_2} d_\ell}{\theta\,2^{R(\ell-1)}}$
            \STATE \textbf{return} $\{\lambda\}$
        \ELSIF{$F_\ell(2) \geq 1 \;\land\; F_\ell(s_{\ell,\max}) \geq 1$}
            \STATE $s_1 \leftarrow \mathrm{BisectionSearch}(F_\ell(s)-1,\,2,\,s_0)$
            \STATE $s_2 \leftarrow \mathrm{BisectionSearch}(F_\ell(s)-1,\,s_0,\,s_{\ell,\max})$
            \STATE $\lambda_1 \leftarrow \dfrac{s_1^2 e^{-s_1} d_\ell}{\theta\,2^{R(\ell-1)}}$
            \STATE $\lambda_2 \leftarrow \dfrac{s_2^2 e^{-s_2} d_\ell}{\theta\,2^{R(\ell-1)}}$
            \STATE \textbf{return} $\{\lambda_1,\lambda_2\}$
        \ELSE
            \STATE \textbf{return} $\emptyset$
        \ENDIF
    \ENDIF
\ENDIF
\end{algorithmic}
\end{algorithm}

\begin{algorithm}[H]
\small
\caption{Computation of a globally optimal solution $\overrightarrow{x}^\star$ in $(\ref{b3})$ }
\label{algorithmglobalmax}
\begin{algorithmic}[1]
  \REQUIRE $R > 0$, $\theta > 0$, vector $\overrightarrow{d} = (d_1, \dots, d_K)$ with $d_1 > \dots > d_K > 0, K \geq 2$.
  \ENSURE $\overrightarrow{x}$ is an optimal solution in $(\ref{b3})$

  \IF{$\theta > -2 W_0\!\left(-\dfrac{1}{e} \sqrt{\dfrac{d_2}{2^R d_1}} \right)$}
    \STATE $\overrightarrow{x} \gets (1, 0, \ldots, 0)$ \label{replacementheuristic1PDS}
    \STATE \textbf{return} $\overrightarrow{x}$ \label{replacementheuristic2PDS}
  \ENDIF

  \FORALL{$\ell \in \{2, \ldots, \ell_{\operatorname{PDS}} \}$  }
    \STATE Define, for $i = 1, \ldots, \ell$,
    \begin{align*}
      c_i(\lambda) &= \frac{\lambda \theta 2^{R(i-1)}}{d_i},\\
      x_i^-(\lambda) &= \frac{\theta}{-2 W_0\!\left(-\frac{\sqrt{c_i(\lambda)}}{2} \right)},
    \end{align*}
    and set
    \begin{align*}
      x_\ell^+(\lambda) &= \frac{\theta}{-2 W_{-1}\!\left(-\frac{\sqrt{c_\ell(\lambda)}}{2} \right)},\\ 
      H_\ell^-(\lambda) &= \sum_{i=1}^\ell 2^{R(i-1)} x_i^-(\lambda),\\
      H_\ell^+(\lambda) &= 2^{R(\ell-1)} x_\ell^+(\lambda) + \sum_{i=1}^{\ell-1} 2^{R(i-1)} x_i^-(\lambda),\\
      \lambda_{\min} &= d_1 \theta e^{-\theta},\\
      \lambda_{\max}(\ell) &= \frac{4 e^{-2} d_\ell}{\theta 2^{R(\ell-1)}}.
    \end{align*}

        \STATE Let $\mathcal{C}_\ell \gets \emptyset$.
    \IF{$H_\ell^-(\lambda_{\min}) \geq 1 \geq H_\ell^-(\lambda_{\max}(\ell))$} \label{proveit1}
      \STATE $\lambda \leftarrow \mathrm{BisectionSearch}\left(H_\ell^-(\lambda)-1,\,\lambda_{\min},\,\lambda_{\max}(\ell)\right)$
      \STATE Define $\overrightarrow{x}$ by
      \[
        \overrightarrow{x}(i) =
        \begin{cases}
          x_i^-(\lambda) & i = 1, \ldots, \ell,\\[2pt]
          0                & i = \ell + 1, \ldots, K.
        \end{cases}
      \]
      \STATE $\mathcal{C}_\ell \gets \mathcal{C}_\ell \cup \{\overrightarrow{x}\}$
    \ENDIF

    \IF{$H_\ell^-(\lambda_{\min}) \geq 1$} 
    \FOR{$\lambda \in \textsc{ComputeRootsForIndex}(\ell,R,\theta,\overrightarrow{d})$}
      \STATE Define $\overrightarrow{x}$ by
      \[
        \overrightarrow{x}(i) =
        \begin{cases}
          x_i^-(\lambda) & i = 1, \ldots, \ell - 1,\\[2pt]
          x_\ell^+(\lambda) & i = \ell,\\[2pt]
          0                  & i = \ell + 1, \ldots, K.
        \end{cases}
      \]
      \STATE $\mathcal{C}_\ell \gets \mathcal{C}_\ell \cup \{\overrightarrow{x}\}$
    \ENDFOR
    \ENDIF

    \IF{$\mathcal{C}_\ell = \emptyset$}
      \STATE $\overrightarrow{x}^{(\ell)} \gets (0, \ldots, 0)$
    \ELSE
      \STATE $\overrightarrow{x}^{(\ell)} \gets
        \displaystyle \argmax_{\overrightarrow{x} \in \mathcal{C}_\ell} G(\overrightarrow{x})$ \label{iskobhireplacekaro1b}
    \ENDIF
 
  \ENDFOR

  \STATE Set $\overrightarrow{x}^{(1)} \gets (1, 0, \ldots, 0)$.
  \STATE \textbf{return} $\displaystyle
    \argmax_{\overrightarrow{x}^{(\ell)}:\, 1 \leq \ell \leq \ell_{\operatorname{PDS}} }
    G\big(\overrightarrow{x}^{(\ell)}\big)$ \label{replacementheuristic3PDS}
\end{algorithmic}
\end{algorithm}

\begin{algorithm}[H]
\small
\caption{Computation of a strict local maximizer $\overrightarrow{x}$ in $(\ref{b3})$}
\label{algorithmlocalmax}
\begin{algorithmic}[1]
  \REQUIRE $R > 0$, $\theta > 0$, vector $\overrightarrow{d} = (d_1, \dots, d_K)$ with $d_1 > \dots > d_K > 0$, $K \geq 2$. 
  \ENSURE $\overrightarrow{x}$ is a local maximizer in $(\ref{b3})$

  \IF{$\theta > -2 W_0\!\left(-\dfrac{1}{e} \sqrt{\dfrac{d_2}{2^R d_1}} \right)$}
    \STATE $\overrightarrow{x} \gets (1, 0, \ldots, 0)$ 
    \STATE \textbf{return} $\overrightarrow{x}$ 
  \ENDIF

  \FORALL{$\ell \in \mathcal{L}_p(\theta, R, \overrightarrow{d})$ with $\ell \geq 2$}
    \STATE Define, for $i = 1, \ldots, \ell$,
    \begin{align*}
      c_i(\lambda) &= \frac{\lambda \theta 2^{R(i-1)}}{d_i},\\
      x_i^-(\lambda) &= \frac{\theta}{-2 W_0\!\left(-\frac{\sqrt{c_i(\lambda)}}{2} \right)},
    \end{align*}
    and set
    \begin{align*}
      H_\ell^-(\lambda) &= \sum_{i=1}^\ell 2^{R(i-1)} x_i^-(\lambda),\\
      \lambda_{\min} &= d_1 \theta e^{-\theta},\\
      \lambda_{\max}(\ell) &= \frac{4 e^{-2} d_\ell}{\theta 2^{R(\ell-1)}}.
    \end{align*}
    \STATE $\lambda \leftarrow \mathrm{BisectionSearch}(H_\ell^-(\lambda)-1,\,\lambda_{\min},\,\lambda_{\max}(\ell))$.  

    \STATE Define $\overrightarrow{x}^{(\ell)}$ as
    \[
      \overrightarrow{x}^{(\ell)}(i) =
      \begin{cases}
        x_i^-(\lambda) & i = 1, \ldots, \ell,\\[2pt]
        0                  & i = \ell + 1, \ldots, K.
      \end{cases}
    \]
  \ENDFOR

  \STATE Set $\overrightarrow{x}^{(1)} \gets (1, 0, \ldots, 0)$.
  \STATE \textbf{return} $\displaystyle
    \argmax_{\overrightarrow{x}^{(\ell)}:\, \ell \in \mathcal{L}_p(\theta, R, \overrightarrow{d})}
    G\big(\overrightarrow{x}^{(\ell)}\big)$ 
\end{algorithmic}
\end{algorithm}

\subsection{Numerical Evaluation of the PDS scheme}

We numerically compared the objective values produced by the globally optimal Algorithm \ref{algorithmglobalmax} and locally optimal Algorithm \ref{algorithmlocalmax} for a fixed $R$, $\theta \in (0, 1)$ and a randomly generated importance vector $\overrightarrow d$ of length $K = 8$ such that $d_1 > \cdots > d_8 > 0$ and $\sum_{i=1}^8 d_i = 1$. Across almost all of our tested parameter ranges, Algorithms \ref{algorithmglobalmax} and \ref{algorithmlocalmax} yield identical objective values. Consequently, Algorithm \ref{algorithmlocalmax} provides a computationally simpler procedure that empirically recovers a global optimizer in many cases.

An alternative to using Algorithms \ref{algorithmglobalmax} and \ref{algorithmlocalmax} to compute $(\ref{b3})$ is to use a generic constrained solver SLSQP instead (available in SciPy.optimize.minimize function in Python). We ran the SLSQP solver $K$ times with initializations $\overrightarrow{x_1}, \ldots, \overrightarrow {x_K}$ with support sizes $1, 2, \ldots, K$, similar to how Algorithms \ref{algorithmglobalmax} and \ref{algorithmlocalmax} search for the maximizer candidates for different values of $\ell$. The components of each initialization are weighted according to the importance vector $\overrightarrow d$ and suitably rescaled to satisfy the sum constraint in $(\ref{consonxx})$.  Specifically, for $i = 1, \ldots, K$,  
\begin{align}
    \overrightarrow{x_i}(j) = \begin{cases}
        \frac{1}{\sum_{n=1}^i d_n } \frac{d_j}{2^{R(j-1)}} & j = 1, \ldots, i, \\
        0 & \text{otherwise} 
    \end{cases} 
    \label{reasonablexs}
\end{align}
so that 
\begin{align*}
    \sum_{j=1}^K 2^{R(j-1)} \overrightarrow{x_i}(j) = 1. 
\end{align*}
In our numerical testing, the objective values obtained by the SLSQP solver with multiple initialization runs 
were \emph{mostly} within $\sim 10^{-7}$ of the optimal objective values produced by Algorithm \ref{algorithmglobalmax}, but were occasionally suboptimal due to numerical instability. Therefore, Algorithm \ref{algorithmglobalmax} offers a more robust method to solve $(\ref{b3})$ than a typical SLSQP solver.   

For illustration, Figure \ref{typicalAlg2} shows a graph of $\max_{\overrightarrow x \in \mathcal{S}_K} G(\overrightarrow x)$ versus $\theta$, for a fixed $R$ and $\overrightarrow d$, obtained using Algorithm \ref{algorithmglobalmax}.

\begin{figure}[H]
    \centering
\includegraphics[width=10cm]{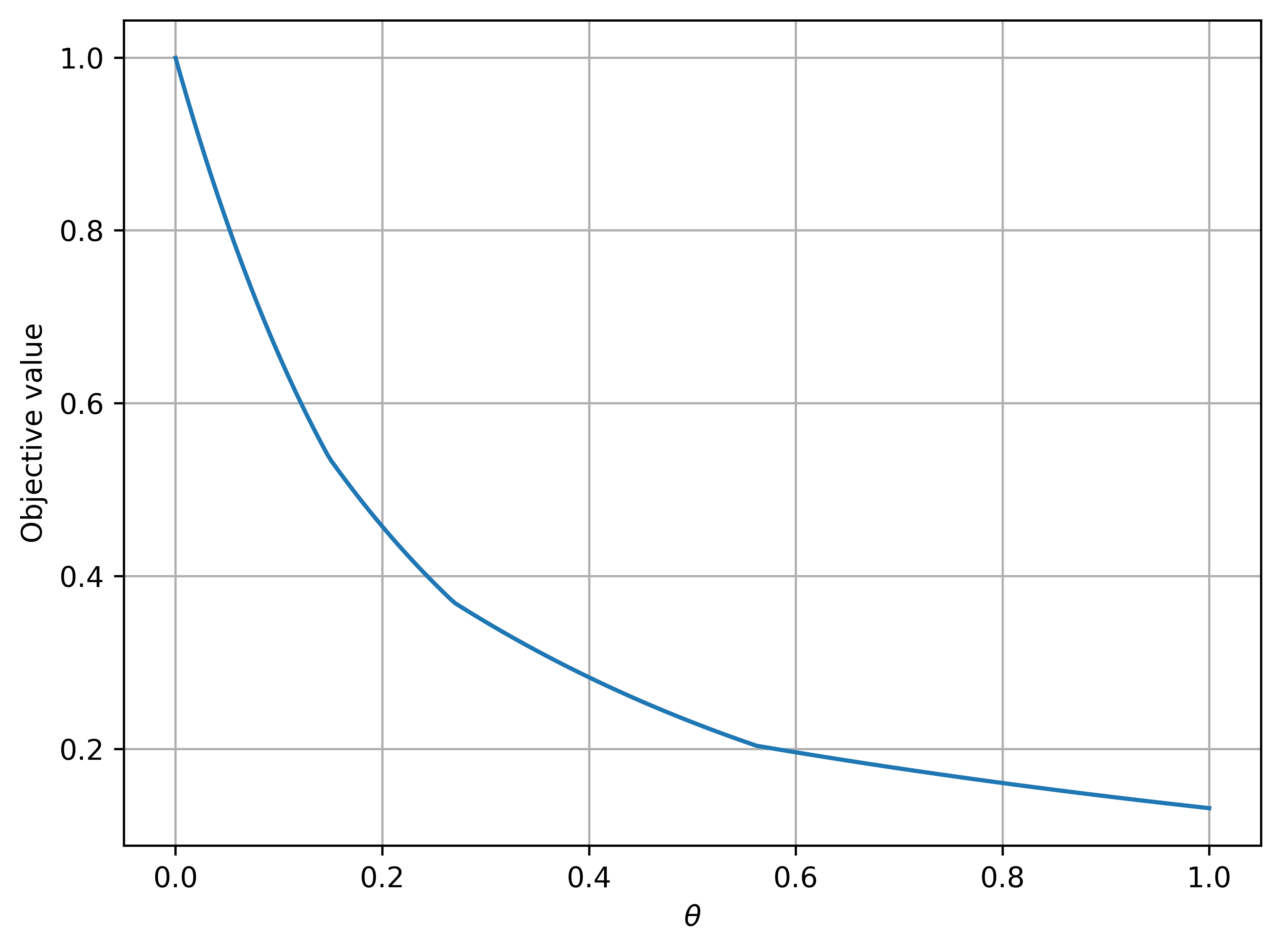}
\caption{For fixed $R = 0.1$ and $\overrightarrow d = \frac{1}{14}(5,4,3,2)$, the objective value $G(\overrightarrow x)$ is plotted against $\theta$, when $\overrightarrow x$ is taken as the output from Algorithm  $\ref{algorithmglobalmax}$. For a fixed $R$, $\theta \propto \frac{1}{P \sigma^2}$ so the $x$-axis should be interpreted as inverse average SNR up to some scaling.}
\label{typicalAlg2}
\end{figure}

Using the outputs $\overrightarrow x$ from Algorithm \ref{algorithmglobalmax} for $\overrightarrow d = \frac{1}{14}(5,4,3,2)$, $\theta \in (0, 1)$ and $R = 0.1$, we can construct the power fractions $(\alpha_1, \alpha_2, \alpha_3, \alpha_4)$ using $(\ref{xtoa})$. Figure \ref{power_plots} plots the power fractions $(\alpha_1, \alpha_2, \alpha_3, \alpha_4)$ versus $\theta$.

\begin{figure}[H]
    \centering
\includegraphics[width=10cm]{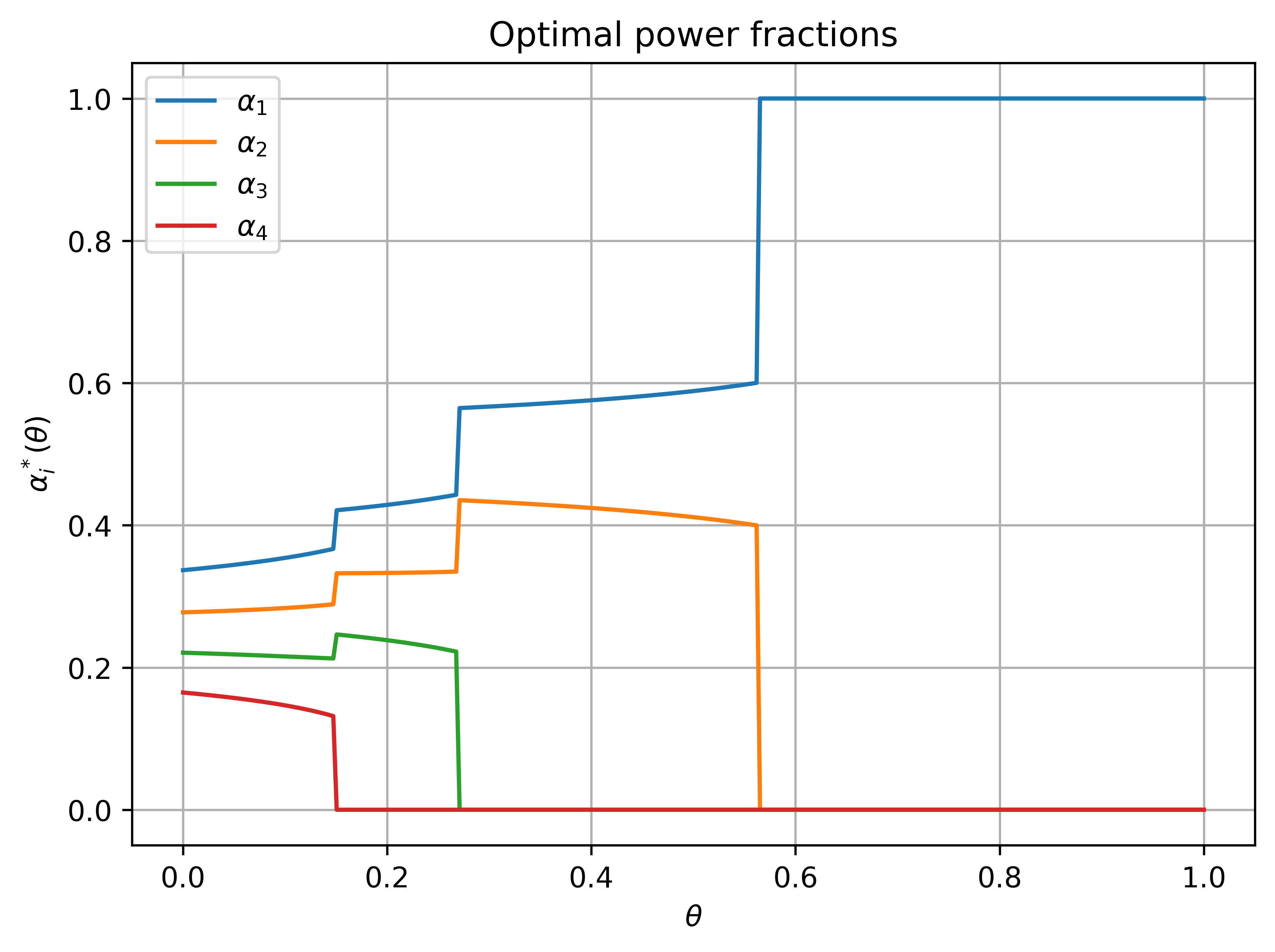}
\caption{For fixed $R = 0.1$ and $\overrightarrow d = \frac{1}{14}(5,4,3,2)$, the power fractions $\overrightarrow \alpha  = (\alpha_1, \alpha_2, \alpha_3, \alpha_4)$ are plotted versus $\theta$. The plots are based on the output of Algorithm \ref{algorithmglobalmax}.}
\label{power_plots}
\end{figure}

\section{Main Results on the ORA scheme \label{ORA_results}}

For each $1 \leq i \leq K$, let $w_i = n_i/n$. Hence, $w_1, \ldots, w_K$ are nonnegative fractions such that $\sum_{i=1}^K w_i = 1$, and $w_i$ is the fraction of $n$ channel uses used for transmitting the $i$th bit block. 
\begin{theorem}
Under the ORA scheme, we have 
    \begin{align}
        \sum_{i=1}^K \mathbb{P}(B_i^{m} = \widehat{B}_i^{m}) d_i \geq \max_{\overrightarrow w \in \Delta_{n}^{K-1}} \sum_{i=1}^K \mathbb{E}_{\gamma}\left [1 - \mathcal{E}\left(w_i n, \frac{R}{w_i}, \gamma P \right)  \right] d_i  \label{finitenORAopt}
    \end{align}
    for any $n \geq 1$, where 
    \begin{align*}
        \Delta_{n}^{K-1} &\coloneqq \left \{ \overrightarrow w \in \Delta^{K-1} : w_i n \in \mathbb{Z} \text{ for all } i \right \}
    \end{align*}
    and the expectation is w.r.t. an exponential random variable $\gamma$ with mean $\sigma^2$.    
    \label{ORAfiniteblocklengthbound}
\end{theorem}

\begin{IEEEproof}
    Conditioned on the fading state $|H|^2 = \gamma$, layer $i$ effectively experiences a complex additive white Gaussian noise channel with SNR equal to $\gamma P$. Hence, we can use the finite blocklength bound $(\ref{finitelengthupperbound})$ to write 
\begin{align*}
    \mathbb{P}\left(B_i^m \neq \widehat{B}_i^m \big | |H|^2 =  \gamma \right) &\leq \mathcal{E}\left(w_i n, \frac{R}{w_i}, \gamma P \right).
\end{align*}
Taking the expectation over $\gamma \sim \operatorname{Exp}(\sigma^2)$ establishes 
\begin{align}
    \sum_{i=1}^K \mathbb{P}(B_i^{m} = \widehat{B}_i^{m}) d_i \geq \max_{\overrightarrow w \in \Delta_{n}^{K-1}} \sum_{i=1}^K \mathbb{E}_{\gamma}\left [1 - \mathcal{E}\left(w_i n, \frac{R}{w_i}, \gamma P \right)  \right] d_i. \label{2.}
\end{align}
\end{IEEEproof}

Define 
\begin{align}
    T_n(\overrightarrow w) &\coloneqq \sum_{i=1}^K \mathbb{E}_{\gamma}\left [1 - \mathcal{E}\left(w_i n, \frac{R}{w_i}, \gamma P \right)  \right] d_i. \label{pqpq}
\end{align}

\begin{theorem}
  Let $R > 0$ be any constant independent of $n$. Then
\begin{align*}
    \lim_{n \to \infty}\,\,  \max_{\overrightarrow w \in \Delta_{n}^{K-1}}\,  T_n(\overrightarrow w) &=  \max_{\overrightarrow v \in \Delta^{K-1}}\, \sum_{i=1}^K \exp \left( -\frac{2^{R/v_i} - 1}{2^R - 1} \theta \right) d_i.  
\end{align*}
\label{orasiml}
\end{theorem}
\textit{Proof:} The proof of Theorem \ref{orasiml} given in Appendix \ref{orasiml_proof}.

\begin{remark}
    For any $\overrightarrow v \in \Delta^{K-1}$, we can construct an approximation $\overrightarrow w \in \Delta^{K-1}_{n}$ such that $|v_i - w_i| < 1/n$ for all $i = 1, \ldots, K$. We call this mapping $M_I : \overrightarrow v \mapsto \overrightarrow w$. Specifically, let $$r = n - \sum_{i=1}^K \lfloor v_i n \rfloor$$ so that $r$ is an integer between $0$ and $K-  1$. Let $f_i = v_i n - \lfloor v_i n \rfloor$ denote the fractional part of $v_i n$. Let $\mathcal{I}_r \subset \{1, \ldots, K \}$ be a set of $r$ indices with the largest values of $f_i$. Then let 
    \begin{align*}
        w_i &= \begin{cases}
            \left(\lfloor v_i n \rfloor + 1\right) / n & \text{ if } i \in \mathcal{I}_r,\\
            \lfloor v_i n \rfloor \, /n & \text{ otherwise.}
        \end{cases}
    \end{align*}
    With this construction, we have that $\sum_{i=1}^K w_i = 1$ and each $|v_i - w_i | < 1/n$. In particular, $v_i = 0$ implies $w_i = 0$.
    \label{vtowmapping}
\end{remark}

\begin{definition}
We define a first-order asymptotically optimal resource split to be a solution to the following optimization problem:
\begin{align}
    \max_{\overrightarrow v \in \Delta^{K - 1}} \sum_{i=1}^K \exp \left( -\frac{2^{R/v_i} - 1}{2^R - 1} \theta \right) d_i. \label{aa1}
\end{align}   
\label{firstorderORAsol}
\end{definition}

Similar to our analysis of the PDS scheme, we first focus on solving the first-order asymptotically optimal solution $\overrightarrow v^\star$ in $(\ref{aa1})$, which gives us the asymptotic performance of the ORA scheme by Theorem \ref{orasiml}. Define $t : [0, 1] \to [0, 1)$ as 
\begin{align}
    t(v) &\coloneqq \exp \left( -\frac{2^{R/v} - 1}{2^R - 1} \theta \right)
\end{align}
with the convention 
$$t(0) = \lim_{v \downarrow 0} t(v) =  0.$$

We rewrite $(\ref{aa1})$ as 
\begin{align}
    \max_{\overrightarrow v \in \Delta^{K - 1}} T(\overrightarrow v), \label{maxORA}
\end{align}
where we define 
\begin{align*}
    T(\overrightarrow v) &\coloneqq \sum_{i=1}^K t(v_i) d_i. 
\end{align*}
\begin{theorem}
    Let $\overrightarrow v^\star$ be an optimal solution in $(\ref{maxORA})$. Then the following hold:  
    \begin{enumerate}
        \item \label{opt_mu_1} $v_1^\star \geq v_2^\star \geq \cdots \geq v_K^\star$. 
        \item \label{opt_mu_2} $v_1^\star > 0$.
     \item \label{opt_mu_3} For each $i \in \{2, \ldots, K \}$, we have  $v_i^\star = 0$ $\iff$ $v_{j}^\star = 0$ for all $j \geq i$. 
     \item \label{opt_mu_4} Let $\ell \in \{1, \ldots, K \}$ be the largest integer such that $v_i^\star > 0$ for all $i \leq \ell$. Then $v_1^\star > \cdots > v_\ell^\star > 0$. 
    \end{enumerate}
    \label{opt_mu_properties_theorem}
\end{theorem}
\textit{Proof:} The result is analogous to Corollary \ref{optxproperties} and the proof is similar to that of Theorem \ref{opt_properties_theorem}. The proof is thus omitted.

Theorems \ref{active_layer_theoremORA}, \ref{global_maximizer_thmORA} and \ref{local_maximizer_thmORA} below are the direct counterparts of Theorems \ref{active_layer_theorem}, \ref{global_maximizer_thm} and \ref{local_maximizer_thm}, respectively, for the ORA scheme. 

\begin{theorem}
\label{active_layer_theoremORA}
Let $\overrightarrow v^\star$ denote an optimal solution in $(\ref{maxORA})$. Then 
\begin{align}
    v_i^\star \begin{cases}
        > 0 & \text{ for } i \in \{1, \ldots, \ell \}\\
        = 0 & \text{ otherwise}, 
    \end{cases} 
    \label{skfdsora}
\end{align}
where
    \begin{align}
        \ell \leq \ell_{\operatorname{ORA}} &\coloneqq 
            \max \left \{1 \leq i \leq K: d_i \geq \frac{2^R e^{-\theta}}{M_{\operatorname{int}}^*} d_1  \right \},
         \label{elllmdefora}\\
         M^*_{\operatorname{int}} &\coloneqq \max_{v \in [0, 1]} \left [ \frac{ 2^{R/v}}{v^2} \exp \left( -\frac{2^{R/v} - 1}{2^R - 1} \theta \right)\right].
    \end{align} 
    In particular, if
    \begin{align}
        \theta \geq \theta_c \coloneqq \frac{2^R - 1}{2^R} \left( \frac{2}{R \ln(2)} + 1 \right ), \label{particularthmORA}    
    \end{align}
     then $\ell = 1$ and $\overrightarrow v^\star = (1, 0, \ldots, 0)$.
\end{theorem}

Define a continuously differentiable function $\mathscr{U}: [0, 1] \to [0, \infty)$ as 
\begin{align*}
    \mathscr{U}(v) &= \frac{ 2^{R/v}}{v^2} \exp \left( -\frac{2^{R/v} - 1}{2^R - 1} \theta \right),\\
    \mathscr{U}(0) &= \mathscr{U}(0^+) = 0.
\end{align*}
If $\theta \geq \theta_c$, the optimal solution $\overrightarrow v^\star = (1, 0, \ldots, 0)$ and there is nothing to do. For $\theta < \theta_c$, $\mathscr{U}(v)$ is monotonically increasing over $[0, v_{\operatorname{int}}^*]$ and monotonically decreasing over $[v_{\operatorname{int}}^*, 1]$ for some $v_{\operatorname{int}}^* \in (0, 1)$. Hence, $\mathscr{U}(v)$ attains a maximum value of $M_{\operatorname{int}}^*$ at $v = v_{\operatorname{int}}^*$. Then for any constant $\mathscr{C} \in [2^R e^{-\theta}, M_{\operatorname{int}}^*]$, the equation 
\begin{align}
    \mathscr{U}(v) = \mathscr{C} \label{ORAequation}
\end{align}
has two solutions $v^+ \in [v_{\operatorname{int}}^* , 1]$ and $v^- \in [0, v_{\operatorname{int}}^*]$. We use the following notation to write down the two solutions: 
\begin{align}
\begin{split}
    v^+ &= V_{R, \theta}^+\left( \mathscr{C} \right),\\
    v^- &= V_{R, \theta}^-\left( \mathscr{C} \right),
\end{split}
     \label{V0V-1defs}
\end{align}
where $V_{R, \theta}^+ : [2^R e^{-\theta}, M_{\operatorname{int}}^*] \to [v_{\operatorname{int}}^* , 1]$ and $V_{R, \theta}^- : [2^R e^{-\theta}, M_{\operatorname{int}}^*] \to [0, v_{\operatorname{int}}^*]$ are easily computable functions. Indeed, for any $\mathscr{C} \in [2^R e^{-\theta}, M_{\operatorname{int}}^*]$, $V_{R, \theta}^+(\mathscr{C})$ is the solution of the equation $(\ref{ORAequation})$ obtained by a simple bisection search over the interval $[v_{\operatorname{int}}^* , 1]$, and $V_{R, \theta}^-(\mathscr{C})$ is the solution of the equation $(\ref{ORAequation})$ obtained by a simple bisection search over the interval $[0, v_{\operatorname{int}}^*]$. 

Theorems \ref{global_maximizer_thmORA} and \ref{local_maximizer_thmORA} below express the optimal solution $\overrightarrow v^\star$ in $(\ref{maxORA})$ in terms of the functions $V_{R, \theta}^+$ and $V_{R, \theta}^-$.   
     
\begin{theorem}
\label{global_maximizer_thmORA}
    An optimal solution $\overrightarrow v^\star$ in $(\ref{maxORA})$ is given by Algorithm $\ref{algorithmglobalmaxORA}$.
\end{theorem}

\begin{theorem}
\label{local_maximizer_thmORA}
A strict local maximizer $\overrightarrow v_{\operatorname{loc}}^\star$ in $(\ref{maxORA})$ is given by 
Algorithm \ref{algorithmlocalmaxORA} such that $\overrightarrow v_{\operatorname{loc}}^\star$ satisfies the KKT conditions, the optimality conditions of Theorem \ref{opt_mu_properties_theorem} and 
    \begin{align}
    \overrightarrow v_{\operatorname{loc}}^\star(i) 
    \begin{cases}
        > 0 & \text{ for } i \in \{1, \ldots, \ell \}\\
        = 0 & \text{ otherwise}, 
    \end{cases} 
    \label{gf3ORA}
\end{align}
where 
    \begin{align}
    \ell \in \mathcal{X}_p\left(\theta, R, \overrightarrow d\right) &\coloneqq \left \{1 \leq j \leq \ell_{\operatorname{ORA}}:  S_j^+(\lambda_{\operatorname{low}}) \geq 1 \geq S_{j}^+(\lambda_{\operatorname{upp}}(j))  \right \}, \label{lspecificationORA} \\
    S_{j}^+(\lambda) &\coloneqq \sum_{i=1}^j V_{R, \theta}^+\left(\mathscr{C}_i(\lambda) \right) \quad \quad  \text{for each }  j = 1, \ldots, \ell_{\operatorname{ORA}},\\
    \mathscr{C}_i(\lambda) &\coloneqq \frac{\lambda\left(2^R-1\right)}{\theta d_i R \ln (2)} \quad\, \quad \quad \quad  \text{for each }  i = 1, \ldots, \ell_{\operatorname{ORA}},\\
    \lambda_{\operatorname{low}} &\coloneqq \frac{2^R e^{-\theta} \theta d_1 R \ln(2)}{2^R - 1},\\
    \lambda_{\operatorname{upp}}(j) &\coloneqq \frac{M_{\operatorname{int}}^*\, \theta d_j R \ln (2)}{2^R - 1} \quad \,\,\,\,\,\,  \text{for each }  j = 1, \ldots, \ell_{\operatorname{ORA}},
\end{align}
and $\ell_{\operatorname{ORA}}$ is defined in $(\ref{elllmdefora})$.
\end{theorem} 

The combined proof of Theorems $\ref{active_layer_theoremORA}, \ref{global_maximizer_thmORA}$ and $\ref{local_maximizer_thmORA}$ is given in Section \ref{combined_proofORA}.

\subsection{Summary of Algorithms \ref{algorithmglobalmaxORA} and \ref{algorithmlocalmaxORA}}

Algorithm \ref{algorithmglobalmaxORA} is based on the fact\footnote{proven in Section \ref{combined_proofORA}} that an optimal $\overrightarrow v^\star$ in $(\ref{maxORA})$ is either equal to $(1, 0, \ldots, 0)$ or satisfies  
\begin{align}
    v_i^\star = v_i^\star(\lambda) = \begin{cases}
        V_{R, \theta}^+\left(\mathscr{C}_i(\lambda)  \right) & \text{ for } i \in \{1, \ldots, \ell - 1 \}\\
        V_{R, \theta}^+\left(\mathscr{C}_\ell(\lambda)  \right) \text{ or } V_{R, \theta}^-\left(\mathscr{C}_\ell(\lambda)  \right) & \text{ for } i = \ell \\
        0 & \text{ for } i \in \{\ell + 1, \ldots, K \}
    \end{cases}
    \label{maxi_candsORA}
\end{align}
for some $2\leq \ell \leq \ell_{\operatorname{ORA}}$ and $\lambda \in [\lambda_{\operatorname{low}}, \lambda_{\operatorname{upp}}(\ell)]$ such that 
\begin{align}
    \sum_{i=1}^\ell v_i^\star(\lambda) = 1. \label{Hconsgrkora}
\end{align}
Algorithm \ref{algorithmglobalmaxORA} first checks a sufficient (but not necessary) condition for $\overrightarrow v^\star = (1, 0, \ldots, 0)$. If the condition is not satisfied, then  Algorithm \ref{algorithmglobalmaxORA} searches across the values $\ell \in \{1, \ldots, \ell_{\operatorname{ORA}} \}$ and, for each $\ell$, constructs the maximizer candidates of the form $(\ref{maxi_candsORA})$ by solving for $\lambda \in [\lambda_{\operatorname{low}}, \lambda_{\operatorname{upp}}(\ell)]$ satisfying $(\ref{Hconsgrkora})$. For $\ell = 1$, the only maximizer candidate is $(1, 0, \ldots, 0)$. Furthermore, the constraints $\ell \leq \ell_{\operatorname{ORA}}$ and $\lambda \in [\lambda_{\operatorname{low}}, \lambda_{\operatorname{upp}}(\ell)]$ guarantee that 
\begin{align*}
    2^R e^{-\theta} \leq \mathscr{C}_i(\lambda) \leq M_{\operatorname{int}}^*
\end{align*}
for all $1 \leq i \leq \ell$ so that $V_{R, \theta}^+$ and $V_{R, \theta}^-$ are well-defined according to $(\ref{V0V-1defs})$. Algorithm \ref{algorithmlocalmaxORA} uses the same logic but only considers $V_{R, \theta}^+$ for $i = \ell$. In Algorithm \ref{algorithmlocalmaxORA}, for each $\ell$, the solution $\lambda$ to the equation $(\ref{Hconsgrkora})$ is unique if it exists, and a necessary and sufficient condition for existence is easily specified, making Algorithm \ref{algorithmlocalmaxORA} simpler to implement. But when $v_\ell^\star$ is chosen from the negative branch $V_{R, \theta}^-\left(\mathscr{C}_\ell(\lambda)  \right)$, then finding the values of $\lambda$ satisfying $(\ref{Hconsgrkora})$ is not as simple. Nevertheless, we reduce the solution set of  
\begin{align}
    V_{R, \theta}^-\left(\mathscr{C}_\ell(\lambda)  \right) + \sum_{i=1}^{\ell - 1} V_{R, \theta}^+\left(\mathscr{C}_i(\lambda)  \right)  = 1, \label{fds4}
\end{align}
to finding the roots of a one-dimensional real-analytic function on a compact interval, thereby showing that at most finitely many solutions exist for $(\ref{fds4})$ for $\lambda \in [\lambda_{\operatorname{low}}, \lambda_{\operatorname{upp}}(\ell)]$. We also derive a simple second-order necessary condition to filter out the solutions that do not lead to a local maximizer. This filtered solution set is given in $(\ref{mixedrootsORAlambda})$, and the proof is given in Section \ref{combined_proofORA}.   

Although Algorithm \ref{algorithmlocalmaxORA} does not do an exhaustive search by ignoring the  possible solutions given by $V_{R, \theta}^-$, the output of Algorithm \ref{algorithmlocalmaxORA} is proven via a sufficient second-order condition to be a strict local maximizer in $(\ref{maxORA})$. All numerical results for ORA in this paper will be generated using the locally optimal Algorithm \ref{algorithmlocalmaxORA}. This choice is motivated partly by simplicity, and partly because our numerical experiments show that the analogous locally optimal Algorithm \ref{algorithmlocalmax} for PDS (which similarly ignores the secondary branch $W_{-1}$) yields results that match the globally optimal Algorithm \ref{algorithmglobalmax} in almost all cases. We thus conjecture that the maximizer $\overrightarrow v^\star$ usually has all components given by the positive branch $V_{R, \theta}^+$. 

On the other hand, the optimal objective values produced by an SLSQP solver run $K$ times with initializations $\overrightarrow{v_1}, \ldots, \overrightarrow {v_K}$, where 
\begin{align}
    \overrightarrow{v_i}(j) = \begin{cases}
        \frac{d_j}{\sum_{n=1}^i d_n } & j = 1, \ldots, i, \\
        0 & \text{otherwise}, 
    \end{cases} \label{reasonablevis}
\end{align}
were \emph{mostly} within $\sim 10^{-7}$ of the optimal objective values produced by Algorithm \ref{algorithmlocalmaxORA} in our numerical testing, but were occasionally suboptimal. Therefore, Algorithm \ref{algorithmlocalmaxORA} offers a more robust method to solve $(\ref{maxORA})$ than a typical SLSQP solver. We omit the analogous plots to Figures \ref{typicalAlg2} and \ref{power_plots} for the ORA scheme, since they are qualitatively similar. Instead, in Section \ref{finalcompnumerical}, we compare the achievable first-order asymptotic and finite blocklength performance of the PDS and ORA schemes.

\subsection{Definitions for Algorithm \ref{algorithmglobalmaxORA} \label{ORAmixedroots}}

For any $2 \leq \ell \leq \ell_{\operatorname{ORA}}$ and $0 < \theta < \theta_c$, define 
\begin{align*}
    F_\ell(t) &\coloneqq t + \sum_{i=1}^{\ell - 1} V_{R, \theta}^+\left(\mathscr{U}(t) \frac{d_\ell}{d_i} \right),\\
    F_\ell'(t) &\coloneqq 1 + d_\ell \mathscr{U}'(t) \sum_{i=1}^{\ell - 1} \frac{1}{d_i\mathscr{U}'\left( V_{R, \theta}^+\left(\mathscr{U}(t) \frac{d_\ell}{d_i} \right) \right)}, 
\end{align*}
for $t \in \left [ t_{\ell, \operatorname{low}}, v_{\operatorname{int}}^* \right ]$, where 
\begin{align*}
    t_{\ell, \operatorname{low}} \coloneqq  V_{R, \theta}^-\left(\frac{2^R e^{-\theta} d_1 }{ d_\ell }   \right).
\end{align*}
Then define 
\begin{align}
    \Lambda_{\ell}^{\operatorname{ORA}} \coloneqq \left \{ \mathscr{U}(t) \frac{d_\ell \theta  R \ln(2)}{2^R - 1}: F_\ell(t) = 1, F_\ell'(t) \geq 0, t \in \left [ t_{\ell, \operatorname{low}}, v_{\operatorname{int}}^* \right ]  \right \}. \label{mixedrootsORAlambda}
\end{align}
In Section \ref{combined_proofORA}, we show that the set $\Lambda_{\ell}^{\operatorname{ORA}}$ is finite, and that $F_\ell(t)$ has a real analytic extension to an open set containing $\left [ t_{\ell, \operatorname{low}}, v_{\operatorname{int}}^* \right ]$.

\begin{algorithm}[H]
\small
\caption{Computation of a globally optimal solution $\overrightarrow{v}^\star$ in $(\ref{maxORA})$ }
\label{algorithmglobalmaxORA}
\begin{algorithmic}[1]
  \REQUIRE $R > 0$, $\theta > 0$, vector $\overrightarrow{d} = (d_1, \dots, d_K)$ with $d_1 > \dots > d_K > 0, K \geq 2$.
  \ENSURE $\overrightarrow{v}$ is an optimal solution in $(\ref{maxORA})$

  \IF{$\theta \geq \theta_c$}
    \STATE $\overrightarrow{v} \gets (1, 0, \ldots, 0)$
    \STATE \textbf{return} $\overrightarrow{v}$
  \ENDIF

  \FORALL{$\ell \in \{2, \ldots, \ell_{\operatorname{ORA}} \}$  }
    \STATE Define, for $i = 1, \ldots, \ell$,
    \begin{align*}
      \mathscr{C}_i(\lambda) &= \frac{\lambda\left(2^R-1\right)}{\theta d_i R \ln (2)},
    \end{align*}
    and set
    \begin{align*} 
      S_\ell^+(\lambda) &= \sum_{i=1}^\ell  V_{R, \theta}^+\left(\mathscr{C}_i(\lambda)  \right),\\
      S_\ell^-(\lambda) &= V_{R, \theta}^-\left( \mathscr{C}_\ell(\lambda) \right) + \sum_{i=1}^{\ell-1} V_{R, \theta}^+\left(\mathscr{C}_i(\lambda)  \right),\\
      \lambda_{\operatorname{low}} &= \frac{2^R e^{-\theta} \theta d_1 R \ln(2)}{2^R - 1},\\
      \lambda_{\operatorname{upp}}(\ell) &= \frac{M_{\operatorname{int}}^*\, \theta d_\ell R \ln (2)}{2^R - 1}.
    \end{align*}

        \STATE Let $\mathcal{C}_\ell \gets \emptyset$.
    \IF{$S_\ell^+(\lambda_{\operatorname{low}}) \geq 1 \geq S_\ell^+(\lambda_{\operatorname{upp}}(\ell))$} \label{proveit1ORA}
      \STATE $\lambda \leftarrow \mathrm{BisectionSearch}\left(S_\ell^+(\lambda)-1,\,\lambda_{\operatorname{low}},\,\lambda_{\operatorname{upp}}(\ell)\right)$
      \STATE Define $\overrightarrow{v}$ by
      \[
        \overrightarrow{v}(i) =
        \begin{cases}
          V_{R, \theta}^+\left(\mathscr{C}_i(\lambda)  \right) & i = 1, \ldots, \ell,\\[2pt]
          0                & i = \ell + 1, \ldots, K.
        \end{cases}
      \]
      \STATE $\mathcal{C}_\ell \gets \mathcal{C}_\ell \cup \{\overrightarrow{v}\}$
    \ENDIF

    \IF{$S_\ell^+(\lambda_{\operatorname{low}}) \geq 1$} 
    \FOR{$\lambda \in \Lambda_\ell^{\operatorname{ORA}}$}
      \STATE Define $\overrightarrow{v}$ by
      \[
        \overrightarrow{v}(i) =
        \begin{cases}
          V_{R, \theta}^+\left(\mathscr{C}_i(\lambda)  \right) & i = 1, \ldots, \ell - 1,\\[2pt]
           V_{R, \theta}^-\left(\mathscr{C}_\ell(\lambda)  \right) & i = \ell,\\[2pt]
          0                  & i = \ell + 1, \ldots, K.
        \end{cases}
      \]
      \STATE $\mathcal{C}_\ell \gets \mathcal{C}_\ell \cup \{\overrightarrow{v}\}$
    \ENDFOR
    \ENDIF

    \IF{$\mathcal{C}_\ell = \emptyset$}
      \STATE $\overrightarrow{v}^{(\ell)} \gets (0, \ldots, 0)$
    \ELSE
      \STATE $\overrightarrow{v}^{(\ell)} \gets
        \displaystyle \argmax_{\overrightarrow{v} \in \mathcal{C}_\ell} T(\overrightarrow{v})$
    \ENDIF
 
  \ENDFOR

  \STATE Set $\overrightarrow{v}^{(1)} \gets (1, 0, \ldots, 0)$.
  \STATE \textbf{return} $\displaystyle
    \argmax_{\overrightarrow{v}^{(\ell)}:\, 1 \leq \ell \leq \ell_{\operatorname{ORA}} }
    T\big(\overrightarrow{v}^{(\ell)}\big)$
\end{algorithmic}
\end{algorithm}

\begin{algorithm}[H]
\small
\caption{Computation of a strict local maximizer $\overrightarrow{v}$ in $(\ref{maxORA})$}
\label{algorithmlocalmaxORA}
\begin{algorithmic}[1]
  \REQUIRE $R > 0$, $\theta > 0$, vector $\overrightarrow{d} = (d_1, \dots, d_K)$ with $d_1 > \dots > d_K > 0$, $K \geq 2$. 
  \ENSURE $\overrightarrow{v}$ is a local maximizer in $(\ref{maxORA})$

  \IF{$\theta \geq \theta_c$}
    \STATE $\overrightarrow{v} \gets (1, 0, \ldots, 0)$ \label{replacementheuristic1ORA}
    \STATE \textbf{return} $\overrightarrow{v}$ \label{replacementheuristic2ORA}
  \ENDIF

  \FORALL{$\ell \in \mathcal{X}_p(\theta, R, \overrightarrow{d})$ with $\ell \geq 2$}
    \STATE Define, for $i = 1, \ldots, \ell$,
    \begin{align*}
      \mathscr{C}_i(\lambda) &= \frac{\lambda\left(2^R-1\right)}{\theta d_i R \ln (2)},
    \end{align*}
    and set
    \begin{align*}
      S_\ell^+(\lambda) &= \sum_{i=1}^\ell  V_{R, \theta}^+(\mathscr{C}_i(\lambda)),\\
      \lambda_{\operatorname{low}} &= \frac{2^R e^{-\theta} \theta d_1 R \ln(2)}{2^R - 1},\\
      \lambda_{\operatorname{upp}}(\ell) &= \frac{M_{\operatorname{int}}^*\, \theta d_\ell R \ln (2)}{2^R - 1}.
    \end{align*}

    \STATE $\lambda \leftarrow \mathrm{BisectionSearch}\left(S_\ell^+(\lambda)-1,\,\lambda_{\operatorname{low}},\,\lambda_{\operatorname{upp}}(\ell)\right)$  

    \STATE Define $\overrightarrow{v}^{(\ell)}$ as
    \[
      \overrightarrow{v}^{(\ell)}(i) =
      \begin{cases}
        V_{R, \theta}^+(\mathscr{C}_i(\lambda)) & i = 1, \ldots, \ell,\\[2pt]
        0                  & i = \ell + 1, \ldots, K.
      \end{cases}
    \]
  \ENDFOR

  \STATE Set $\overrightarrow{v}^{(1)} \gets (1, 0, \ldots, 0)$.
  \STATE \textbf{return} $\displaystyle
    \argmax_{\overrightarrow{v}^{(\ell)}:\, \ell \in \mathcal{X}_p(\theta, R, \overrightarrow{d})}
    T\big(\overrightarrow{v}^{(\ell)}\big)$ \label{replacementheuristic3ORA}
\end{algorithmic}
\end{algorithm}

\section{Numerical Comparison between PDS and ORA \label{finalcompnumerical}}

We define the following:
\begin{itemize}
    \item $N_1 \coloneqq \max_{\overrightarrow \alpha \in \Delta^{K - 1}}\,G_n(\overrightarrow \alpha)$ is the finite blocklength achievable performance of the PDS scheme (Theorem \ref{PDSfiniteblocklengthbound} and $(\ref{tt33})$). 
    \item $N_2 \coloneqq \max_{\overrightarrow x \in \mathcal{S}_K} G(\overrightarrow x)$ is the asymptotically achievable performance of the PDS scheme (Theorem \ref{finitetofirstorderreduction} and Lemma \ref{atox}). 
    \item $N_3 \coloneqq\max_{\overrightarrow w \in \Delta_{n}^{K-1}}\, T_n(\overrightarrow w)$ is the finite blocklength achievable performance of the ORA scheme (Theorem \ref{ORAfiniteblocklengthbound} and $(\ref{pqpq})$).
    \item $N_4 \coloneqq \max_{\overrightarrow v \in \Delta^{K - 1}} T(\overrightarrow v)$ is the asymptotically achievable performance of the ORA scheme (Theorem \ref{orasiml} and Remark \ref{vtowmapping}).
\end{itemize}

Maximizers in $N_1$ and $N_3$ are hard to compute; numerical optimization is difficult because even a single evaluation of the functions $G_n$ and $T_n$ is numerically expensive. To approximate a maximizer in $N_1$, we consider a modified Algorithm \ref{algorithmglobalmax} with the following two modifications:
\begin{itemize}
    \item Replace lines \ref{replacementheuristic1PDS} and \ref{replacementheuristic2PDS} by 
    \begin{align*}
        \textbf{return} \quad 
    \argmax_{\overrightarrow{x}_i: 1 \leq i \leq K}
    G_n\left( M_B^{-1}(  \overrightarrow{x}_i )\right) 
    \end{align*}
    \item Replace line \ref{iskobhireplacekaro1b} by 
    \begin{align*}
        \overrightarrow{x}^{(\ell)} \gets
        \displaystyle \argmax_{\overrightarrow{x} \in \mathcal{C}_\ell} G_n( M_B^{-1}\left(  \overrightarrow{x} \right)  )
    \end{align*}
    \item Replace line \ref{replacementheuristic3PDS} by 
    \begin{align*}
        \textbf{return} \quad 
    \argmax_{ \overrightarrow x \in \left \{ \overrightarrow{x}^{(\ell)}:\, \ell \in \mathcal{L}_p(\theta, R, \overrightarrow{d}) \right \} \cup \left \{ \overrightarrow x_i: 1 \leq i \leq K \right \} }
    G_n\left(M_B^{-1}(\overrightarrow{x} )\right) 
    \end{align*}
\end{itemize}
where $M_B$ is defined in Lemma \ref{atox} and the $\overrightarrow x_i$'s are defined in $(\ref{reasonablexs})$. Let $\overrightarrow x^{\star, n}$ denote the output of this modified Algorithm \ref{algorithmglobalmax}. Define  
\begin{align}
    N_5 \coloneqq G_n\left( M_B^{-1}(\overrightarrow x^{\star, n}) \right). \label{n5defo}
\end{align}
We will use $N_5$ as an approximation for $N_1$. In plain words, we are approximating the best achievable finite blocklength performance $N_1$ by taking the maximum value of $G_n$ among all the points $\overrightarrow x^{(\ell)}$ that are candidates for the asymptotically optimal solution as well as the points $\overrightarrow x_i$ which are arguably good heuristic choices for a maximizer in $N_1$. Since these points include the asymptotically optimal solution for PDS and we have the convergence result from Theorem \ref{finitetofirstorderreduction}, $N_5$ becomes an increasingly accurate approximation for $N_1$ for large $n$. Further discussion on this approximation accuracy is given after Figure \ref{finite_blocklength_within_ORA_5000}. 

Similarly, consider a modified Algorithm \ref{algorithmlocalmaxORA} with the following two modifications:
\begin{itemize}
    \item Replace lines \ref{replacementheuristic1ORA} and \ref{replacementheuristic2ORA} by 
    \begin{align*}
        \textbf{return} \quad 
    \argmax_{\overrightarrow{v}_i: 1 \leq i \leq K}
    T_n\left( M_I(  \overrightarrow{v}_i )\right) 
    \end{align*}
    \item Replace line \ref{replacementheuristic3ORA} by 
    \begin{align}
        \textbf{return} \quad 
    \argmax_{ \overrightarrow v \in \left \{ \overrightarrow{v}^{(\ell)}:\, \ell \in \mathcal{X}_p(\theta, R, \overrightarrow{d}) \right \} \cup \left \{ \overrightarrow v_i: 1 \leq i \leq K \right \} }
    T_n\left(M_I(\overrightarrow{v} )\right) \label{oneoftheseORA}
    \end{align}
\end{itemize}
where $M_I$ is defined in Remark \ref{vtowmapping} and the $\overrightarrow v_i$'s are defined in $(\ref{reasonablevis})$. Let $\overrightarrow v^{\star, n}$ denote the output of this modified Algorithm \ref{algorithmlocalmaxORA}. Define  
\begin{align}
    N_6 \coloneqq T_n\left( M_I(\overrightarrow v^{\star, n}) \right). \label{n6defo}
\end{align}
We will use $N_6$ as an approximation for $N_3$. Since one of the points in the argmax $(\ref{oneoftheseORA})$ is the asymptotically optimal solution for ORA and we have the convergence result from Theorem \ref{orasiml}, $N_6$ becomes an increasingly accurate approximation for $N_3$ as $n$ increases.

Also recall the definition 
\begin{align*}
    \theta = \frac{2^R - 1}{P \sigma^2}.
\end{align*}
For fixed $R$ and $P$ or a fixed $R$ only, we can interpret $\theta$ as the inverse average SNR of the channel (up to some scaling).

\begin{itemize}
    \item Figures \ref{asymptotic_PDS_vs_ORA} - \ref{finite_blocklength_PDS_vs_ORA_5000} compare the performance of PDS and ORA in the asymptotic regime and the finite blocklength regime ($n = 1000$ and $n = 5000$).
    \item Figures \ref{finite_blocklength_within_PDS_1000} and \ref{finite_blocklength_within_PDS_5000} show the gap between the asymptotic and finite blocklength performance for the PDS scheme for $n = 1000$ and $n = 5000$.
    \item Figures \ref{finite_blocklength_within_ORA_1000} and \ref{finite_blocklength_within_ORA_5000} show the gap between the asymptotic and finite blocklength performance for the ORA scheme for $n = 1000$ and $n = 5000$.
    \item Figures \ref{partition_asymptotic} and \ref{partition_finite} show the performance improvements for a multi-layered transmission using ORA in the asymptotic and finite blocklength regimes ($n = 5000$). 
\end{itemize}

\begin{figure}[H]
    \centering
\includegraphics[width=10cm]{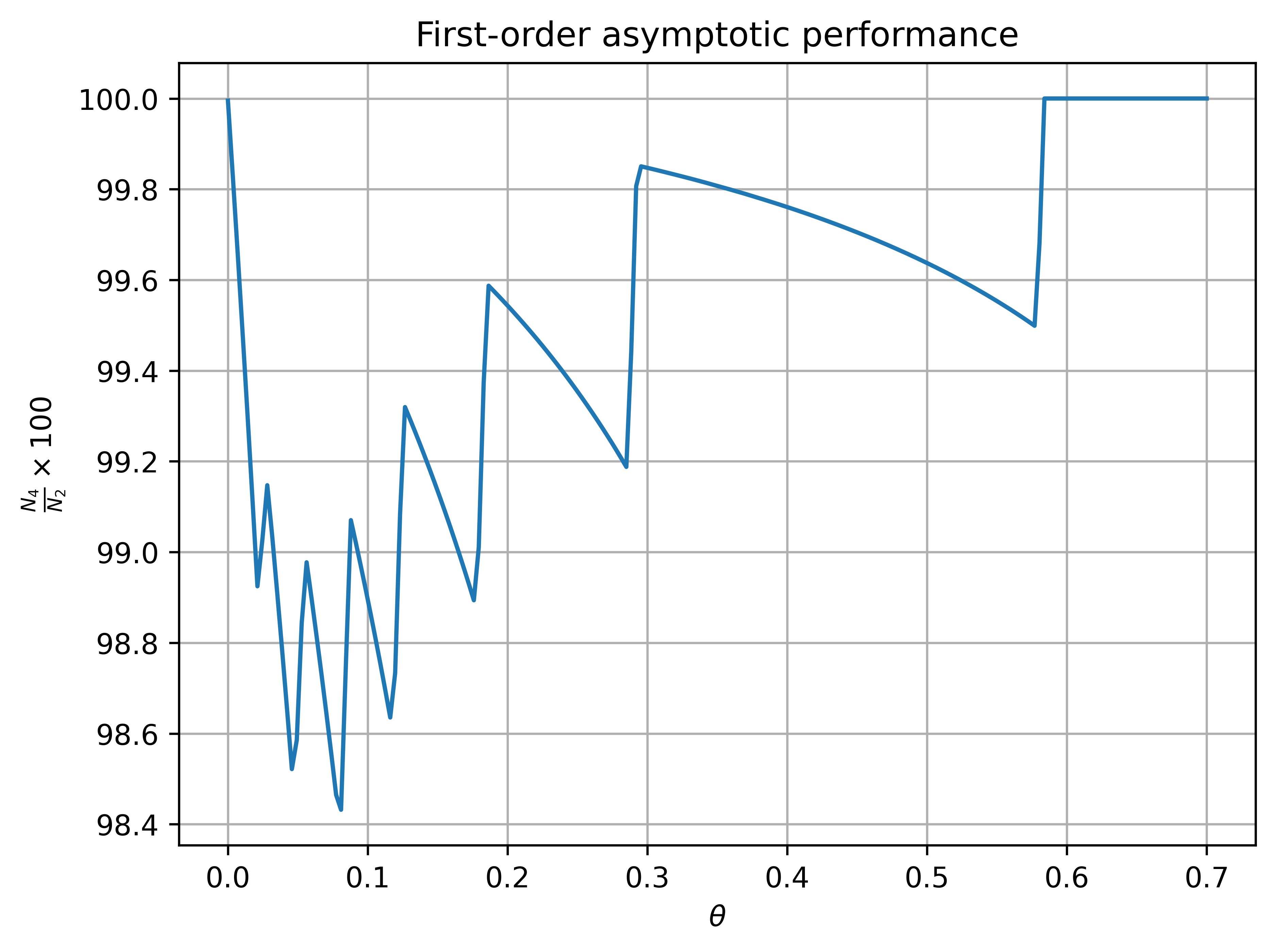}
\caption{For $R = 0.1$, importance vector $\overrightarrow d = \frac{1}{440}[100, 85, 70, 60, 50, 40, 25, 10]$, the percentage difference $ \frac{N_4}{N_2} \times 100$ is plotted against $\theta$, where $N_1$ is computed using Algorithm $\ref{algorithmglobalmax}$ and $N_4$ is computed using Algorithm \ref{algorithmlocalmaxORA}. The asymptotic performance of the ORA scheme is only slightly less than that of the PDS scheme. The troughs represent points where the PDS and ORA drop packets as the channel gets worse (see, e.g., Figure \ref{power_plots}).}
\label{asymptotic_PDS_vs_ORA}
\end{figure}

\begin{figure}[H]
    \centering
\includegraphics[width=10cm]{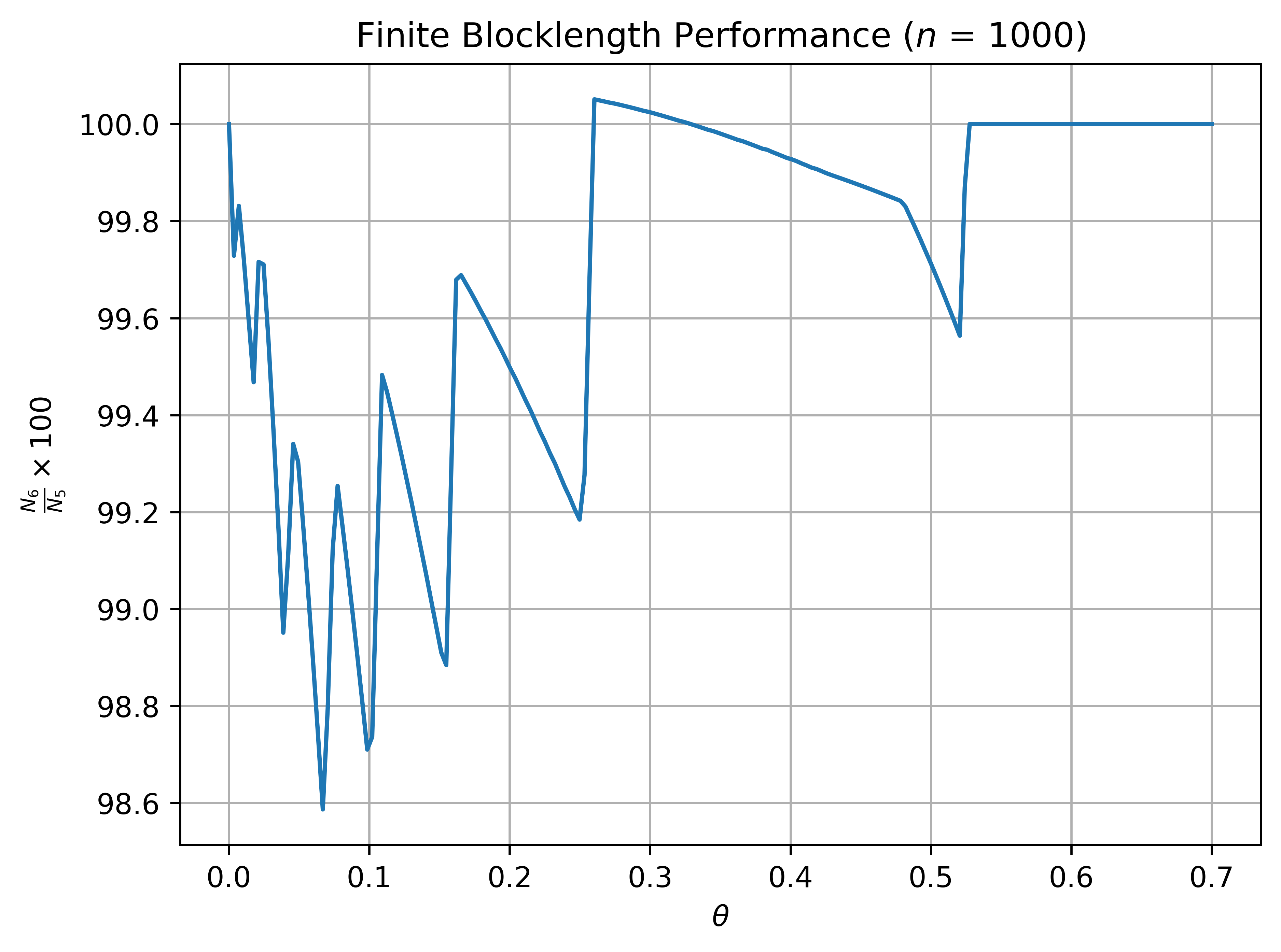}
\caption{For $n=1000$, $R = 0.1$, $P = 1$, importance vector $\overrightarrow d = \frac{1}{440}[100, 85, 70, 60, 50, 40, 25, 10]$, the percentage difference $ \frac{N_6}{N_5} \times 100$ is plotted against $\theta$. Similar to Figure \ref{asymptotic_PDS_vs_ORA}, the ORA performance even at finite blocklength is within $2\%$ of the PDS scheme.}
\label{finite_blocklength_PDS_vs_ORA_1000}
\end{figure}

\begin{figure}[H]
    \centering
\includegraphics[width=10cm]{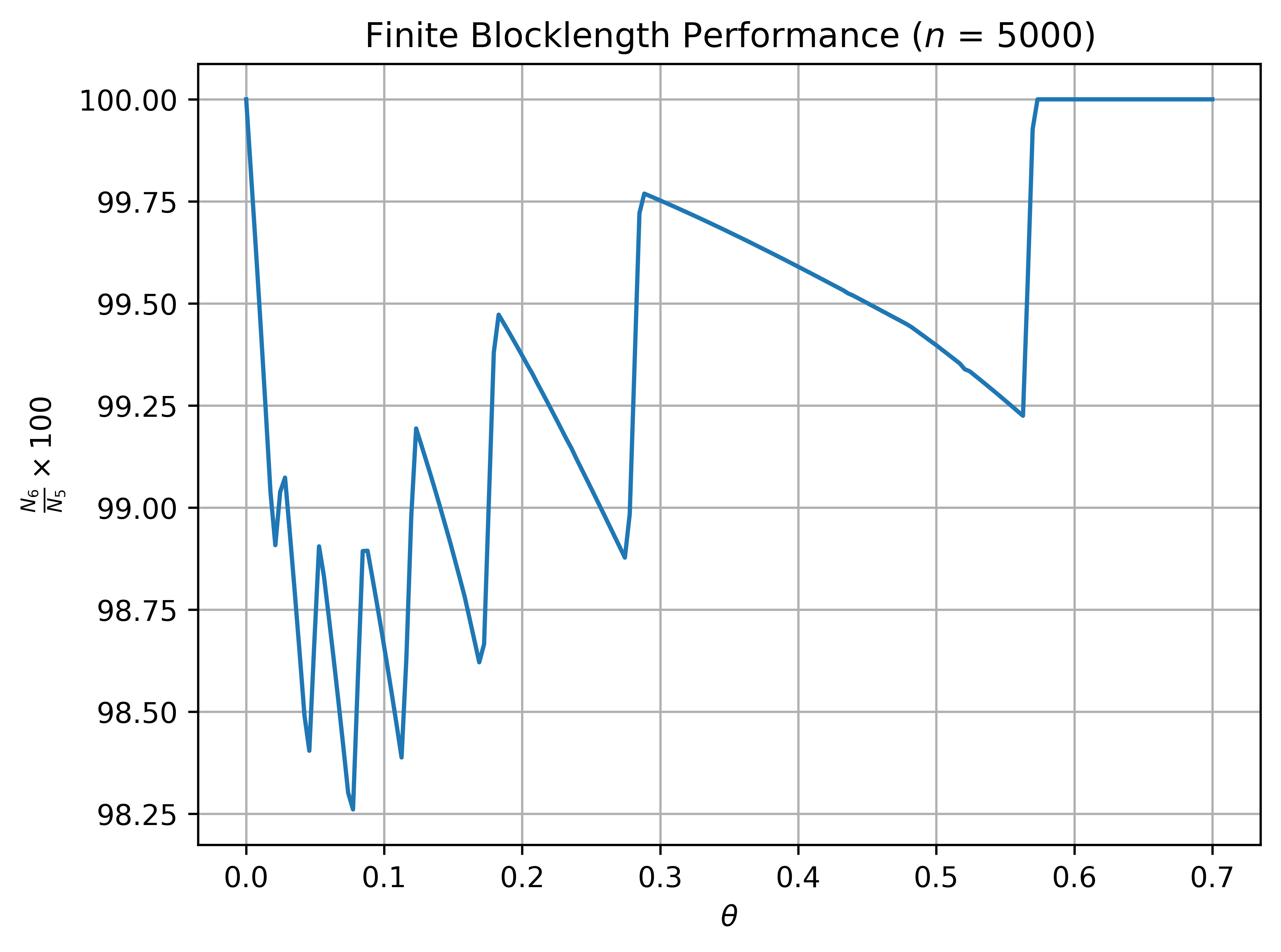}
\caption{For $n=5000$, $R = 0.1$, $P = 1$, importance vector $\overrightarrow d = \frac{1}{440}[100, 85, 70, 60, 50, 40, 25, 10]$, the percentage difference $ \frac{N_6}{N_5} \times 100$ is plotted against $\theta$.}
\label{finite_blocklength_PDS_vs_ORA_5000}
\end{figure}

A question which is not fully answered in this paper is how good $N_5$ and $N_6$ are as approximations for $N_1$ and $N_3$, respectively. Certainly, as we increase the blocklength from $n = 1000$ (Fig. \ref{finite_blocklength_PDS_vs_ORA_1000}) to $n = 5000$ (Fig. \ref{finite_blocklength_PDS_vs_ORA_5000}), the finite blocklength performance profile shows "convergence" to the asymptotic performance profile in Fig. \ref{asymptotic_PDS_vs_ORA}. To obtain additional insight, we next plot the percentage difference between the asymptotic and finite blocklength performance separately for the PDS and ORA schemes to get a sense of the speed of convergence to their respective asymptotic limits, where the said convergence was proved in Theorems \ref{finitetofirstorderreduction} and \ref{orasiml} for PDS and ORA, respectively. 

\begin{figure}[H]
    \centering
\includegraphics[width=10cm]{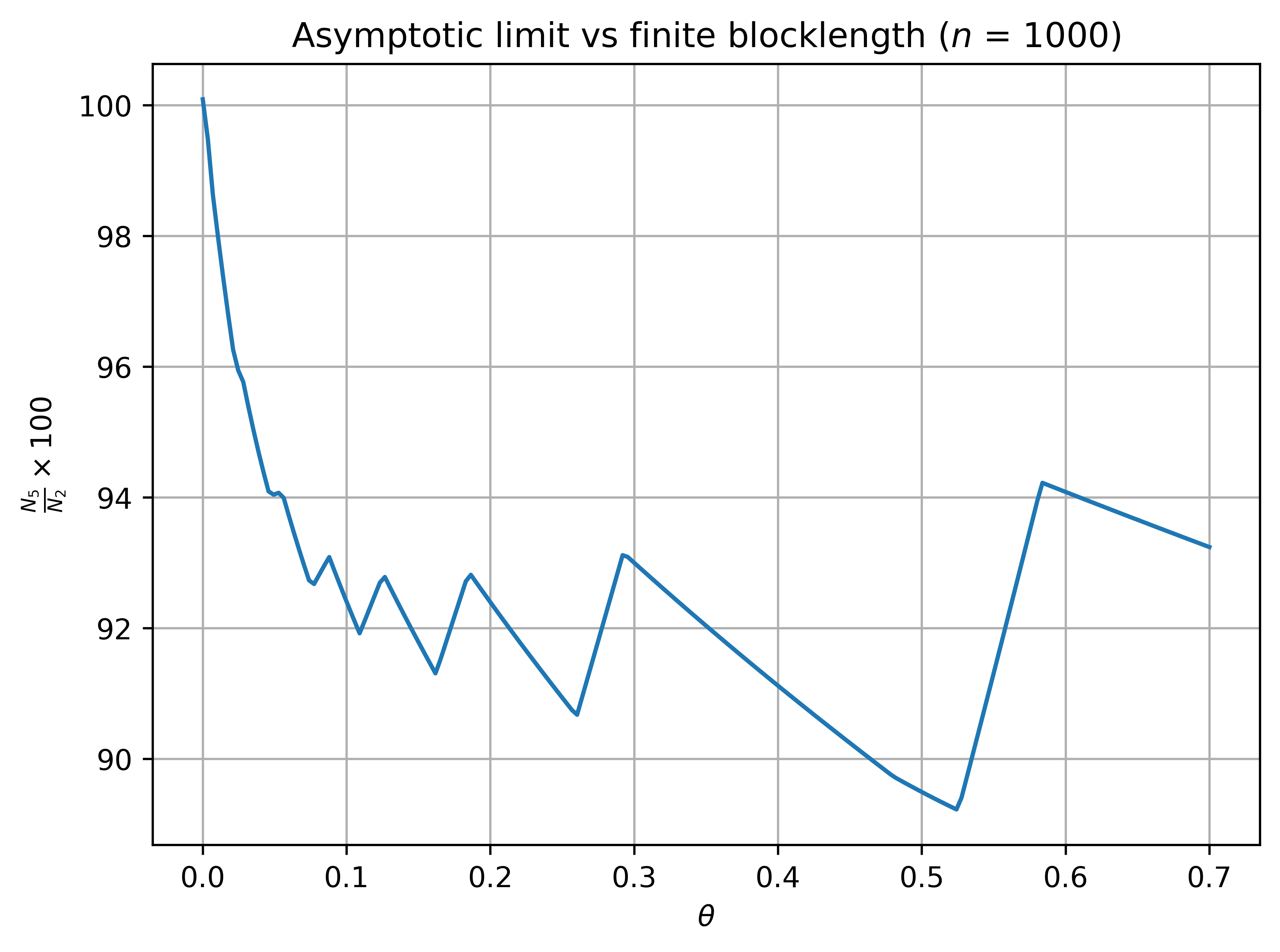}
\caption{For $n=1000$, $R = 0.1$, $P = 1$, importance vector $\overrightarrow d = \frac{1}{440}[100, 85, 70, 60, 50, 40, 25, 10]$, the percentage difference $ \frac{N_5}{N_2} \times 100$ is plotted against $\theta$ for the PDS scheme. At $n = 1000$, the finite blocklength performance is roughly within $10\%$ of the asymptotic limit and the backoff generally increases as the channel conditions get worse.}
\label{finite_blocklength_within_PDS_1000}
\end{figure}

\begin{figure}[H]
    \centering
\includegraphics[width=10cm]{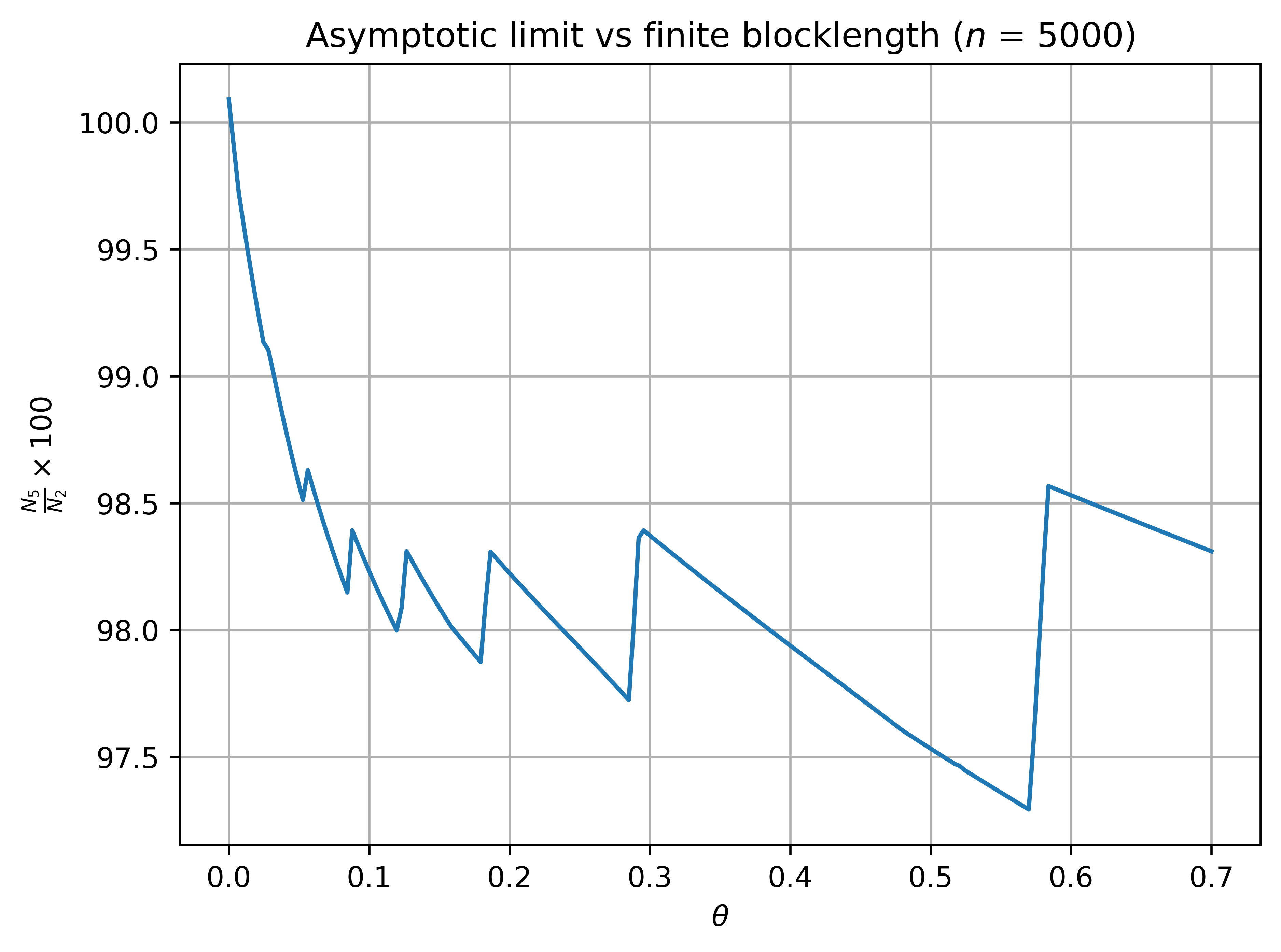}
\caption{For $n=5000$, $R = 0.1$, $P = 1$, importance vector $\overrightarrow d = \frac{1}{440}[100, 85, 70, 60, 50, 40, 25, 10]$, the percentage difference $ \frac{N_5}{N_2} \times 100$ is plotted against $\theta$ for the PDS scheme. At $n = 5000$, the finite blocklength performance is roughly within $2.5\%$ of the asymptotic limit and the backoff generally increases as the channel conditions get worse.}
\label{finite_blocklength_within_PDS_5000}
\end{figure}

\begin{figure}[H]
    \centering
\includegraphics[width=10cm]{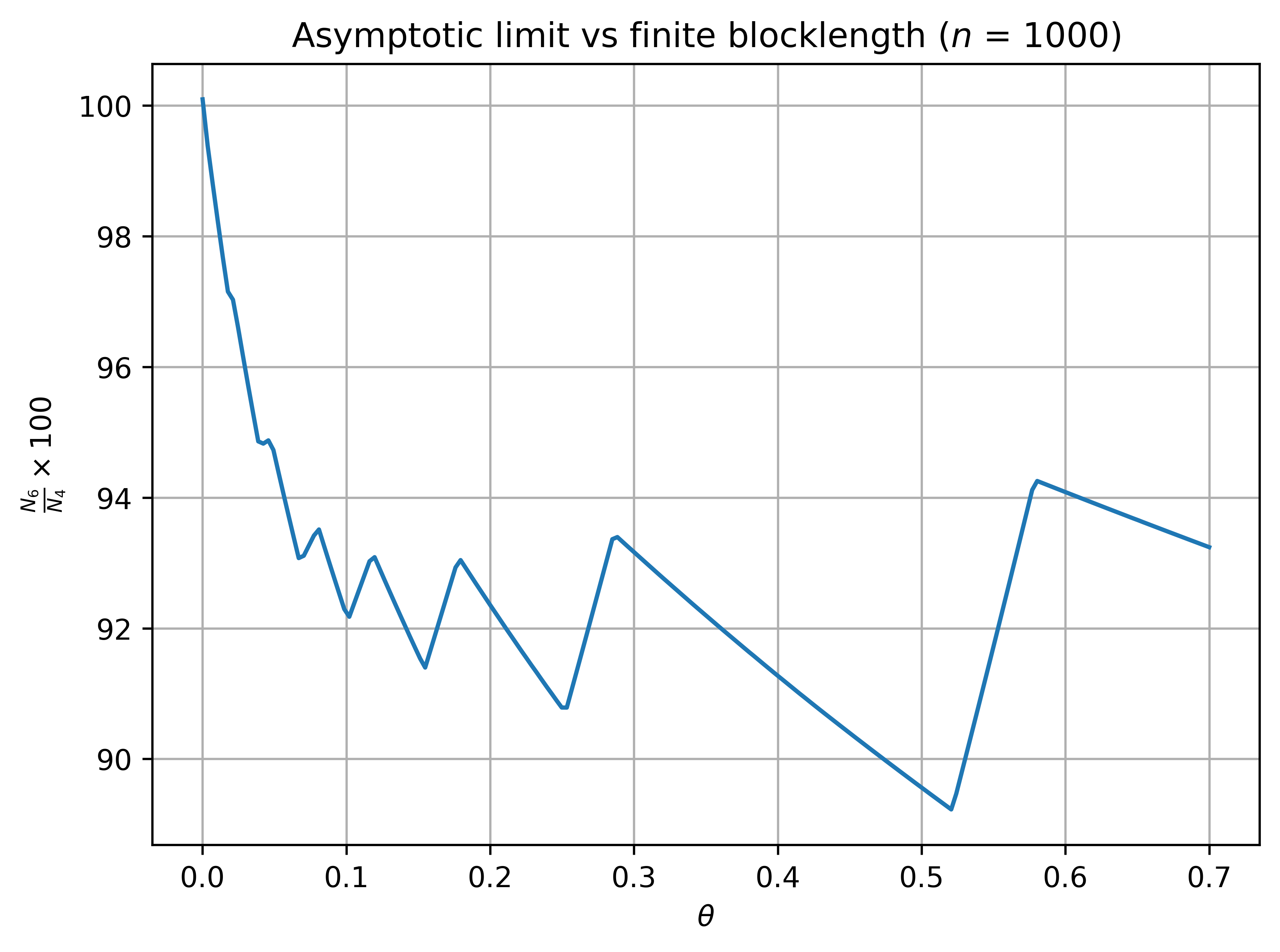}
\caption{For $n=1000$, $R = 0.1$, $P = 1$, importance vector $\overrightarrow d = \frac{1}{440}[100, 85, 70, 60, 50, 40, 25, 10]$, the percentage difference $ \frac{N_6}{N_4} \times 100$ is plotted against $\theta$ for the ORA scheme. At $n = 1000$, the finite blocklength performance is roughly within $10\%$ of the asymptotic limit and the backoff generally increases as the channel conditions get worse.}
\label{finite_blocklength_within_ORA_1000}
\end{figure}

\begin{figure}[H]
    \centering
\includegraphics[width=10cm]{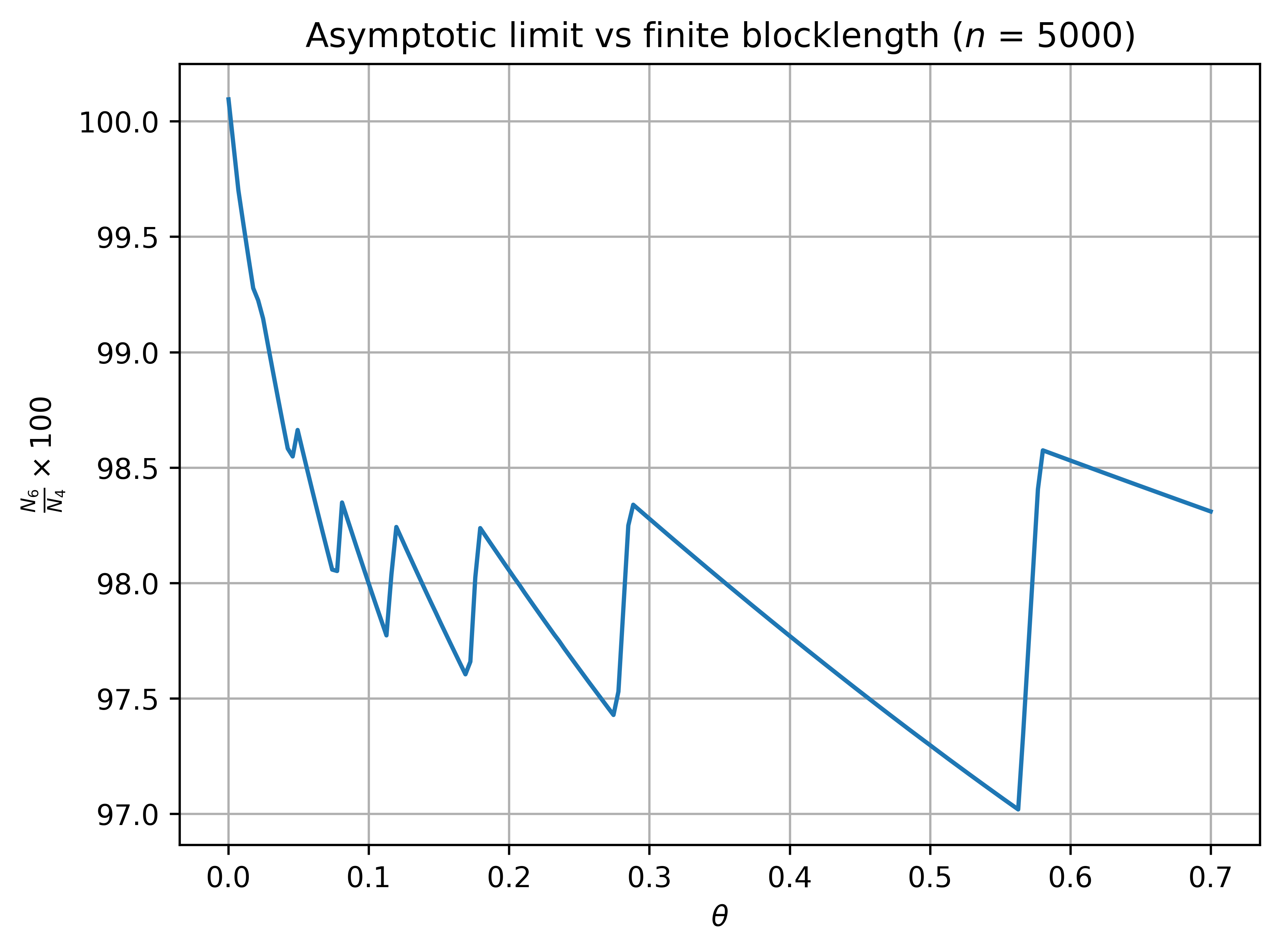}
\caption{For $n=5000$, $R = 0.1$, $P = 1$, importance vector $\overrightarrow d = \frac{1}{440}[100, 85, 70, 60, 50, 40, 25, 10]$, the percentage difference $ \frac{N_6}{N_4} \times 100$ is plotted against $\theta$ for the ORA scheme. At $n = 5000$, the finite blocklength performance is within $3\%$ of the asymptotic limit and the backoff generally increases as the channel conditions get worse.}
\label{finite_blocklength_within_ORA_5000}
\end{figure}

Note that the $10\%$ and $3\%$ backoff from the asymptotic limits in Figures \ref{finite_blocklength_within_PDS_1000} - \ref{finite_blocklength_within_ORA_5000} are upper bounds since we do not fully maximize the finite blocklength performance in $(\ref{n5defo})$ and $(\ref{n6defo})$. However, we expect these bounds to be reasonably tight by the following two heuristic arguments: 
\begin{itemize}
\item Consider the finite blocklength performance given in $(\ref{finitenPDSopt})$ or $(\ref{finitenORAopt})$. The asymptotic performance is then obtained by essentially approximating the error probability $\mathcal{E}(n, R, \rho)$ by a step function $\mathds{1}(R > C(\rho))$. The step function is a coarse approximation because the actual $\mathcal{E}(n, R, \rho)$ has a continuous transition from $0$ to $1$ as $R$ increases from $R < C(\rho)$ to $R > C(\rho)$. However, the maximization objective involves $\mathbb{E}_\gamma\left [ \mathcal{E}(n, R, \rho)\right]$, where the channel gain $\gamma$ has a continuous probability distribution. Therefore, the discontinuity of the step approximation is killed by the expectation, e.g., $\mathbb{E}_\gamma\left [ \mathds{1}(\gamma < \tau) \right] = \mathbb{P}(\gamma < \tau)$ which is a smooth function of $\tau$. Since the channel gain is unknown at the transmitter, the transmitter optimizes the power split $\overrightarrow \alpha$ or the resource split $\overrightarrow w$ with respect to $\mathbb{E}_\gamma\left [ \mathcal{E}(n, R, \rho)\right]$, where $\mathbb{E}_\gamma\left [ \mathcal{E}(n, R, \rho)\right]$ mimics a continuous transition from $0$ to $1$ characteristic of the finite blocklength regime, even if an asymptotic step function approximation is used for $\mathcal{E}(n, R, \rho)$. Therefore, the resulting asymptotic solutions $\overrightarrow \alpha^\star$ and $\overrightarrow w^\star$ would tend to perform well in the finite blocklength regime. 
    \item It is reported in \cite[p. 4233]{6802432} that communication strategies that are optimized for the (asymptotic) outage probability $\mathbb{E}_\gamma\left [ \mathds{1}(\gamma < \tau) \right] = \mathbb{P}(\gamma < \tau)$ tend to perform well at finite blocklength. This follows from the result \cite[(4)]{6802432} that for quasi-static fading channels, the maximum achievable rate $R^*(n, \epsilon)$ at blocklength $n$ and with block error probability $\epsilon$ converges much faster to the asymptotic limit than the typical $1/\sqrt{n}$ convergence rate that applies for many non-fading channels \cite{ppv}. 
\end{itemize}

In our last numerical experiment, we evaluate the performance improvement of the ORA scheme for increasing values of $K$. Note that $$N_4 = N_4\left(R, \overrightarrow d, K, P, \sigma^2 \right)$$ is a function of $R, \overrightarrow d, K, P$ and $\sigma^2$. Consider a fixed importance vector $\overrightarrow d$ of length $K$, where $K$ is even. Then for $i=K/2, K/4, \ldots, 1$, let $\overrightarrow d^{(i)}$ denote the importance vector obtained from $\overrightarrow d$ by repeatedly aggregating adjacent entries in pairs until the vector has length $i$. For example, if $\overrightarrow d = [0.5, 0.25, 0.2, 0.05]$, then $\overrightarrow d^{(4)} = \overrightarrow d$, $\overrightarrow d^{(2)} = [0.75, 0.25]$ and $\overrightarrow d^{(1)} = [1]$.

We then plot $N_4$ versus $\sigma^2$ with the remaining parameters fixed as follows: 
\begin{align*}
    &N_4\left(R, \overrightarrow d^{(16)}, 16, P, \sigma^2 \right)\\
    &N_4\left(2R, \overrightarrow d^{(8)}, 8, P, \sigma^2 \right)\\
    &N_4\left(4R, \overrightarrow d^{(4)}, 4, P, \sigma^2 \right)\\
    &N_4\left(8 R, \overrightarrow d^{(2)}, 2, P, \sigma^2 \right)\\
    &N_4\left(16 R, \overrightarrow d^{(1)}, 1, P, \sigma^2 \right), 
\end{align*}
where $R = 0.1$, $P = 1$ and $\overrightarrow d^{(16)} = \frac{1}{2560} [1000, 300, 250, 200, 150, 110, 100, 90, 80, 70, 60, 50, 40, 30, 20, 10]$. Figure \ref{partition_asymptotic} shows these five plots whereas Figure \ref{partition_finite} shows similar plots for $N_6$ for blocklength $n = 5000$. The performance improvement is more significant at poor channel conditions; recall that $\sigma^2$ is the expected channel power gain. There is also a diminishing marginal improvement as the number of blocks $K$ increases. Plots for $N_2$ and $N_5$ for the PDS scheme are similar and hence omitted.  

\begin{figure}[H]
    \centering
\includegraphics[width=10cm]{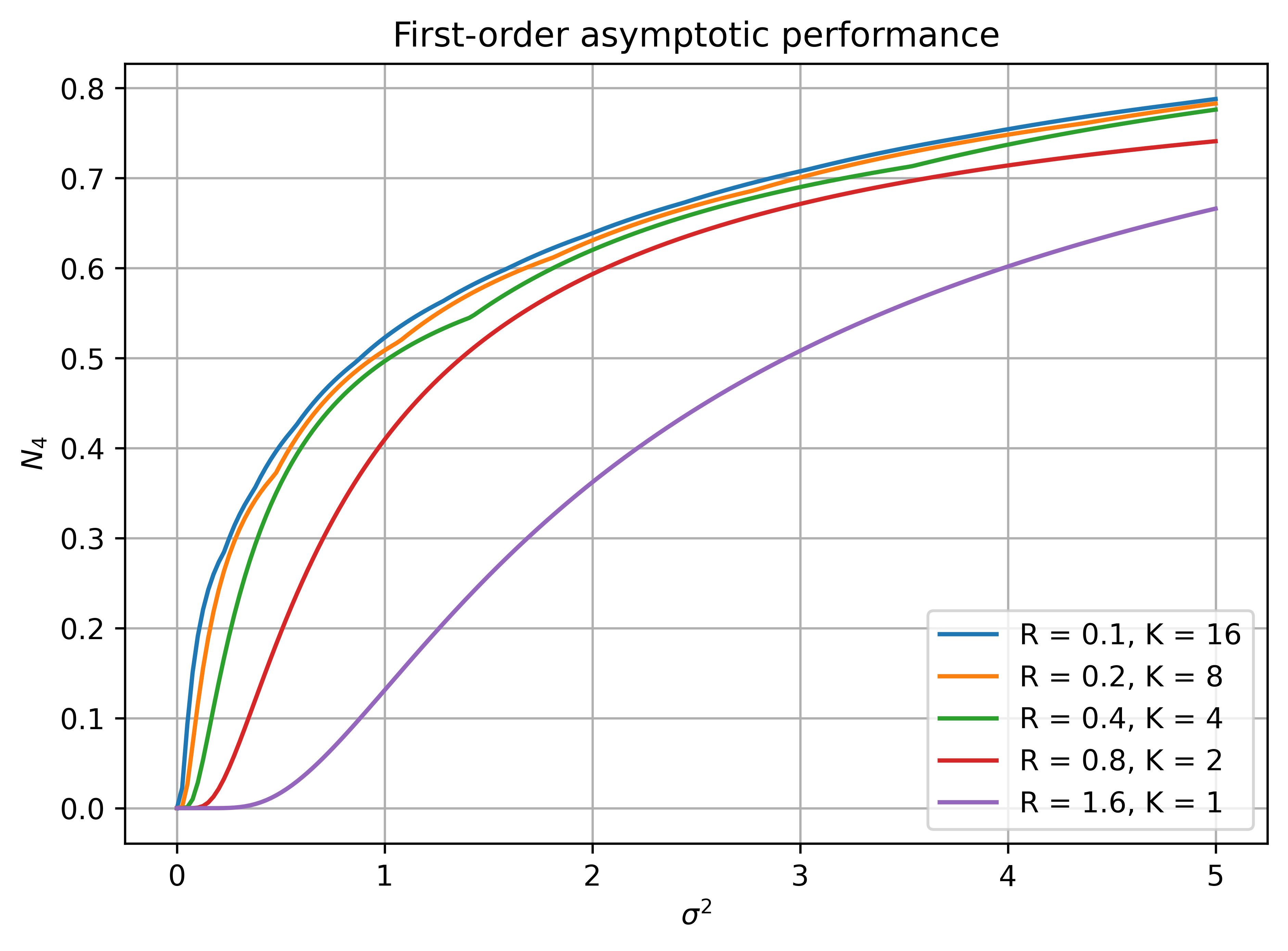}
\caption{}
\label{partition_asymptotic}
\end{figure}

\begin{figure}[H]
    \centering
\includegraphics[width=10cm]{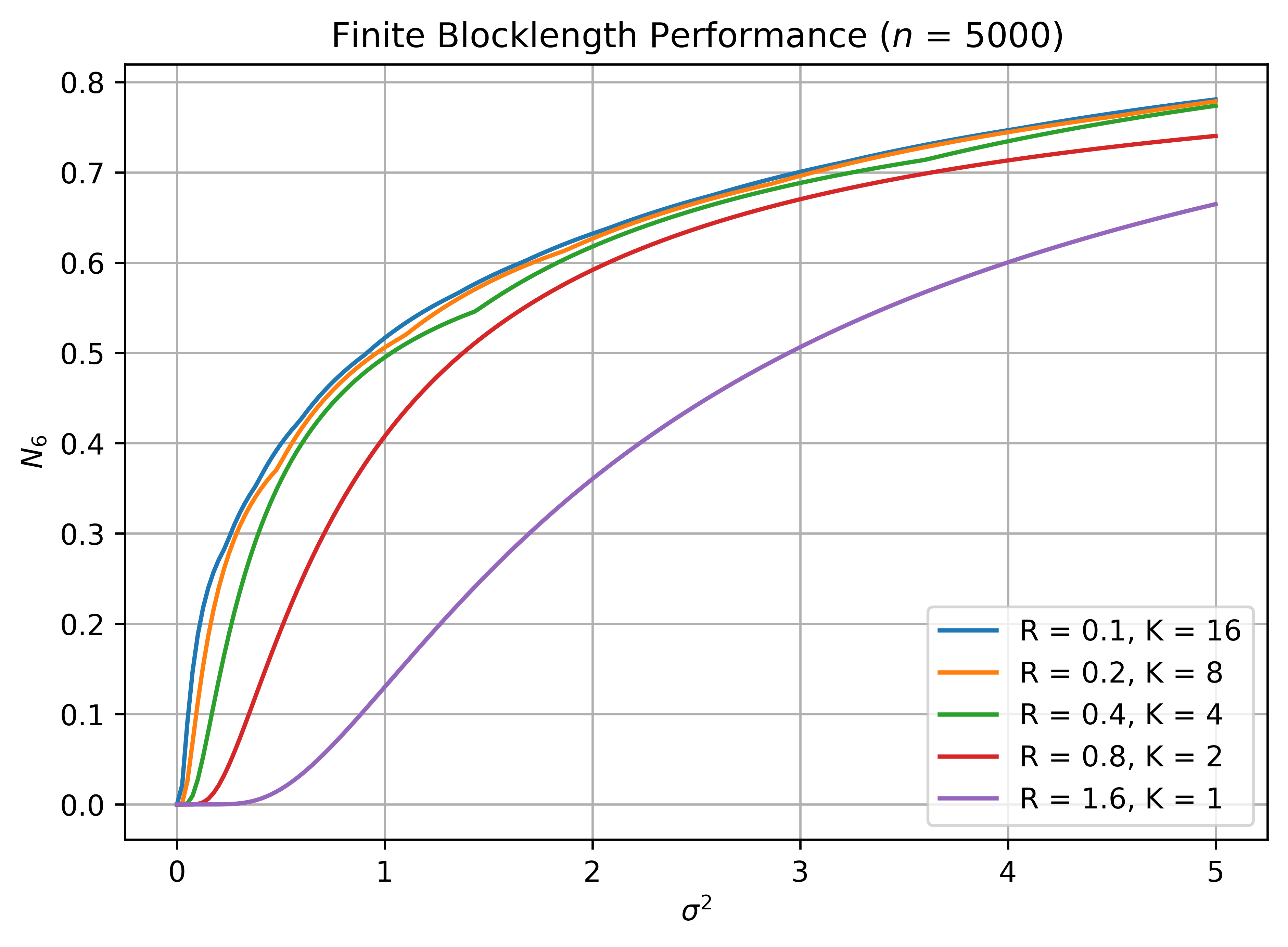}
\caption{}
\label{partition_finite}
\end{figure}

\section{Proof of Theorems \ref{active_layer_theorem}, \ref{global_maximizer_thm} and \ref{local_maximizer_thm} \label{combined_proof}}

In addition to the optimality conditions in Corollary $\ref{optxproperties}$, the KKT conditions are also necessary for an optimal $\overrightarrow x^\star$ in $(\ref{b3})$ since the Linear
Independence Constraint Qualification (LICQ) \cite[Definition 12.4]{nonlinopt} holds at any optimal $\overrightarrow x^\star$. Specifically, the gradient vectors corresponding to the constraint functions in $\mathcal{S}_K$ are given by 
\begin{align}
    \overrightarrow{e}_1, \ldots, \overrightarrow{e}_K \text{ and } \left(1, 2^R, \ldots, 2^{R(K-1)}\right)^T,  \label{grad_vects}
\end{align}
where $\overrightarrow{e}_i \in \mathbb{R}^K$ is the $i$th standard unit vector. From Corollary \ref{optxproperties}, the constraint $x_1 \geq 0$ is inactive for an optimal $\overrightarrow x^\star$. Since the set of vectors in $(\ref{grad_vects})$ excluding $\overrightarrow{e}_1$ is linearly independent, LICQ holds at any maximizer $\overrightarrow x^\star$ in $(\ref{b3})$, making the KKT conditions necessary \cite[Theorem 12.1]{nonlinopt}. A sufficient condition for a point $\overrightarrow x$ to be a strict local maximizer in $(\ref{b3})$ will be given later in the proof. The mathematical tools used in this section are described in \cite[Sections 12.3 \& 12.5]{nonlinopt}. 

We first rewrite $(\ref{b3})$ as 
\begin{align}
    &\min_{\overrightarrow x \in \mathcal{S}_K} \overline{G}(\overrightarrow x), \text{ where } \label{maxtomin}\\
    & \overline{G}(\overrightarrow x) = - \sum_{i=1}^K g(x_i) d_i.
\end{align}
Note that $\overline{G}(\overrightarrow x)$ is twice continuously differentiable over $\mathcal{S}_K$. Specifically, for $x > 0$, 
\begin{align*}
    g'(x) &= \frac{\theta}{x^2} e^{-\theta/x},\\
    g''(x) &= \frac{\theta(\theta - 2x)}{x^4} e^{-\theta/x}
\end{align*}
with 
\begin{align*}
    g'(0) &= \lim_{x \downarrow 0} g'(x) = 0,\\
    g''(0) &= \lim_{x \downarrow 0} g''(x) = 0.
\end{align*}
Also note that 
\begin{itemize}
    \item $g'(x) > 0$ for all $x > 0$, 
    \item $g''(x) > 0$ for $0 < x < \theta/2$ (convex)
    \item $g''(x) < 0$ for $x > \theta/2$ (concave)
\end{itemize}

Define the Lagrangian as 
\begin{align*}
    \mathcal{L}(\overrightarrow x, \lambda, \overrightarrow \mu ) &= -\sum_{i=1}^K g(x_i) d_i + \lambda\left(  \sum_{i=1}^K (2^R)^{i-1} x_i - 1\right) - \sum_{i=1}^K  \mu_i x_i, 
\end{align*}
where $\lambda \in \mathbb{R}$ and $\mu_i \geq 0$. The KKT conditions \cite[Theorem 12.1]{nonlinopt} and the optimality properties from Corollary \ref{optxproperties} imply the following: for any maximizer in $(\ref{b3})$, there exist Lagrange multipliers $(\lambda, \overrightarrow \mu)$ such that 
\begin{align}
    \nabla_{x_i} \mathcal{L}(\overrightarrow x, \lambda, \overrightarrow \mu ) &= -   \frac{\theta}{x_i^2} e^{-\theta/x_i} d_i + \lambda (2^R)^{i-1} - \mu_i = 0 \label{stationarity}\\
    \mu_i x_i &= 0 \label{comp_slack}\\
    \mu_i &\geq 0 \label{dual_feasibility}\\
    \sum_{i=1}^K (2^R)^{i-1} x_i &= 1 \label{eq_const}\\
    1 &\geq x_1 \geq \cdots \geq x_K \geq 0, \label{monotonicity_x}     
\end{align}
where $(\ref{stationarity}), (\ref{comp_slack})$ and $(\ref{dual_feasibility})$ hold for all $i \in\{1, \ldots, K \}$.

\begin{lemma}
For any point $(\overrightarrow x, \lambda, \overrightarrow \mu)$ satisfying $(\ref{stationarity}) - (\ref{monotonicity_x})$, we have $\lambda > 0$. 
    \label{lambda_pos_res}
\end{lemma}
\begin{IEEEproof}
Fix $j \in \{1, \ldots, K \}$ such that $x_j \in (0, 1]$. If $\lambda \leq 0$, then $(\ref{stationarity})$ cannot hold for $i = j$.
\end{IEEEproof}

\begin{lemma}
    For any minimizer $\overrightarrow x$ in $(\ref{maxtomin})$, there exists a unique $(\lambda, \overrightarrow \mu)$ such that $(\overrightarrow x, \lambda, \overrightarrow \mu)$ satisfies the conditions $(\ref{stationarity}) - (\ref{monotonicity_x})$. 
    \label{unique_mult}
\end{lemma}
\begin{IEEEproof}
As mentioned in the discussion following $(\ref{grad_vects})$, LICQ holds at any maximizer in $(\ref{b3})$ (equivalently, minimizer in $(\ref{maxtomin})$). Hence, for any given maximizer in $(\ref{b3})$,  
\begin{itemize}
    \item KKT conditions are satisfied \cite[Theorem 12.1]{nonlinopt}, and 
    \item its associated Lagrange mulitpliers are unique \cite[p. 321]{nonlinopt}. 
\end{itemize} 
\end{IEEEproof}

By defining $y_i = \theta/x_i$ so that each $y_i \geq \theta$, $(\ref{stationarity})$ can be rewritten as
\begin{align}
    \psi(y_i) &= \frac{\lambda \theta (2^R)^{i-1}}{d_i} - \frac{\theta \mu_i}{d_i}\\
    \psi(y_i) &= c_i(\lambda) - \frac{\theta \mu_i}{d_i}, \label{m08}
\end{align}
where $\psi : (0, \infty] \to [0, 4 e^{-2}]$ is defined as $\psi(y) \coloneqq y^2 e^{-y}$ and $c_i(\lambda) \coloneqq \frac{\lambda \theta (2^R)^{i-1}}{d_i}$. Key facts about $\psi(y)$:
\begin{itemize}
    \item $\psi(0^+) = 0$ and $\psi(y) \to 0$ as $y \to \infty$.
    \item $\psi(y)$ is unimodal, increasing for $y \leq 2$ and decreasing for $y \geq 2$.
    \item $\psi'(y) = y e^{-y}(2-y)$; therefore, $\psi(y)$ attains a maximum at $y = 2$ with $\psi(2) = 4e^{-2}$. 
\end{itemize}

\textbf{Case 1: $\theta \geq 2$}

In this case, $\psi(y_i) \in [0, \theta^2 e^{-\theta}]$ for all $i$. Hence, if $c_t(\lambda) > \theta^2 e^{-\theta}$ for some $t \in \{1, \ldots, K \}$, then $(\ref{m08})$ can hold only if $\mu_t > 0$ in which case $y_t = +\infty$ or $x_t = 0$ by the complementary slackness condition in $(\ref{comp_slack})$. Hence, if $c_t(\lambda) > \theta^2 e^{-\theta}$, $(\ref{m08})$ simplifies to $\mu_t = \lambda 2^{R(t-1)}$. Given $\lambda > 0$ from Lemma \ref{lambda_pos_res}, $c_i(\lambda)$ is increasing in $i$ since $d_i$ is decreasing in $i$. Hence, if $c_t(\lambda) > \theta^2 e^{-\theta}$ for some $t$, then $c_j(\lambda) > \theta^2 e^{-\theta}$ for all $j \geq t$, leading to $x_j = 0$ for all $j \geq t$. This is consistent with the structure of an optimal $\overrightarrow x$ as established in Corollary $\ref{optxproperties}.(\ref{opt_prop_3x})$. On the other hand, if $0 < c_i(\lambda) \leq \theta^2 e^{-\theta}$, then $0 \leq \mu_i \leq \lambda 2^{R(i-1)}$ in order to satisfy $(\ref{m08})$. If $\mu_i = \lambda 2^{R(i-1)}$, then $x_i = 0$. If $\mu_i < \lambda 2^{R(i-1)}$, then there is a finite solution to $(\ref{m08})$ so that $\mu_i = 0$ and $(\ref{m08})$ simplifies to 
\begin{align}
    \psi(y_i) &= c_i(\lambda).  \label{sol_ybis}
\end{align}
Using Corollary \ref{optxproperties}, an optimal $\overrightarrow x$ must satisfy 
\begin{align}
    x_i \begin{cases}
        > 0 \text{ for } i \in \{1, \ldots, \ell \}\\
        = 0 \text{ otherwise}
    \end{cases} 
    \label{bo2nt}
\end{align}
for some integer $\ell \in \{1, \ldots, K \}$. Hence, we must have 
\begin{align*}
    \ell &= \max \{1 \leq i \leq K: \mu_i = 0 \}\\
    &\leq \max \left \{1 \leq i \leq K : c_i(\lambda) \leq \theta^2 e^{-\theta} \right \}. 
\end{align*}
We have 
\begin{align}
    0 < c_1(\lambda) < \cdots  < c_\ell(\lambda) \leq \theta^2 e^{-\theta} \label{c_m2ono}
\end{align}
and, from Corollary $\ref{optxproperties}.(\ref{opt_prop_4x})$,  
\begin{align}
    1\geq x_1 > x_2 > \cdots > x_\ell > 0 \quad \iff \quad \theta \leq y_1 < \cdots < y_\ell < \infty. \label{mo2no_y}
\end{align}
Since $\psi$ is strictly decreasing over $[\theta, \infty)$ for $\theta \geq 2$, there is a unique solution $y_i^+ \in [\theta, \infty)$ of $\psi(y) = c_i(\lambda)$ for each $1 \leq i \leq \ell$. From $(\ref{c_m2ono})$ and the fact that $\psi(\cdot)$ is decreasing over $[2, \infty)$, we have 
\begin{align}
    \theta \leq y_\ell^+ < \cdots < y_1^+ < \infty. \label{order2_of_sols} 
\end{align}
From $(\ref{order2_of_sols})$, it is clear that in order to satisfy $(\ref{mo2no_y})$, we must have $\ell=1$. Hence, for $\theta \geq 2$, 
\begin{itemize}
    \item $\mu_1 = 0$ and $\mu_i = \lambda 2^{R(i-1)}$ for all $i \geq 2$,
    \item $c_1(\lambda) = \theta^2 e^{-\theta} \iff \lambda = d_1 \theta e^{-\theta}$,
    \item $x_1 = 1$ and $x_i = 0$ for all $i \geq 2$. 
\end{itemize}
From \cite[Theorem 12.6]{nonlinopt}, a sufficient condition for $\overrightarrow x = (1, 0, \ldots, 0)$ to be a minimizer in $(\ref{maxtomin})$ is $\overrightarrow w^T \nabla^2_{xx} \mathcal{L}(\overrightarrow x, \lambda, \overrightarrow \mu) \overrightarrow w > 0$ for all $\overrightarrow w \neq 0$ in the critical cone \cite[p. 330]{nonlinopt} for this $\overrightarrow x$. But it can be checked that the critical cone at this point is just the zero vector so $\overrightarrow x = (1, 0, \ldots, 0)^T$ trivially satisfies the sufficient condition for being a minimizer in $(\ref{maxtomin})$.

\textbf{Case 2: $\theta < 2$}

In this case, $\psi(y_i) \in [0, 4e^{-2}]$. If $c_t(\lambda) > 4 e^{-2}$ for some $t \in \{1, \ldots, K \}$, then $(\ref{m08})$ can hold only if $\mu_t > 0$ in which case $x_t = 0$ by the complementary slackness condition in $(\ref{comp_slack})$. Hence, if $c_t(\lambda) > 4 e^{-2}$, then $(\ref{m08})$ simplifies to $\mu_t = \lambda 2^{R(t-1)}$. Since $c_i(\lambda)$ is increasing in $i$, if $c_t(\lambda) > 4 e^{-2}$ for some $t$, then $c_j(\lambda) > 4 e^{-2}$ for all $j \geq t$, leading to $x_j = 0$ for all $j \geq t$. On the other hand, if $0 < c_i(\lambda) \leq 4 e^{-2}$, then $0 \leq \mu_i \leq \lambda 2^{R(i-1)}$ in order to satisfy $(\ref{m08})$. If $\mu_i = \lambda 2^{R(i-1)}$, then $y_i = +\infty$ is the solution to $(\ref{m08})$, giving us $x_i = 0$. If $\mu_i < \lambda 2^{R(i-1)}$, then there is a finite solution to $(\ref{m08})$ so that $\mu_i = 0$ and $(\ref{m08})$ simplifies to 
\begin{align}
    \psi(y_i) &= c_i(\lambda).  \label{sol_yis}
\end{align}
This argument also establishes strict complementarity, i.e., either $\mu_i = 0$ or $x_i = 0$, but not both.

Using Corollary \ref{optxproperties}, an optimal $\overrightarrow x$ must satisfy 
\begin{align}
    x_i \begin{cases}
        > 0 \text{ for } i \in \{1, \ldots, \ell \}\\
        = 0 \text{ otherwise}
    \end{cases} 
    \label{bont}
\end{align}
for some integer $\ell \in \{1, \ldots, K \}$. Hence, we have 
\begin{align*}
    \ell &= \max \{1 \leq i \leq K: \mu_i = 0 \}\\
    &= \max \{1 \leq i \leq K: x_i > 0  \}\\
    &\leq \max \left \{1 \leq i \leq K : c_i(\lambda) \leq 4 e^{-2} \right \}. 
\end{align*}
We have 
\begin{align}
    0 < c_1(\lambda) < \cdots  < c_\ell(\lambda) \leq 4 e^{-2} \label{c_mono}
\end{align}
and, from Corollary $\ref{optxproperties}.(\ref{opt_prop_4x})$,  
\begin{align}
    1\geq x_1 > x_2 > \cdots > x_\ell > 0 \quad \iff \quad \theta \leq y_1 < \cdots < y_\ell < \infty \label{mono_y}
\end{align}

\begin{lemma}
Let $\theta \in (0, 2)$. Then for any point $(\overrightarrow x, \lambda, \overrightarrow \mu)$ satisfying $(\ref{stationarity}) - (\ref{monotonicity_x})$ and $(\ref{bont})$, where $\overrightarrow x$ is also\footnote{Every maximizer in $(\ref{b3})$ satisfies $(\ref{stationarity}) - (\ref{monotonicity_x})$ and $(\ref{bont})$ for a unique $(\lambda, \overrightarrow \mu)$, but not every $(\overrightarrow x, \lambda, \overrightarrow \mu)$ satisfying $(\ref{stationarity}) - (\ref{monotonicity_x})$ and $(\ref{bont})$ implies that $\overrightarrow x$ is a maximizer in $(\ref{b3})$.} a maximizer in $(\ref{b3})$, we must have $\theta^2 e^{-\theta} \leq c_1(\lambda) < \cdots < c_\ell(\lambda) \leq 4 e^{-2}$. In particular,
\begin{align}
d_1 \theta e^{-\theta} \leq \lambda \leq \frac{4 e^{-2} d_\ell}{\theta 2^{R(\ell-1)}}.    
\end{align} 
    \label{lambda_lwbd}
\end{lemma}

 \begin{IEEEproof}
 We first note that for $1 \leq i \leq \ell$, all solutions to $(\ref{sol_yis})$ must be finite so that each $x_i > 0$. In particular, they must satisfy $(\ref{mono_y})$. 
 
 Now suppose $0 < c_u(\lambda) < \theta^2 e^{-\theta}$ for some $u \in \{1, \ldots, \ell \}$. Then there is a unique solution $y_u^+ \in (2, \infty)$ of $\psi(y) = c_i(\lambda)$ over the interval $[\theta, \infty)$ so that $x_u = \theta/y_u^+ < 1$. If $u < \ell$, then since $c_{u + 1}(\lambda) > c_u(\lambda)$, there is no solution $y_{u+1}$ to $(\ref{sol_yis})$ for $i=u + 1$ such that $\infty > y_{u+1} > y_u^+$. Hence, we must have $u = l$. Then if $u = \ell > 1$, since $c_{u-1}(\lambda) < c_u(\lambda)$, the unique solution $y_{u-1}^+$ to $(\ref{sol_yis})$ for $i=u-1$ satisfies $y_{u-1}^+ > y_u^+$, contradicting $(\ref{mono_y})$. So we must have $u = \ell = 1$. This gives $x_1 = \theta/y_1^+$ and $x_i = 0$ for all $i \geq 2$. Since $c_1(\lambda) < \theta^2 e^{-\theta}$, we have $y_1^+ > 2$ and $x_1 < 1$ since $\theta < 2$. But this is clearly not a maximizer in $(\ref{b3})$ so we must have $c_i(\lambda) \geq \theta^2 e^{-\theta}$ for all $i \in \{1, \ldots, \ell \}$. 
 
 Lastly, since $c_\ell(\lambda) \leq 4e^{-2}$ from $(\ref{c_mono})$, we have   
\begin{align}
    \frac{\lambda \theta 2^{R(\ell-1)}}{d_\ell} \leq 4 e^{-2} \iff \lambda \leq \lambda_{\max}(\ell) \coloneqq \frac{4 e^{-2} d_\ell}{\theta 2^{R(\ell-1)}}.
\end{align}     
 \end{IEEEproof}

From Lemma \ref{lambda_lwbd}, we have for $\theta \in (0, 2)$,    
\begin{align}
d_1 \theta e^{-\theta} &\leq \frac{4 e^{-2} d_\ell}{\theta 2^{R(\ell-1)}}\\
    \iff   \theta^2 e^{-\theta} &\leq \frac{4 e^{-2} d_\ell}{ 2^{R(\ell-1)} d_1 } \label{lstrongcond}\\
    \implies  \theta^2 e^{-\theta} &\leq \frac{4 e^{-2} }{ 2^{R(\ell-1)}  }. \label{lweakcond}
\end{align}

The parameter $\ell$ of a minimizer $\overrightarrow x$ satisfies both $(\ref{lstrongcond})$ and $(\ref{lweakcond})$. Therefore, if a certain $\ell' \in \{1, \ldots, K \}$ does not satisfy $(\ref{lstrongcond})$ or $(\ref{lweakcond})$, then since the RHS of both $(\ref{lstrongcond})$ and $(\ref{lweakcond})$ is decreasing in $\ell$, the number of strictly positive $x_i$'s in a minimizer is $\leq \ell' -1$. Since $\theta \in (0, 2)$, we have proved $(\ref{elllmdef})$ in Theorem $\ref{active_layer_theorem}$. Furthermore, $(\ref{particularthm})$ in Theorem $\ref{active_layer_theorem}$ is the  statement that $(\ref{lstrongcond})$ is not satisfied for $\ell = 2$ so that it not satisfied for any $\ell \geq 2$. This completes the proof of Theorem \ref{active_layer_theorem}. A weakened version of Theorem $\ref{active_layer_theorem}$ is given below. 
\begin{corollary}
The parameter $\ell$ of any maximizer $\overrightarrow x$ in $(\ref{b3})$ satisfies 
    \begin{align*}
        \ell \leq \begin{cases}
            1 & \text{ if } \theta \geq 2\\
            1 + \frac{1}{R} \log \left(  \frac{4 e^{-2}}{\theta^2 e^{-\theta}} \right) & \text{ if } \theta < 2.
        \end{cases}    
    \end{align*}
    \label{ell3}
\end{corollary}
\begin{IEEEproof}
Corollary $\ref{ell3}$ is a restatement of $(\ref{lweakcond})$ coupled with the result from Case $1$: $\theta \geq 2$.  
\end{IEEEproof}

For $1 \leq i \leq \ell$, since  $\theta^2 e^{-\theta} \leq c_i(\lambda) \leq 4 e^{-2}$, we have two solutions, denoted henceforth as $y_i^-, y_i^+$, of $\psi(y) = c_i(\lambda)$ so that $y_i^- \in [\theta, 2]$ and $y_i^+ \in [2, \infty)$ for each $1 \leq i \leq \ell$. Specifically, 
\begin{align}
    y_i^- &= -2W_0\left( -\frac{\sqrt{c_i(\lambda)}}{2} \right), \label{yi-sol}\\
    y_i^+ &= -2W_{-1}\left( -\frac{\sqrt{c_i(\lambda)}}{2} \right). \label{yi+sol}
\end{align}

From $(\ref{c_mono})$ and the fact that $\psi(\cdot)$ is increasing over $[\theta, 2]$ and decreasing over $[2, \infty)$, we have 
\begin{align}
    \theta &\leq y_1^- < y_2^- < \cdots < y_\ell^- \leq 2 \leq y_\ell^+ < \cdots < y_1^+ < \infty. \label{order_of_sols} 
\end{align}
From $(\ref{order_of_sols})$, it is clear that if $\ell = 1$, we must have $x_1 = \theta/y_1^- = 1$, $c_1(\lambda) = \theta^2 e^{-\theta}$ and $y_1^- = \theta$. If $\ell \geq 2$, then in order to satisfy $(\ref{mono_y})$, we must have $x_i = \theta/y_i^-$ for all $i \in \{1, \ldots, \ell - 1 \}$ whereas $x_\ell$ can either be $x_\ell = \theta/y_\ell^-$ or $x_\ell = \theta/y_\ell^+$. Note that $x_i \geq \theta/2$ for all $i \in \{1,\ldots, \ell - 1 \}$, whereas 
\begin{align*}
    x_\ell \begin{cases}
        \geq \theta/2 & \text{ if } x_\ell = \theta/y_\ell^-\\
        \leq \theta/2 & \text{ if } x_\ell = \theta/y_\ell^+. 
    \end{cases}
\end{align*}
\begin{lemma}
\label{x1cond}
    For $\theta \in (0, 2)$, any minimizer $(x_1, \ldots, x_K)$ in $(\ref{maxtomin})$ satisfies $x_1 > \theta/2$.
\end{lemma}

\begin{IEEEproof}
    If $\ell = 1$, we have $x_1 = 1 > \theta/2$. If $\ell > 1$, then $c_1(\lambda) < 4e^{-2}$ so that $y_1^- < 2$ and $x_1 = \theta/y_1^- > \theta/2$.  
\end{IEEEproof}

The only candidates for the minimizers in $(\ref{maxtomin})$ are $\overrightarrow x^- = (x_1^-, \ldots, x_{\ell}^-, 0, \ldots, 0)$ and $\overrightarrow x^+ = (x_1^+, \ldots, x_\ell^+, 0, \ldots, 0)$, where  
\begin{enumerate}
    \item \label{cand1min} $x_i^- = \theta/y_i^-$ for all $i \in \{1, \ldots, \ell \}$ and $x_i^- = 0$ for $\ell < i \leq K$.   
    \item \label{cand2min} $x_i^+ = \theta/y_i^-$ for all $i \in \{1, \ldots, \ell-1 \}$, $x_\ell^+ = \theta/y_\ell^+$ and $x_i^+ = 0$ for $\ell  <i \leq K$.  
\end{enumerate}
Note that the dependence of $\overrightarrow x^-$ and $\overrightarrow x^+$ on $\lambda$ is implicit and is given by $(\ref{yi-sol})$ and $(\ref{yi+sol})$. From $(\ref{eq_const})$, any minimizer $\overrightarrow x$ in $(\ref{maxtomin})$ must satisfy  
\begin{align}
    \sum_{i=1}^\ell 2^{R(i-1)} x_i = 1 \label{eq_toc}
\end{align}
for some $\lambda \in [\lambda_{\min}, \lambda_{\max}(\ell)]$. Algorithm \ref{algorithmglobalmax} searches for a global maximizer in $(\ref{b3})$ by evaluating $\overrightarrow x^-$ and $\overrightarrow x^+$ for all possible values of $\ell \leq \ell_{\operatorname{PDS}}$. For $\overrightarrow x = \overrightarrow x^-$, it can be checked that each $x_i^- = \theta/y_i^-$ is continuous and decreasing as a function of $\lambda \in [\lambda_{\min}, \lambda_{\max}(\ell)]$, which makes the LHS of $(\ref{eq_toc})$ also continuous and decreasing in $\lambda$. This proves that $H_\ell^-(\lambda_{\min}) \geq 1 \geq H_{\ell}^-(\lambda_{\max}(\ell))$ is a necessary and sufficient condition for a unique $\lambda$ satisfying $(\ref{eq_toc})$ for $\overrightarrow x = \overrightarrow x^-$. 

For $\overrightarrow x = \overrightarrow x^+$, finding the roots of $H_{\ell}^+(\lambda) = 1$ is more complex. For every $\lambda \in [\lambda_{\min}, \lambda_{\max}(\ell)]$, we have $H_\ell^-(\lambda) \geq H_\ell^+(\lambda)$. Hence, if $H_\ell^-(\lambda_{\min}) < 1$, we have $H_\ell^+(\lambda) \leq H_\ell^-(\lambda) < 1$
for all $\lambda \in [\lambda_{\min}, \lambda_{\max}(\ell)]$, i.e., no roots exist. Hence, a necessary condition for the existence of solutions of $H_\ell^+(\lambda) = 1$ over $\lambda \in [\lambda_{\min}, \lambda_{\max}(\ell)]$ is $H_\ell^-(\lambda) \geq 1$. Let $s = y_\ell^+$. Then  
\begin{align*}
    s = -2W_{-1}\left( -\frac{\sqrt{c_\ell(\lambda)}}{2} \right) \iff  \frac{s^2 e^{-s} d_\ell}{\theta 2^{R(\ell - 1)}} = \lambda \text{ and } 2 \leq s &\leq -2W_{-1}\left( - \frac{1}{2}\sqrt{\frac{d_1 \theta^2 e^{-\theta} 2^{R(\ell-1)}}{d_\ell}} \right). 
\end{align*}
Note that as $\lambda$ increases from $\lambda_{\min}$ to $\lambda_{\max}(\ell)$, we have $s$ decreasing from $s_{\ell, \max}$ to $2$, where
\begin{align*}
   s_{\ell, \max} \coloneqq  -2W_{-1}\left( - \frac{1}{2}\sqrt{\frac{d_1 \theta^2 e^{-\theta} 2^{R(\ell-1)}}{d_\ell}} \right). 
\end{align*}
We reparametrize $H_{\ell}^+(\lambda) = F_{\ell}(s)$ in terms of $s$ as follows: 
\begin{align*}
    F_{\ell}(s) &= \theta\left( \sum_{i=1}^{\ell-1}  \frac{2^{R(i-1)}}{-2W_0\left(-\frac{1}{2} \sqrt{s^2 e^{-s}\frac{d_{\ell} 2^{iR}}{d_i 2^{\ell R}}}  \right)} + \frac{ 2^{R(\ell-1)}}{s}\right).
\end{align*}
Define
\begin{align*}
  \beta_{i, \ell}(s) &= -\frac{1}{2} \sqrt{s^2 e^{-s}\frac{d_{\ell} a_i}{d_i a_{\ell}}} \quad \, \text{ for } i < \ell,\\
  a_i &= 2^{R(i-1)} \quad\quad\quad\quad\quad \text{ for } i \leq \ell,\\
  t_i(s) &= -2 W_0\left(\beta_{i, \ell}(s)  \right)\,\,\, \quad \text{ for } i < \ell.
\end{align*}
We then have  
\begin{align*}
    F_{\ell}(s) &= \theta\left( \sum_{i=1}^{\ell-1}  \frac{a_i}{t_i(s)} + \frac{a_\ell}{s}\right),\\
    F_{\ell}'(s) &= \frac{\theta}{s^2}\left( \sum_{i=1}^{\ell-1}  a_i Q_{i, \ell}(s) - a_\ell\right),
\end{align*}
where 
\begin{align*}
   Q_{i, \ell}(s) &= \frac{s(2-s)}{t_i(t_i - 2)},\\
    Q_{i, \ell}'(s) &= \frac{2(s - t_i)(s + t_i - st_i)}{t_i(2 -t_i)^3}, 
\end{align*}
where we used the shorthand $t_i = t_i(s)$. Recall that $t_i(s) = y_i^-$ from before so that 
$$0 < \theta \leq t_1(s) < \cdots < t_{\ell-1}(s) < 2 \leq y_\ell^+ = s.$$
In particular, 
\begin{align}
    t_i(2) &= -2W_0\left( -\frac{1}{2} \sqrt{4 e^{-2}\frac{d_{\ell} a_i}{d_i a_{\ell}}} \right) \notag\\
    &< -2 W_0\left(-e^{-1} \right) = 2. \label{sjhdfjfddsfd}
\end{align}
\begin{lemma}
For $s > 2$, $Q_{i, \ell}'(s) > 0$ for all $i = 1, \ldots, \ell - 1$.
\end{lemma}
\begin{IEEEproof}
Recall the definition $\psi(t) = t^2 e^{-t}$ and that $\psi(t_i) < \psi(s)$. Let $\tau = \frac{s}{s - 1} \in (1, 2)$. Consider 
\begin{align*}
   f(u) = \ln \left(\frac{\psi(\tau)}{\psi(s)}\right) \Big |_{s = 1 + u} &=  u -\frac{1}{u} -2 \ln(u).
\end{align*}
Its derivative is 
\begin{align*}
    f'(u) = \frac{(u-1)^2}{u^2} > 0
\end{align*}
since $u = s-1 > 1$. Hence, $f(1) = 0$ and $f(u) > 0$ for all $u > 1$. This implies that 
\begin{align*}
    \psi(\tau) > \psi(s) > \psi(t_i(s))
\end{align*}
for all $s > 2$. Since $\tau, t_i < 2$ and $\psi$ is increasing over $(0, 2)$, we have $\tau > t_i$ which is the same as $s + t_i > s t_i$.  Since $s > 2 > t_i$, we have the desired result.    
\end{IEEEproof}
   Therefore, 
\begin{align*}
    M_\ell(s) &\coloneqq  \sum_{i=1}^{\ell-1}  a_i Q_{i, \ell}(s) - a_\ell
\end{align*}
is strictly increasing over $(2, s_{\ell, \max}]$. 

Since $F_\ell'(s) = \frac{\theta}{s^2} M_\ell(s)$, we have 
\begin{align*}
    \left \{s \in (2, s_{\ell, \max}] : F_\ell'(s) = 0 \right \} = \left \{s \in (2, s_{\ell, \max}] : M_\ell(s) = 0 \right \}.  
\end{align*}
Since $M_\ell(s)$ is strictly increasing, there is at most one zero of $F_\ell'(s) = 0$. Also, we have 
\begin{align*}
    \lim_{s \downarrow 2} F_\ell'(s) &= \lim_{s \downarrow 2}  \frac{\theta}{s^2}\left( \sum_{i=1}^{\ell-1}  a_i Q_{i, \ell}(s) - a_\ell\right)\\
    &= \lim_{s \downarrow 2}  \frac{\theta}{4}\left( \sum_{i=1}^{\ell-1}  a_i \frac{0}{t_i(2)(t_i(2) - 2)} - a_\ell\right)\\
    &= -\frac{\theta a_\ell}{4} < 0, 
\end{align*}
where the last equality follows from $(\ref{sjhdfjfddsfd})$.

From this information, one of the following must be true:  
\begin{itemize}
    \item $F_\ell'(s) < 0$ for all $s \in (2, s_{\ell, \max}]$ in which case $F_\ell(s)$ is strictly decreasing over $(2, s_{\ell, \max}]$ so a unique solution $s_0$ of $F_\ell(s) = 1$ obtainable by a bisection search exists if and only if $F_\ell(2) \geq 1 \geq F_\ell(s_{\ell, \max})$. We then set $\lambda = \frac{s_0^2 e^{-s_0} d_\ell}{\theta 2^{R(\ell - 1)}}$. 
    \item $F_\ell'(s) < F_\ell'(s_0) = 0 < F_\ell'(s)$ for some $s_0 \in (2, s_{\ell, \max}]$, in which case $F_\ell(s)$ is decreasing over $(2, s_0]$ and then increasing over $[s_0, s_{\ell, \max}]$. In this case, no solution of $F_\ell(s) = 1$ exists if $F_\ell(s_0) > 1$, exactly one solution exists if $F_\ell(s_0) = 1$ and at most two solutions exist if $F_\ell(s_0) < 1$.   
\end{itemize}

Since $M_\ell(2^+) < 0$, the above two scenarios can be checked as follows:
\begin{itemize}
    \item If $M_\ell(s_{\ell, \max}) \leq 0$ and $F_\ell(2) \geq 1 \geq F_\ell(s_{\ell, \max})$, then use bisection to find the unique root $s_0$ of $F_\ell(s) = 1$ over the interval $(2, s_{\ell, \max}]$. Return $\{\lambda \}$ where $\lambda = \frac{s_0^2 e^{-s_0} d_\ell}{\theta 2^{R(\ell - 1)}}$. 
    \item If $M_\ell(s_{\ell, \max}) > 0$, then first find the unique root $s_0$ of $M_\ell(s) = 0$ using bisection on the interval $(2, s_{\ell, \max}]$. If $F_\ell(s_0) > 1$, then return empty set. If $F_\ell(s_0) = 1$  return $\{\lambda \}$ where $\lambda = \frac{s_0^2 e^{-s_0} d_\ell}{\theta 2^{R(\ell - 1)}}$. If $F_\ell(s_0) < 1$, then  
    \begin{itemize}
        \item If $ F_\ell(2) \geq 1$ and $F_\ell(s_{\ell, \max}) < 1$, then use bisection to find the unique root $s_1$ of $F_\ell(s) = 1$ over the interval $[2, s_0]$. Then return $\{\lambda \}$ where $\lambda = \frac{s_1^2 e^{-s_1} d_\ell}{\theta 2^{R(\ell - 1)}}$.
        \item  If $F_\ell(2) < 1$ and $F_\ell(s_{\ell, \max}) \geq 1$, then use bisection to find the unique root $s_2$ of $F_\ell(s) = 1$ over the interval $[s_0, s_{\ell, \max}]$. Then return $\{\lambda \}$ where $\lambda = \frac{s_2^2 e^{-s_2} d_\ell}{\theta 2^{R(\ell - 1)}}$.
        \item If $F_\ell(2) \geq 1$ and $F_\ell(s_{\ell, \max}) \geq 1$, then use bisection to find the unique root $s_1$ of $F_\ell(s) = 1$ over the interval $[2, s_0]$ and use another bisection to find the unique root $s_2$ of $F_\ell(s) = 1$ over the interval $[s_0, s_{\ell, \max}]$. Then return $\{\lambda_1, \lambda_2 \}$ where $\lambda_1 = \frac{s_1^2 e^{-s_1} d_\ell}{\theta 2^{R(\ell - 1)}}$ and $\lambda_2 = \frac{s_2^2 e^{-s_2} d_\ell}{\theta 2^{R(\ell - 1)}}$.
    \end{itemize}
    \item Otherwise, return an empty set.    
\end{itemize}
The above description is exactly Algorithm \ref{mixed_branch_roots}. This completes the proof of Theorem \ref{global_maximizer_thm}.     

To prove Theorem \ref{local_maximizer_thm}, note that $\overrightarrow x^-$ already satisfies $(\ref{gf3})$ and $(\ref{lspecification})$ as well as the KKT conditions. Hence, it suffices to show that $\overrightarrow x^-$ is a strict local minimizer in $(\ref{maxtomin})$. To do that, we evaluate the Lagrangian Hessian, which is a diagonal matrix given by  
\begin{align*}
    \frac{\partial^2 \mathcal{L}(\overrightarrow x, \lambda, \overrightarrow \mu)}{\partial x_i^2} &= - g''(x_i) d_i \\
    &= \begin{cases}
        - g''(x_i) d_i & \text{ if } 1 \leq i \leq \ell\\
        0 & \text{ if } i > \ell.
    \end{cases}
\end{align*}
The critical cone \cite[p. 330]{nonlinopt} at the point $\overrightarrow x^-$ is given by 
\begin{align*}
    \mathcal{C}(\overrightarrow x^-) &= \Bigg \{\overrightarrow w \in \mathbb{R}^K: \sum_{i=1}^K (2^R)^{i-1} w_i = 0,\\ 
    &\quad \quad \quad \quad \quad \quad \quad  w_{\ell+1} = \cdots = w_K = 0 \Bigg \}\\
    &= \Bigg \{\overrightarrow w \in \mathbb{R}^K: \sum_{i=1}^\ell (2^R)^{i-1} w_i = 0,\\ 
    & \quad \quad \quad \quad \quad \quad \quad  w_{\ell+1} = \cdots = w_K = 0 \Bigg \}. 
\end{align*}
Since each KKT point $\overrightarrow x$ for the optimization problem in $(\ref{maxtomin})$ is associated with unique Lagrange multipliers $\lambda$ and $\overrightarrow \mu$ by Lemma \ref{unique_mult}, we write the critical cone as $\mathcal{C}(\overrightarrow x)$ instead of $\mathcal{C}(\overrightarrow x, \lambda, \overrightarrow \mu)$. From \cite[Theorem 12.5]{nonlinopt}, a second-order necessary condition for a KKT point $\overrightarrow x$ to be a local minimizer in $(\ref{maxtomin})$ is  
\begin{align}
    \sum_{i=1}^\ell - g''(x_i) d_i w_i^2 &\geq 0\\
    \iff \sum_{i=1}^\ell g''(x_i) d_i w_i^2 &\leq 0 \label{second_order_cond}
\end{align}
for all $(w_1, \ldots, w_\ell)$ satisfying 
\begin{align*}
    \sum_{i=1}^\ell (2^R)^{i-1} w_i = 0. 
\end{align*}
A sufficient condition is when the inequality $(\ref{second_order_cond})$ is strict. 

For $\overrightarrow x = \overrightarrow x^-$, each $x_i^- \geq \theta/2$ with $x_1^- > \theta/2$. Since $g''(x) < 0$ for $x > \theta/2$, it is evident that the inequality in $(\ref{second_order_cond})$ is a strictly inequality for $\overrightarrow x = \overrightarrow x^-$ and $\overrightarrow w \neq \overrightarrow 0$. Hence, from Theorem \cite[Theorem 12.6]{nonlinopt}, $\overrightarrow x = \overrightarrow x^-$ is a strict local minimizer in $(\ref{maxtomin})$.

\section{Proof of Theorems \ref{active_layer_theoremORA}, \ref{global_maximizer_thmORA} and \ref{local_maximizer_thmORA} \label{combined_proofORA}}

Using the same argument as used in the beginning of Section \ref{combined_proof}, it can be checked that the KKT conditions are necessary for a maximizer in $(\ref{maxORA})$. We start by rewriting  $(\ref{maxORA})$ as  
\begin{align}
     &\min_{\overrightarrow v \in \Delta^{K - 1}} \overline{T}(\overrightarrow v), \label{maxtominora} \\
    \text{ where }\quad  \overline{T}(\overrightarrow v) &= - \sum_{i=1}^K t(v_i) d_i,\\
    t(v) &= \exp \left( -\frac{2^{R/v} - 1}{2^R - 1} \theta \right)
\end{align}
and $t : [0, 1] \to [0, 1)$. Note that $\overline{T}(\overrightarrow v)$ is twice continuously differentiable over $\Delta^{K-1}$. Specifically, for $v > 0$, 
\begin{align*}
    t'(v) &= \frac{\theta  R \ln (2) 2^{R/v} t(v)}{\left(2^R-1\right) v^2},\\
    t''(v) &= \frac{\theta  R \ln (2)}{2^R - 1}  \frac{2^{R/v} t(v)}{v^4} \left(R \ln (2) \left(\frac{\theta  2^{R/v}}{2^R-1}-1\right)-2 v \right)
\end{align*}
with 
\begin{align*}
    t'(0) &= \lim_{v \downarrow 0} t'(v) = 0,\\
    t''(0) &= \lim_{v \downarrow 0} t''(v) = 0.
\end{align*}
Also note that $t'(v) > 0$ for all $v > 0$. Define the Lagrangian as 
\begin{align*}
    \mathcal{L}(\overrightarrow v, \lambda, \overrightarrow \mu ) &= -\sum_{i=1}^K t(v_i) d_i + \lambda\left(  \sum_{i=1}^K  v_i - 1\right) - \sum_{i=1}^K  \mu_i v_i, 
\end{align*}
where $\lambda \in \mathbb{R}$ and $\mu_i \geq 0$. The KKT conditions \cite[Theorem 12.1]{nonlinopt} and the optimality properties from Theorem \ref{opt_mu_properties_theorem} imply the following: for any maximizer in $(\ref{maxORA})$, there exist Lagrange multipliers $(\lambda, \overrightarrow \mu)$ such that
\begin{align}
    -t'(v_i) d_i + \lambda - \mu_i &= 0\label{stationarityORA}\\
    \mu_i v_i &= 0 \label{compslackORA}\\
    \mu_i &\geq 0 \label{dualfeasibilityora}\\
    \sum_{i=1}^K v_i &= 1 \label{eq_consora}\\
    1 &\geq v_1 \geq \cdots \geq v_K \geq 0, \label{monotonicity_v}
\end{align}
where $(\ref{stationarityORA}), (\ref{compslackORA})$ and $(\ref{dualfeasibilityora})$ hold for all $i \in\{1, \ldots, K \}$.

\begin{lemma}
For any point $(\overrightarrow v, \lambda, \overrightarrow \mu)$ satisfying $(\ref{stationarityORA}) - (\ref{monotonicity_v})$, we have $\lambda > 0$. 
    \label{lambda_posi}
\end{lemma}
\begin{lemma}
    For any minimizer $\overrightarrow v$ in $(\ref{maxtominora})$, there exists a unique $(\lambda, \overrightarrow \mu)$ such that $(\overrightarrow x, \lambda, \overrightarrow \mu)$ satisfies the conditions $(\ref{stationarityORA}) - (\ref{monotonicity_v})$.
    \label{unique_multi_ORA}
\end{lemma}
The proofs of Lemmas \ref{lambda_posi} and \ref{unique_multi_ORA} are similar to those of Lemmas \ref{lambda_pos_res} and \ref{unique_mult}, respectively.

We can rewrite $(\ref{stationarityORA})$ as
\begin{align}
    \mathscr{U}(v_i)  &= \mathscr{C}_i(\lambda)  -\frac{\mu_i \left(2^R-1\right)}{\theta d_i  R \ln (2)}, \label{m08ORA}
\end{align}
where we define $\mathscr{C}_i(\lambda) \coloneqq \frac{\lambda\left(2^R-1\right)}{\theta d_i R \ln (2)}$ and the function $\mathscr{U}: [0, 1] \to [0, \infty)$ as 
\begin{align*}
    \mathscr{U}(v) &= \frac{ 2^{R/v}}{v^2} \exp \left( -\frac{2^{R/v} - 1}{2^R - 1} \theta \right),\\
    \mathscr{U}(0) &= \mathscr{U}(0^+) = 0.
\end{align*}
Let 
\begin{align*}
    \mathscr{V}(v) =  \ln \mathscr{U}(v) &= \frac{R}{v} \ln(2) - 2 \ln(v)  -\frac{2^{R/v} - 1}{2^R - 1} \theta.
\end{align*}
Then 
\begin{align*}
    \mathscr{V}'(v) &= \frac{1}{\mathscr{U}(v)} \mathscr{U}'(v)  \\
    \mathscr{U}'(v) &= \mathscr{U}(v) \mathscr{V}'(v)\\
    &= \frac{\mathscr{U}(v)}{v^2} \left(R \ln (2) \left(\frac{\theta  2^{R/v}}{2^R-1}-1\right)-2 v \right). 
\end{align*}
Since $\mathscr{U}(v) > 0$ for $v > 0$, the sign of $\mathscr{U}'(v)$ is the sign of $N(v)$ defined as 
\begin{align*}
    N(v) &= R \ln (2) \left(\frac{\theta  2^{R/v}}{2^R-1}-1\right)-2 v. 
\end{align*}
We have 
\begin{align*}
    N'(v) &= -\frac{\theta  R^2 \ln ^2(2) 2^{R/v}}{\left(2^R-1\right) v^2}-2 < 0
\end{align*}
for all $v \in [0, 1]$. Since $N(0^+) = +\infty$, $\mathscr{U}(v)$ is always increasing for sufficiently small $v \to 0^+$. Hence, if $N(1) < 0$, $\mathscr{U}(v)$ increases and then decreases over the interval $[0,1]$. Otherwise, $\mathscr{U}(v)$ is increasing throughout $[0, 1]$. We have   
\begin{align*}
    N(1) &= R \ln (2) \left(\frac{\theta  2^{R}}{2^R-1}-1\right)-2.
\end{align*}
\begin{align*}
    N(1) \geq 0 \iff \theta \geq \theta_c \coloneqq \frac{2^R - 1}{2^R} \left( \frac{2}{R \ln(2)} + 1 \right ).
\end{align*}

\textbf{Case 1 $\theta \geq \theta_c$:}

In this case, $\mathscr{U}$ is monotonically increasing over $[0, 1]$. Recall that $\mathscr{U}(0^+) = 0$ and $\mathscr{U}(1) = 2^R e^{-\theta}$. If $\mathscr{C}_t(\lambda) > 2^R e^{-\theta}$ for some $t \in \{1, \ldots, K \}$, then $(\ref{m08ORA})$ can hold only if $\mu_t > 0$ in which case $v_t = 0$ by the complementary slackness condition in $(\ref{compslackORA})$. Hence, if $\mathscr{C}_t(\lambda) > 2^R e^{-\theta}$, $(\ref{m08ORA})$ simplifies to $\mu_t = \lambda$. Given $\lambda > 0$ from Lemma \ref{lambda_posi}, $\mathscr{C}_i(\lambda)$ is increasing in $i$ since $d_i$ is decreasing in $i$. Hence, if $\mathscr{C}_t(\lambda) > 2^R e^{-\theta}$ for some $t$, then $\mathscr{C}_j(\lambda) > 2^R e^{-\theta}$ for all $j \geq t$, leading to $v_j = 0$ for all $j \geq t$. This is consistent with the structure of an optimal $\overrightarrow v$ as established in Theorem $\ref{opt_mu_properties_theorem}$. On the other hand, if $0 < \mathscr{C}_i(\lambda) \leq 2^R e^{-\theta}$, then $0 \leq \mu_i \leq \lambda$ in order to satisfy $(\ref{m08ORA})$. If $\mu_i = \lambda$, then $v_i = 0$. If $\mu_i < \lambda$, then there is a nonzero solution to $(\ref{m08ORA})$ so that $\mu_i = 0$ and $(\ref{m08ORA})$ simplifies to
\begin{align}
    \mathscr{U}(v_i) = \mathscr{C}_i(\lambda). \label{simpli_for_ORA_ilessl}
\end{align}
Using Theorem \ref{opt_mu_properties_theorem}, an optimal $\overrightarrow v$ must satisfy 
\begin{align}
    v_i \begin{cases}
        > 0 \text{ for } i \in \{1, \ldots, \ell \}\\
        = 0 \text{ otherwise}
    \end{cases} 
    \label{bo2ntORA}
\end{align}
for some integer $\ell \in \{1, \ldots, K \}$. Hence, we must have 
\begin{align*}
    \ell &= \max \{1 \leq i \leq K: \mu_i = 0 \}\\
    &\leq \max \left \{1 \leq i \leq K : \mathscr{C}_i(\lambda) \leq 2^R e^{-\theta} \right \}. 
\end{align*}
We have 
\begin{align}
    0 < \mathscr{C}_1(\lambda) < \cdots  < \mathscr{C}_\ell(\lambda) \leq 2^R e^{-\theta} \label{c_m2onoORA}
\end{align}
and, from Theorem $\ref{opt_mu_properties_theorem}$,  
\begin{align}
    1\geq v_1 > v_2 > \cdots > v_\ell > 0. \label{mo2noORA}
\end{align}

Since $\mathscr{U}$ is strictly increasing over $[0, 1]$, there exists a unique solution $v_i^* > 0$ of $(\ref{simpli_for_ORA_ilessl})$ for $i \in \{1, \ldots, \ell \}$. From $(\ref{c_m2onoORA})$ and the fact that $\mathscr{U}$ is strictly increasing over $[0, 1]$, we have 
\begin{align}
    0 < v_1^* < \cdots < v_\ell^* \leq 1. \label{wrongorder}
\end{align}
From $(\ref{wrongorder})$, it is clear that in order to satisfy $(\ref{mo2noORA})$, we must have $\ell = 1$. Hence, for $\theta \geq \theta_c$, 
\begin{itemize}
    \item $\mu_1 = 0$ and $\mu_i = \lambda $ for all $i \geq 2$,
    \item $\mathscr{C}_1(\lambda) = 2^R e^{-\theta},$
    \item $v_1 = 1$ and $v_i = 0$ for all $i \geq 2$. 
\end{itemize}
From \cite[Theorem 12.6]{nonlinopt}, a sufficient condition for $\overrightarrow v = (1, 0, \ldots, 0)$ to be a minimizer in $(\ref{maxtominora})$ is $\overrightarrow w^T \nabla^2_{vv} \mathcal{L}(\overrightarrow v, \lambda, \overrightarrow \mu) \overrightarrow w > 0$ for all $\overrightarrow w \neq 0$ in the critical cone \cite[p. 330]{nonlinopt} for this $\overrightarrow v$. But it can be checked that the critical cone at this point is just the zero vector so $\overrightarrow v = (1, 0, \ldots, 0)^T$ trivially satisfies the sufficient condition for being a minimizer in $(\ref{maxtominora})$.

\textbf{Case 2 $\theta < \theta_c$:}

In this case, $\mathscr{U}$ is unimodal, i.e., increases then decreases over $[0, 1]$. Define 
\begin{align*}
    v_{\operatorname{int}}^* &\coloneqq \argmax_{v \in [0, 1]} \mathscr{U}(v),\\
    M^*_{\operatorname{int}} &\coloneqq \max_{v \in [0, 1]} \mathscr{U}(v).  
\end{align*}
Note that $v_{\operatorname{int}}^* \in (0, 1)$ when $\theta < \theta_c$. Note that 
\begin{itemize}
    \item $t''(v) > 0$ for $0 < v < v_{\operatorname{int}}^*$ (convex),
    \item $t''(v) < 0$ for $v > v_{\operatorname{int}}^*$ (concave). 
\end{itemize}

If $\mathscr{C}_t(\lambda) > M^*_{\operatorname{int}}$ for some $t \in \{1, \ldots, K \}$, then $(\ref{m08ORA})$ can hold only if $\mu_t > 0$ in which case $v_t = 0$ by the complementary slackness condition in $(\ref{compslackORA})$. Hence, if $\mathscr{C}_t(\lambda) > M^*_{\operatorname{int}}$, $(\ref{m08ORA})$ simplifies to $\mu_t = \lambda$. Given $\lambda > 0$ from Lemma \ref{lambda_posi}, $\mathscr{C}_i(\lambda)$ is increasing in $i$ since $d_i$ is decreasing in $i$. Hence, if $\mathscr{C}_t(\lambda) > M^*_{\operatorname{int}}$ for some $t$, then $\mathscr{C}_j(\lambda) > M^*_{\operatorname{int}}$ for all $j \geq t$, leading to $v_j = 0$ for all $j \geq t$. This is consistent with the structure of an optimal $\overrightarrow v$ as established in Theorem $\ref{opt_mu_properties_theorem}$. On the other hand, if $0 < \mathscr{C}_i(\lambda) \leq M^*_{\operatorname{int}}$, then $0 \leq \mu_i \leq \lambda$ in order to satisfy $(\ref{m08ORA})$. If $\mu_i = \lambda$, then $v_i = 0$. If $\mu_i < \lambda$, then there is a nonzero solution to $(\ref{m08ORA})$ so that $\mu_i = 0$ and $(\ref{m08ORA})$ simplifies to
\begin{align}
    \mathscr{U}(v_i) = \mathscr{C}_i(\lambda). \label{simpli_for_ORA_ilessl2}
\end{align}
Using Theorem \ref{opt_mu_properties_theorem}, an optimal $\overrightarrow v$ must satisfy 
\begin{align}
    v_i \begin{cases}
        > 0 \text{ for } i \in \{1, \ldots, \ell \}\\
        = 0 \text{ otherwise}
    \end{cases} 
    \label{bo2ntORA2}
\end{align}
for some integer $\ell \in \{1, \ldots, K \}$. Hence, we must have 
\begin{align*}
    \ell &= \max \{1 \leq i \leq K: \mu_i = 0 \}\\
    &\leq \max \left \{1 \leq i \leq K : \mathscr{C}_i(\lambda) \leq M^*_{\operatorname{int}} \right \}. 
\end{align*}
We have 
\begin{align}
    0 < \mathscr{C}_1(\lambda) < \cdots  < \mathscr{C}_\ell(\lambda) \leq M^*_{\operatorname{int}} \label{c_m2onoORA2}
\end{align}
and, from Theorem $\ref{opt_mu_properties_theorem}$,  
\begin{align}
    1\geq v_1 > v_2 > \cdots > v_\ell > 0. \label{mo2noORA2}
\end{align}

\begin{lemma}
Let $\theta \in (0, \theta_c)$. Then for any point $(\overrightarrow v, \lambda, \overrightarrow \mu)$ satisfying $(\ref{stationarityORA}) - (\ref{monotonicity_v})$ and $(\ref{bo2ntORA2})$, where $\overrightarrow v$ is also a minimizer in $(\ref{maxtominora})$, we must have $2^R e^{-\theta} \leq \mathscr{C}_1(\lambda) < \cdots < \mathscr{C}_\ell(\lambda) \leq M^*_{\operatorname{int}}$. In particular,
\begin{align}
\frac{2^R e^{-\theta} \theta d_1 R \ln(2)}{2^R - 1} \leq \lambda \leq \frac{M_{\operatorname{int}}^*\, \theta d_\ell R \ln (2)}{2^R - 1}.    \label{lambda_int_ora}
\end{align} 
    \label{lambda_lwbdora}
\end{lemma}

 \begin{IEEEproof}
 We first note that for $1 \leq i \leq \ell$, all solutions to $(\ref{simpli_for_ORA_ilessl2})$ must be nonzero and satisfy $(\ref{mo2noORA2})$. 
 
 Now suppose $0 < \mathscr{C}_u(\lambda) < 2^R e^{-\theta}$ for some $u \in \{1, \ldots, \ell \}$. Then there is a unique solution $v_u^- \in (0, v_{\operatorname{int}}^*)$ of $\mathscr{U}(v) = \mathscr{C}_i(\lambda)$ over the interval $[0,1]$. If $u < \ell$, then since $\mathscr{C}_{u + 1}(\lambda) > \mathscr{C}_u(\lambda)$, there is no solution $v_{u+1}$ to $(\ref{simpli_for_ORA_ilessl2})$ for $i=u + 1$ such that $v_{u+1} > v_u^-$. Hence, we must have $u = \ell$. Then if $u = \ell > 1$, since $\mathscr{C}_{u-1}(\lambda) < \mathscr{C}_u(\lambda)$, the unique solution $v_{u-1}^-$ to $(\ref{simpli_for_ORA_ilessl2})$ for $i=u-1$ satisfies $v_{u-1}^- < v_u^-$, contradicting $(\ref{mo2noORA2})$. So we must have $u = \ell = 1$. This gives $v_1 = v_1^- \in (0, v_{\operatorname{int}}^*)$ and $v_i = 0$ for all $i \geq 2$. Since $v_{\operatorname{int}}^* < 1$, $v_1 < 1$. But this is clearly not a minimizer in $(\ref{maxtominora})$ so we must have $\mathscr{C}_i(\lambda) \geq 2^R e^{-\theta}$ for all $i \in \{1, \ldots, \ell \}$. 
 
 Lastly, since $\mathscr{C}_\ell(\lambda) \leq M_{\operatorname{int}}^*$ from $(\ref{c_m2onoORA2})$, we have   
\begin{align}
    \frac{\lambda\left(2^R-1\right)}{\theta d_\ell R \ln (2)} \leq M_{\operatorname{int}}^* \iff \lambda \leq  \frac{M_{\operatorname{int}}^*\, \theta d_\ell R \ln (2)}{2^R - 1}.
\end{align}     
Similarly, 
\begin{align*}
    \frac{\lambda\left(2^R-1\right)}{\theta d_1 R \ln (2)} \geq 2^R e^{-\theta} \iff \lambda \geq \frac{2^R e^{-\theta} \theta d_1 R \ln(2)}{2^R - 1}. 
\end{align*}
 \end{IEEEproof}

From Lemma \ref{lambda_lwbdora}, we have for $\theta \in (0, \theta_c)$,
\begin{align}
   \frac{2^R e^{-\theta} \theta d_1 R \ln(2)}{2^R - 1} &\leq \frac{M_{\operatorname{int}}^*\, \theta d_\ell R \ln (2)}{2^R - 1}\\
   2^R e^{-\theta}  d_1  &\leq M_{\operatorname{int}}^*\,  d_\ell\\
   \frac{2^R e^{-\theta}}{M_{\operatorname{int}}^*} d_1 &\leq d_{\ell}. \label{63i} 
\end{align}
The parameter $\ell$ of a minimizer $\overrightarrow v$ in $(\ref{maxtominora})$ satisfies $(\ref{63i})$. Therefore, if a certain $\ell' \in \{1, \ldots, K \}$ does not satisfy $(\ref{63i})$, then since the RHS of $(\ref{63i})$ is decreasing in $\ell$, the number of strictly positive $v_i$'s in a minimizer is $\leq \ell' -1$. If $\theta \geq \theta_c$, then $M_{\operatorname{int}}^* = 2^R e^{-\theta}$ so $(\ref{63i})$ is equivalent to $\ell = 1$. This proves Theorem \ref{active_layer_theoremORA}. 

For $1 \leq i \leq \ell$, since  $2^R e^{-\theta} \leq \mathscr{C}_i(\lambda) \leq M_{\operatorname{int}}^*$, we have two solutions, denoted henceforth as $v_i^-$ and $v_i^+$, of $\mathscr{U}(v) = \mathscr{C}_i(\lambda)$ so that $v_i^- \in [0, v^*_{\operatorname{int}}]$ and $v_i^+ \in [v^*_{\operatorname{int}}, 1]$ for each $1 \leq i \leq \ell$. From $(\ref{c_m2onoORA2})$ and the fact that $\mathscr{U}$ is increasing over $[0, v_{\operatorname{int}}^*]$ and decreasing over $[v_{\operatorname{int}}^*, 1]$, we have 
\begin{align}
    0 < v_1^- <  \cdots < v_{\ell}^- \leq v_{\operatorname{int}}^* \leq v_\ell^+ < \cdots < v_{1}^+ \leq 1. \label{orderofsolsora}
\end{align}
From $(\ref{orderofsolsora})$, it is clear that if $\ell = 1$, we must have $v_1 = v_1^+$, $\mathscr{C}_1(\lambda) = 2^R e^{-\theta}$. If $\ell \geq 2$, then in order to satisfy $(\ref{mo2noORA2})$, we must have $v_i = v_i^+$ for all $i \in \{1, \ldots, \ell - 1 \}$ whereas $v_\ell$ can either be $v_\ell = v_{\ell}^+$ or $v_\ell = v_\ell^-$. Note that $v_i \geq v_{\operatorname{int}}^*$ for all $i \in \{1,\ldots, \ell - 1 \}$, whereas 
\begin{align*}
    v_\ell \begin{cases}
        \geq v_{\operatorname{int}}^* & \text{ if } v_\ell = v_{\ell}^+\\
        \leq v_{\operatorname{int}}^* & \text{ if } v_\ell = v_{\ell}^-. 
    \end{cases}
\end{align*}
\begin{lemma}
\label{v1cond}
    For $\theta \in (0, \theta_c)$, any minimizer $(v_1, \ldots, v_K)$ in $(\ref{maxtominora})$ satisfies $v_1 > v_{\operatorname{int}}^*$.
\end{lemma}

\begin{IEEEproof}
    If $\ell = 1$, we have $v_1 = 1 > v_{\operatorname{int}}^*$. If $\ell > 1$, then $\mathscr{C}_1(\lambda) < M_{\operatorname{int}}^*$ so that $v_1 = v_1^+ > v_{\operatorname{int}}^*$.  
\end{IEEEproof}

The only candidates for the minimizers in $(\ref{maxtominora})$ are $\overrightarrow v^+ \coloneqq (v_1^+, \ldots, v_{\ell}^+, 0, \ldots, 0)$ and $\overrightarrow v^- \coloneqq (v_1^+, \ldots, v_{\ell-1}^+, v_{\ell}^-, 0, \ldots, 0)$. 

Note that the dependence of $\overrightarrow v^-$ and $\overrightarrow v^+$ on $\lambda$ is implicit. From $(\ref{eq_consora})$, any minimizer $\overrightarrow v$ in $(\ref{maxtominora})$ must satisfy  
\begin{align}
    \sum_{i=1}^\ell v_i = 1 \label{eq_tocora}
\end{align}
for some $\lambda$ satisfying $(\ref{lambda_int_ora})$. Algorithm \ref{algorithmglobalmaxORA} searches for a global maximizer in $(\ref{maxORA})$ by evaluating $\overrightarrow v^-$ and $\overrightarrow v^+$ for all possible values of $\ell \leq \ell_{\operatorname{ORA}}$. For $\overrightarrow v = \overrightarrow v^+$, it can be checked that each $v_i^+$ is continuous and decreasing as a function of $\lambda \in [\lambda_{\operatorname{low}}, \lambda_{\operatorname{upp}}(\ell)]$, which makes the LHS of $(\ref{eq_tocora})$ also continuous and decreasing in $\lambda$. This proves that $S_\ell^+(\lambda_{\operatorname{low}}) \geq 1 \geq S_{\ell}^+(\lambda_{\operatorname{upp}}(\ell))$ is a necessary and sufficient condition for a unique $\lambda$ satisfying $(\ref{eq_tocora})$ for $\overrightarrow v = \overrightarrow v^+$.

Finding the roots of $S_{\ell}^-(\lambda) = 1$ is more complex. For every $\lambda \in [\lambda_{\operatorname{low}}, \lambda_{\operatorname{upp}}(\ell)]$, we have $S_\ell^+(\lambda) \geq S_\ell^-(\lambda)$. Hence, if $S_\ell^+(\lambda_{\operatorname{low}}) < 1$, we have $S_\ell^-(\lambda) \leq S_\ell^+(\lambda) < 1$
for all $\lambda \in [\lambda_{\operatorname{low}}, \lambda_{\operatorname{upp}}(\ell)]$, i.e., no roots exist. Hence, a necessary condition for the existence of solutions of $S_\ell^-(\lambda) = 1$ over $\lambda \in [\lambda_{\operatorname{low}}, \lambda_{\operatorname{upp}}(\ell)]$ is $S_\ell^+(\lambda) \geq 1$. Let $t = v_\ell^-$. Then 
\begin{align}
    t = V_{R, \theta}^-\left(\mathscr{C}_\ell(\lambda)  \right) \iff \mathscr{U}(t) \frac{d_\ell \theta  R \ln(2)}{2^R - 1} = \lambda \text{ and } V_{R, \theta}^-\left(\mathscr{C}_\ell(\lambda_{\operatorname{low}})  \right) \leq t \leq v_{\operatorname{int}}^*. \label{reparx0}
\end{align}
Note that as $\lambda$ increases from $\lambda_{\operatorname{low}}$ to $\lambda_{\operatorname{upp}}(\ell)$, we have $t$ increasing from $V_{R, \theta}^-\left(\mathscr{C}_\ell(\lambda_{\operatorname{low}})  \right)$ to $v_{\operatorname{int}}^*$. We define 
\begin{align}
    t_{\ell, \operatorname{low}} \coloneqq V_{R, \theta}^-\left(\mathscr{C}_\ell(\lambda_{\operatorname{low}})  \right) = V_{R, \theta}^-\left(\frac{2^R e^{-\theta} d_1 }{ d_\ell }   \right).
\end{align}
We reparametrize $S_{\ell}^-(\lambda) = F_{\ell}(t)$ in terms of $t$ as follows:
\begin{align}
    F_\ell(t) &= t + \sum_{i=1}^{\ell - 1} V_{R, \theta}^+\left(\mathscr{U}(t) \frac{d_\ell}{d_i} \right). \label{looki}
\end{align}
Recall from Lemma \ref{lambda_lwbdora} that for $i < \ell$ and for $t_{\ell, \operatorname{low}} \leq t \leq v_{\operatorname{int}}^*$,
\begin{align}
    2^R e^{-\theta} \leq \mathscr{U}(t) \frac{d_\ell}{d_i} \leq M_{\operatorname{int}}^* \frac{d_{\ell}}{d_i} < M_{\operatorname{int}}^*. \label{ana2nana}
\end{align}
Now define for $i < \ell$, 
\begin{align*}
  \beta_{i, \ell} &= \frac{d_\ell}{d_i} < 1,\\
  \psi_{i, \ell}(t) &= V_{R, \theta}^+\left(\mathscr{U}(t) \beta_{i, \ell} \right), 
\end{align*}
where recall from $(\ref{orderofsolsora})$ that
\begin{align*}
    \psi_{1, \ell}(t) > \cdots > \psi_{\ell-1, \ell}(t) > v_{\operatorname{int}}^* \geq t. 
\end{align*}
We can then write  
\begin{align}
    F_\ell(t) &= t + \sum_{i=1}^{\ell - 1} \psi_{i, \ell}(t), \label{reparx1}\\
    F_\ell'(t) &= 1 + \sum_{i=1}^{\ell - 1} \psi_{i, \ell}'(t) \notag \\
    &= 1 + \sum_{i=1}^{\ell - 1} \beta_{i, \ell}\frac{ \mathscr{U}'(t) }{\mathscr{U}'\left( \psi_{i, \ell}(t) \right)} \notag \\
    &= 1 + d_\ell \mathscr{U}'(t) \sum_{i=1}^{\ell - 1} \frac{1}{d_i\mathscr{U}'\left( \psi_{i, \ell}(t) \right)}. \notag 
\end{align}
First, we show that $F_\ell(t)$ extends to a real analytic function on some open interval $I_\ell \supset [t_{\ell, \operatorname{low}}, v_{\operatorname{int}}^* ]$, i.e., there exists a real analytic function $\widetilde{F}_\ell(t)$ on $I_\ell$ such that $\widetilde{F}_\ell(t) = F_\ell(t)$ for all $t \in [t_{\ell, \operatorname{low}}, v_{\operatorname{int}}^* ]$. Looking at $(\ref{looki})$, we can see that $\mathscr{U}(t)$ is analytic for all $t > 0$, and note that $t_{\ell, \operatorname{low}} > 0$. Furthermore, since $V_{R, \theta}^+$ is the inverse of $\mathscr{U}$ over the domain $[v_{\operatorname{int}}^*, 1]$ and $\mathscr{U}'(t) < 0$ for all $t \in (v_{\operatorname{int}}^*, 1]$, it follows from the inverse function theorem that $V_{R, \theta}^+$ has an analytic extension on an open interval containing $[2^R e^{-\theta}, M_{\operatorname{int}}^*)$. From $(\ref{ana2nana})$, the argument to $V_{R, \theta}^+$ in $(\ref{looki})$ lies in the interval $[2^R e^{-\theta}, M_{\operatorname{int}}^*)$ for $t \in [t_{\ell, \operatorname{low}}, v_{\operatorname{int}}^* ]$. Hence, for small enough $I_\ell \supset [t_{\ell, \operatorname{low}}, v_{\operatorname{int}}^* ]$, the argument to $V_{R, \theta}^+$ in $(\ref{looki})$ lies in an open interval containing $[2^R e^{-\theta}, M_{\operatorname{int}}^*)$ on which $V_{R, \theta}^+$ has a real analytic extension. This gives us that $F_\ell(t)$ extends to a real analytic function on $I_\ell$. Now note that   
\begin{align*}
    F_\ell(t_{\ell, \operatorname{low}}) &= t_{\ell, \operatorname{low}} + \sum_{i=1}^{\ell - 1}  V_{R, \theta}^+\left(\mathscr{U}(t_{\ell, \operatorname{low}}) \beta_{i, \ell} \right)\\
    &= t_{\ell, \operatorname{low}} + \sum_{i=1}^{\ell - 1}  V_{R, \theta}^+\left(2^R e^{-\theta} \frac{d_1}{d_i}\right) > V_{R, \theta}^+\left(2^R e^{-\theta} \right) = 1.  
\end{align*}
Hence, $F_\ell(t)$ is not identically zero on $I_\ell$. Hence, the compact interval $[t_{\ell, \operatorname{low}} , v_{\operatorname{int}}^*]$ contains finitely many solutions of $F_\ell(t) = 1$. 

For any point $(\overrightarrow v, \lambda, \overrightarrow \mu)$ satisfying $(\ref{stationarityORA}) - (\ref{monotonicity_v})$ and $(\ref{bo2ntORA2})$, the critical cone \cite[p. 330]{nonlinopt} at the point $\overrightarrow v$ is given by 
\begin{align*}
    \mathcal{C}(\overrightarrow v) &= \Bigg \{\overrightarrow w \in \mathbb{R}^K: \sum_{i=1}^K w_i = 0,\\ 
    &\quad \quad \quad \quad \quad \quad \quad  w_{\ell+1} = \cdots = w_K = 0 \Bigg \}\\
    &= \Bigg \{\overrightarrow w \in \mathbb{R}^K: \sum_{i=1}^\ell  w_i = 0,\\ 
    & \quad \quad \quad \quad \quad \quad \quad  w_{\ell+1} = \cdots = w_K = 0 \Bigg \}. 
\end{align*}
Since each KKT point $\overrightarrow v$ for the optimization problem in $(\ref{maxtominora})$ is associated with unique Lagrange multipliers $\lambda$ and $\overrightarrow \mu$ by Lemma \ref{unique_mult}, we write the critical cone as $\mathcal{C}(\overrightarrow v)$ instead of $\mathcal{C}(\overrightarrow v, \lambda, \overrightarrow \mu)$. The Lagrangian Hessian is a diagonal matrix given by  
\begin{align*}
    \frac{\partial^2 \mathcal{L}(\overrightarrow v, \lambda, \overrightarrow \mu)}{\partial v_i^2} &= - t''(v_i) d_i \\
    &= \begin{cases}
        - t''(v_i) d_i & \text{ if } 1 \leq i \leq \ell\\
        0 & \text{ if } i > \ell.
    \end{cases}
\end{align*}
From \cite[Theorem 12.5]{nonlinopt}, a second-order necessary condition for a KKT point $\overrightarrow v$ to be a local minimizer in $(\ref{maxtominora})$ is  
\begin{align}
    \sum_{i=1}^\ell - t''(v_i) d_i w_i^2 &\geq 0\\
    \iff \sum_{i=1}^\ell t''(v_i) d_i w_i^2 &\leq 0 \label{second_order_condora}
\end{align}
for all $(w_1, \ldots, w_\ell)$ satisfying 
\begin{align*}
    \sum_{i=1}^\ell w_i = 0. 
\end{align*}
A sufficient condition is when the inequality $(\ref{second_order_condora})$ is strict. Note that 
\begin{align*}
    t''(v) &= \frac{\theta  R \ln (2)}{2^R - 1}  \frac{2^{R/v} t(v)}{v^4} \left(R \ln (2) \left(\frac{\theta  2^{R/v}}{2^R-1}-1\right)-2 v \right)\\
    &= \mathscr{U}'(v) \frac{\theta  R \ln (2)}{2^R - 1}. 
\end{align*}
We can then rewrite  $(\ref{second_order_condora})$ as
\begin{align*}
    \sum_{i=1}^\ell t''(v_i) d_i w_i^2 &\leq 0\\
    \iff \sum_{i=1}^\ell \mathscr{U}'(v_i) d_i w_i^2 &\leq 0.
\end{align*}
For $\overrightarrow v = (v_1^+, \ldots, v_{\ell-1}^+, v_{\ell}^-, 0, \ldots, 0)$ and using the reparameterization given in $(\ref{reparx0}) - (\ref{reparx1})$, the above can be written as  
\begin{align*}
    \mathscr{U}'(t) d_\ell w_\ell^2 + \sum_{i=1}^{\ell-1} \mathscr{U}'(\psi_{i, \ell}(t)) d_i w_i^2 &\leq 0. 
\end{align*}
Then since $w_\ell = - \sum_{i=1}^{\ell - 1} w_i$, the second-order necessary condition becomes 
\begin{align}
    a \left( \sum_{i=1}^{\ell - 1} w_i\right)^2 &\leq \sum_{i=1}^{\ell-1} b_i w_i^2 \label{,.7432}
\end{align}
for all $w_1, \ldots, w_{\ell-1}$, where we define $a = \mathscr{U}'(t) d_\ell$ and $b_i = - \mathscr{U}'(\psi_{i, \ell}(t)) d_i > 0$. Then by Cauchy–Schwarz inequality,  
\begin{align*}
    \left( \sum_{i=1}^{\ell - 1} w_i\right)^2 \leq  \left( \sum_{i=1}^{\ell - 1} b_i w_i^2\right) \left( \sum_{i=1}^{\ell - 1} \frac{1}{b_i}\right).  
\end{align*}
Hence, if 
\begin{align}
    a\left( \sum_{i=1}^{\ell - 1} \frac{1}{b_i}\right) \leq 1, \label{alsonece} 
\end{align}
then $(\ref{,.7432})$ holds for all $w_1, \ldots, w_{\ell -1}$. But $(\ref{alsonece})$ is also necessary because if $(\ref{alsonece})$ does not hold, then choosing $w_i = 1/b_i$ in $(\ref{,.7432})$ would violate the inequality. Finally, we note that $(\ref{alsonece})$ is just $F_\ell'(t) \geq 0$. This completes the proof of Theorem \ref{global_maximizer_thmORA}.

To prove Theorem \ref{local_maximizer_thmORA}, note that $\overrightarrow v^+$ already satisfies $(\ref{gf3ORA})$ and $(\ref{lspecificationORA})$ as well as the KKT conditions. Hence, it suffices to show that $\overrightarrow v^+$ is a strict local maximizer in $(\ref{maxORA})$. For $\overrightarrow v = \overrightarrow v^+$, each $v_i^+ \geq v_{\operatorname{int}}^*$ for $1 \leq i \leq \ell$ and $v_1^+ > v_{\operatorname{int}}^*$. It can be checked that $t''(v) < 0$ for $v > v_{\operatorname{int}}^*$. Hence, the inequality in $(\ref{second_order_condora})$ is a strictly inequality for $\overrightarrow v = \overrightarrow v^+$ and $\overrightarrow w \neq \overrightarrow 0$. Hence, from Theorem \cite[Theorem 12.6]{nonlinopt}, $\overrightarrow v = \overrightarrow v^+$ is a strict local minimizer in $(\ref{maxtominora})$, equivalently, a strict local maximizer in $(\ref{maxORA})$.

\appendices

\section{Proof of Lemma \ref{power_cons_satisfied} \label{power_cons_satisfied_proof}}

If treat the random vectors as elements of the Hilbert space $\mathcal{H} := L^2(\Omega; \mathbb{C}^n)$ with the inner product 
$$\left \langle \mathbf{X}, \mathbf{Y} \right \rangle_\mathcal{H} \coloneqq \mathbb{E}\left [ \left \langle \mathbf{X}, \mathbf{Y} \right \rangle \right] = \mathbb{E}\left [\sum_{i=1}^n X_i \overline{Y}_i \right],$$ 
then we have

\begin{align*}
    \mathbb{E}\left[ \left | \left | \mathbf{X} \right | \right |^2 \right] &= \left \langle \mathbf{X},\mathbf{X} \right \rangle_{\mathcal{H}}\\
    &=  \left \langle \sum_{i=1}^K  \sqrt{\alpha_i P} \mathbf{X}_i, \sum_{j=1}^K \sqrt{ \alpha_j P } \mathbf{X}_j \right \rangle_\mathcal{H}\\
    &= \sum_{i=1}^K \sqrt{ \alpha_i P }  \left \langle  \mathbf{X}_i, \sum_{j=1}^K  \sqrt{\alpha_j P }  \mathbf{X}_j \right \rangle_\mathcal{H} \\
    &=  \sum_{i=1}^K \sum_{j=1}^K P \sqrt{\alpha_j \alpha_i} \left \langle  \mathbf{X}_i,  \mathbf{X}_j \right \rangle_\mathcal{H} \\
    &= \sum_{i=1}^K \sum_{j=1}^K P \sqrt{\alpha_j \alpha_i} \left[ \sum_{t=1}^n \mathbb{E} [ X_{i, t} \overline{X_{j, t}}] \right]  \\
    &\stackrel{(a)}{=} \sum_{i=1}^K P \alpha_i \left[ \sum_{t=1}^n \mathbb{E} [ |X_{i, t}|^2] \right]\\
    &= nP \sum_{i=1}^K \alpha_i = nP,
\end{align*}
where the equality $(a)$ above uses the fact that codebooks for $i \neq j$ are independently generated.

\section{Proof of Lemma \ref{error_exp_iid} \label{error_exp_iid_proof}}

The starting point of the proof is Lemma \ref{gallager_bound} below which can be found in \cite[(11)]{gallager_paper_65} or \cite[Theorem 3]{ppv}, although the latter does not explicitly mention the random coding argument used to derive the bound $(\ref{x3})$. 
\begin{lemma}
    Consider any channel $W(\cdot|\cdot)$ from $\mathbb{C}^n$ to $\mathbb{C}^n$. Consider a random channel code $(\operatorname{f}_n, \operatorname{g}_n)$ with codebook size $2^{nR}$ such that the channel input $\mathbf{X} \sim \overline{P}$ for some distribution $\overline{P} \in \mathcal{P}(\mathbb{C}^n)$. Then for any $n, \lambda \in [0, 1]$ and $R$, the ensemble average error probability of this code is upper bounded by  
    \begin{align}
        2^{n\lambda R}\, \mathbb{E}\left [ \left( \mathbb{E}\left [ \exp\left(\frac{1}{1 + \lambda} \ln \left(\frac{W(\mathbf{Y}|\overline{\mathbf{X}})}{\overline{P}W(\mathbf{Y})} \right)  \right) \Bigg | \mathbf{Y} \right] \right)^{1 + \lambda} \right ], \label{x3}
    \end{align}
    where $\overline{\mathbf{X}} \sim \overline{P}$, $\mathbf{Y} \sim \overline{P}W$ and $\overline{\mathbf{X}} \perp \!\!\! \perp \mathbf{Y}$. 
    \label{gallager_bound}
\end{lemma}

Choosing $\overline{P} = \mathcal{C} \mathcal{N}(\mathbf{0}, \rho \mathbf{I}_n)$, $W(\cdot|\mathbf{x}) = \mathcal{C} \mathcal{N}(\mathbf{x}, \mathbf{I}_n)$ and $\mathbf{Y} \sim \mathcal{C} \mathcal{N}(\mathbf{0}, (1 + \rho) \mathbf{I}_n )$ in Lemma \ref{gallager_bound}, we first evaluate 
\begin{align*}
    &\mathbb{E}\left [ \exp\left(\frac{1}{1 + \lambda} \ln \frac{W(\mathbf{Y}|\overline{\mathbf{X}})}{\overline{P}W(\mathbf{Y})}  \right) \Bigg | \mathbf{Y} = \mathbf{y} \right]\\
    &= \mathbb{E}\left [ \exp \left( \frac{1}{1 + \lambda } \left [ \ln \left( \frac{1}{\pi^n} \right) - ||\mathbf{y} - \overline{\mathbf{X}}||^2 - \ln \left(\frac{1}{\pi^n(\rho + 1)^n} \right) + \frac{||\mathbf{y}||^2}{\rho + 1} \right] \right) \Bigg | \mathbf{Y} = \mathbf{y} \right ]\\
    &= \mathbb{E}\left [ \exp \left(   \ln \left(1 + \rho \right)^{\frac{n}{1 + \lambda}} - \frac{1}{1 + \lambda} ||\mathbf{y} - \overline{\mathbf{X}}||^2  + \frac{||\mathbf{y}||^2}{(1 + \rho)(1 + \lambda)} \right) \Bigg | \mathbf{Y} = \mathbf{y} \right ].
\end{align*}
We have 
\begin{align*}
    &\exp \left(   \ln \left(1 + \rho \right)^{\frac{n}{1 + \lambda}} - \frac{1}{1 + \lambda} ||\mathbf{y} - \overline{\mathbf{X}}||^2  + \frac{||\mathbf{y}||^2}{(1 + \rho)(1 + \lambda)} \right)\\
    &=\left(1 + \rho \right)^{\frac{n}{1 + \lambda}} \exp \left(\frac{||\mathbf{y}||^2}{(1 + \rho)(1 + \lambda)} \right) \exp \left( - \frac{1}{1 + \lambda} ||\mathbf{y} - \overline{\mathbf{X}}||^2  \right) 
\end{align*}
and 
\begin{align*}
    \mathbb{E}\left [ \exp \left( - \frac{1}{1 + \lambda} ||\mathbf{y} - \overline{\mathbf{X}}||^2  \right) \right] &= \prod_{i=1}^n \mathbb{E}\left[ \exp\left( -\frac{1}{1 + \lambda} \big |\overline{X}_i- y_i \big |^2 \right) \right]\\
    &= \prod_{i=1}^n \frac{(\lambda +1) \exp\left(-\frac{|y_i|^2}{\lambda +\rho +1}\right)  }{\lambda +\rho +1}\\
    &= \left(\frac{1 + \lambda}{1 + \lambda + \rho} \right)^n \exp \left(- \frac{||\mathbf{y}||^2}{1 + \lambda + \rho} \right).
\end{align*}
So we have 
\begin{align*}
    &\mathbb{E}\left [ \exp\left(\frac{1}{1 + \lambda} \ln \frac{W(\mathbf{Y}|\overline{\mathbf{X}})}{\overline{P}W(\mathbf{Y})}  \right) \Bigg | \mathbf{Y} \right]\\
    &= \left(1 + \rho \right)^{\frac{n}{1 + \lambda}} \exp \left(\frac{||\mathbf{Y}||^2}{(1 + \rho)(1 + \lambda)} \right) \left(\frac{1 + \lambda}{1 + \lambda + \rho} \right)^n \exp \left(- \frac{||\mathbf{Y}||^2}{1 + \lambda + \rho} \right)\\
    &= \left(1 + \rho \right)^{\frac{n}{1 + \lambda}} \left(\frac{1 + \lambda}{1 + \lambda + \rho} \right)^n \exp \left(-\frac{\lambda  \rho ||\mathbf{Y}||^2 }{(\lambda +1) (\rho +1) (\lambda +\rho +1)} \right).
\end{align*}
Then 
\begin{align*}
    &\left(\mathbb{E}\left [ \exp\left(\frac{1}{1 + \lambda} \ln \frac{W(\mathbf{Y}|\overline{\mathbf{X}})}{\overline{P}W(\mathbf{Y})}  \right) \Bigg | \mathbf{Y} \right]\right)^{1 + \lambda}\\
    &=  \left(1 + \rho \right)^{n} \left(\frac{1 + \lambda}{1 + \lambda + \rho} \right)^{n(1 + \lambda)} \exp \left(-\frac{\lambda  \rho ||\mathbf{Y}||^2 }{ (\rho +1) (\lambda +\rho +1)} \right).
\end{align*}
Then we take the outer expectation w.r.t. $\mathbf{Y} \sim \mathcal{CN}(\mathbf{0}, (1 + \rho) \mathbf{I}_n)$ to obtain 
\begin{align*}
    &\mathbb{E}\left [ \exp \left(-\frac{\lambda  \rho ||\mathbf{Y}||^2 }{ (\rho +1) (\lambda +\rho +1)} \right) \right ]\\
    &= \prod_{i=1}^n \mathbb{E}\left [ \exp \left(-\frac{\lambda  \rho |Y_i|^2 }{ (\rho +1) (\lambda +\rho +1)} \right) \right ]\\
    &= \left( \frac{\lambda +\rho +1}{\lambda  \rho +\lambda +\rho +1}\right)^n,
\end{align*}
giving us 
\begin{align*}
    &2^{n \lambda R}\,\mathbb{E}\left [ \left(\mathbb{E}\left [ \exp\left(\frac{1}{1 + \lambda} \ln \frac{W(\mathbf{Y}|\overline{\mathbf{X}})}{\overline{P}W(\mathbf{Y})}  \right) \Bigg | \mathbf{Y} \right]\right)^{1 + \lambda} \right]\\
    &= e^{n \lambda R_{\operatorname{nats}}}\, \left(1 + \rho \right)^{n} \left(\frac{1 + \lambda}{1 + \lambda + \rho} \right)^{n(1 + \lambda)} \left( \frac{\lambda +\rho +1}{\lambda  \rho +\lambda +\rho +1}\right)^n\\
    &= e^{n \lambda R_{\operatorname{nats}}}   \left(\frac{1 + \lambda}{1 + \lambda + \rho} \right)^{n \lambda}\\
    &= \exp \left(n \lambda R \ln(2) - n\lambda \ln \left(\frac{1 + \lambda + \rho}{1 + \lambda} \right) \right)\\
    &= \exp \left(-n \left [ \lambda \ln \left(1 + \frac{\rho}{1 + \lambda} \right) - \lambda R \ln(2) \right] \right). 
\end{align*}

\section{Proof of Lemma \ref{lemma_happy_CLT} \label{lemma_happy_CLT_proof}}

The starting point of the proof is Shannon's achievability bound \cite[Theorem 1]{SHANNON19576}, presented as Lemma \ref{stillaaron} below. A variant of Lemma \ref{stillaaron} for random feedback codes can be found in \cite[Lemma 14]{9099482}.  

\begin{lemma}
    Consider any channel $W(\cdot|\cdot)$ from $\mathbb{C}^n$ to $\mathbb{C}^n$. Consider a random channel code $(\operatorname{f}_n, \operatorname{g}_n)$ with codebook size $2^{nR}$ such that the channel input $\mathbf{X} \sim \overline{P}$ for some distribution $\overline{P} \in \mathcal{P}(\mathbb{C}^n)$. Then for any $n, \kappa$ and $R$, the ensemble average error probability of the code is upper bounded by  
    \begin{align*}
        \mathbb{P}\left(\frac{1}{n} \log \left(\frac{W(\mathbf{Y} | \mathbf{X})}{\overline{P}W(\mathbf{Y})} \right)\leq R + \kappa \right) + 2^{-n \kappa}, 
    \end{align*}
    where $(\mathbf{X}, \mathbf{Y}) \sim \overline{P} \circ W$ and $\mathbf{Y} \sim \overline{P}W$.
    \label{stillaaron}
\end{lemma}

We now specialize the result in Lemma \ref{stillaaron} to our case by letting $\mathbf{X} \sim \overline{P} = \mathcal{C} \mathcal{N}(\mathbf{0}, \rho \mathbf{I}_n)$, $\mathbf{Z} \sim \mathcal{C} \mathcal{N}(\mathbf{0},\mathbf{I}_n)$, $\mathbf{X} \perp \!\!\! \perp \mathbf{Z}$, $\mathbf{Y} = \mathbf{X} + \mathbf{Z}$ and $W$ be the channel from $\mathbf{X}$ to $\mathbf{Y}$. Then $W(\cdot|\mathbf{x}) = \mathcal{C} \mathcal{N}(\mathbf{x}, \mathbf{I}_n)$ and $\mathbf{Y} \sim \mathcal{C} \mathcal{N}(0, (1 + \rho) \mathbf{I}_n )$. We have
\begin{align*}
    W(\mathbf{y}|\mathbf{x}) & = \frac{1}{\pi^n} \exp \left( - ||\mathbf{y} - \mathbf{x}||^2 \right) \\
    \overline{P}W(\mathbf{y}) & = \frac{1}{\pi^n (1 +\rho)^{n}} \exp \left( -\frac{||\mathbf{y}||^2}{1 + \rho} \right).
\end{align*} 
Hence, 
\begin{align*}
    &\ln \left(\frac{W(\mathbf{Y} | \mathbf{X})}{\overline{P}W(\mathbf{Y})} \right) \notag \\
    &= \ln \left( \frac{1}{\pi^n} \right) - ||\mathbf{Y} - \mathbf{X}||^2 - \ln \left(\frac{1}{\pi^n(\rho + 1)^n} \right) + \frac{||\mathbf{Y}||^2}{\rho + 1} \notag \\
    &= n\ln \left(1 + \rho  \right) - ||\mathbf{Z}||^2  + \frac{||\mathbf{X} + \mathbf{Z} ||^2}{\rho + 1} \notag  \\
    &= n\ln \left(1 + \rho  \right) - \sum_{i=1}^n |Z_i|^2 + \frac{1}{1 + \rho}\sum_{i=1}^n |X_i + Z_i|^2\\
    &= n\ln \left(1 + \rho  \right) + \sum_{i=1}^n \left[ \frac{1}{1 + \rho} |X_i|^2 - \frac{\rho}{1 + \rho} |Z_i|^2 + \frac{2}{1 + \rho} \operatorname{Re}\left\{ X_iZ_i^* \right \}  \right ] \\
    &= n\ln \left(1 + \rho  \right) + \sum_{i=1}^n \zeta_i,
\end{align*}
where the $\zeta_i$'s are i.i.d. RVs. Let $\zeta_i \stackrel{d}{=} \zeta$. Let $U, V \sim \mathcal{CN}(0, 1)$ and $U \ind V$. Let $\overrightarrow g = (U, V)^T$ so that $\overrightarrow g \in \mathbb{C}^2$. Let 
\begin{align*}
    \mathbf{A} = \frac{1}{1 + \rho} \begin{bmatrix}
        \rho & \sqrt{\rho}\\
        \sqrt{\rho} & -\rho
    \end{bmatrix}.
\end{align*}
Then since $\mathbf{A}$ is Hermitian, we can write $\mathbf{A} = \mathbf{Q}^H \mathbf{D} \mathbf{Q}$, where $\mathbf{Q}$ is unitary, $\mathbf{D} = \operatorname{diag}(\lambda, -\lambda)$  and 
\begin{align*}
    \lambda &= \sqrt{\frac{\rho}{1 + \rho}}. 
\end{align*}
Then 
\begin{align*}
    \zeta &\stackrel{d}{=} \frac{\rho}{1 + \rho} \left( |U|^2 - |V|^2 \right) + \frac{2 \sqrt{\rho}}{1 + \rho} \operatorname{Re} \left \{ U V^* \right \}\\
    &= \overrightarrow g^H \mathbf{A} \overrightarrow g\\
    &= \overrightarrow g^H \mathbf{Q}^H \mathbf{D} \mathbf{Q} \overrightarrow g\\
    &\stackrel{d}{=} \overrightarrow g^H \mathbf{D} \overrightarrow g\\
    &= \lambda \left( |U|^2 - |V|^2 \right)\\
    &\stackrel{d}{=} T,
\end{align*}
where $T \sim \operatorname{Laplace}(0, \lambda)$. Then it can be checked that  
\begin{align*}
    \mathbb{E}[\zeta] &= 0,\\
    \operatorname{Var}(\zeta) &= \frac{2 \rho}{1 + \rho},\\
    \mathbb{E}\left [ |\zeta|^3 \right] &= 6\left(\frac{\rho}{1 + \rho} \right)^{3/2}.
\end{align*}
Recall that 
\begin{align*}
    \log \left(\frac{W(\mathbf{Y} | \mathbf{X})}{\overline{P}W(\mathbf{Y})} \right) &= n\log \left(1 + \rho  \right) + \log(e) \sum_{i=1}^n \zeta_i. 
\end{align*}
Then we apply the Berry-Esseen Theorem for identically distributed summands to obtain 
\begin{align}
    &\mathbb{P}\left( \log \left(\frac{W(\mathbf{Y} | \mathbf{X})}{\overline{P}W(\mathbf{Y})} \right)\leq nR + n\kappa \right) + 2^{-n \kappa}\\
    &= \mathbb{P}\left( n\log \left(1 + \rho  \right) + \log(e) \sum_{i=1}^n \zeta_i \leq nR + n\kappa \right) + 2^{-n \kappa}\\
    &= \mathbb{P}\left( \sum_{i=1}^n \zeta_i \leq \frac{1}{\log(e)}\left( nR - n\log \left(1 + \rho  \right) + n\kappa \right) \right) + 2^{-n \kappa}\\
    &\leq \Phi \left( \frac{\sqrt{1 + \rho}}{\log(e) \sqrt{2 \rho n}}\left( nR - n\log \left(1 + \rho  \right) + n\kappa \right) \right) + \frac{C_{\operatorname{BE}}}{\sqrt{n}} \frac{3}{\sqrt{2}} +  2^{-n \kappa} \\
    &\leq \Phi \left( \frac{ \sqrt{1 + \rho} \sqrt{n}}{\log(e) \sqrt{2 \rho}}  \left( R - \log \left(1 + \rho  \right) + \kappa \right) \right) + \frac{1}{\sqrt{n}} +  2^{-n \kappa}. \label{e6}
\end{align}
In the last inequality above, we used the bound $C_{\operatorname{BE}} \leq 0.4690$ \cite{Shevtsova2013}. 
Since the above bound holds for any $\kappa$, we choose $\kappa = \frac{\log n}{2n}$ to upper bound $(\ref{e6})$ by  
\begin{align*}
    &\Phi \left( \frac{ \sqrt{1 + \rho} \sqrt{n}}{\log(e) \sqrt{2 \rho}}  \left( R - \log \left(1 + \rho  \right) + \frac{\log n}{2n} \right) \right) +  \frac{2}{\sqrt{n}}\\
    &= \Phi\left( \frac{\sqrt{n}\left(R - C(\rho) \right) + \frac{\log n}{2\sqrt{n}}}{\sqrt{V_{\operatorname{tot}}(\rho)}} \right) + \frac{2}{\sqrt{n}}. 
\end{align*}

\section{Proof of Theorem \ref{finitetofirstorderreduction} \label{finitetofirstorderreduction_proof}}

First, consider any sequence $\overrightarrow \alpha^{(n)}$ that converges to some $\overrightarrow \alpha \in \Delta^{K-1}$. Let $\beta_i^{(n)} = \sum_{j=i + 1}^K \alpha_j^{(n)}$ and $\beta_i = \sum_{j=i + 1}^K \alpha_j$. Define 
\begin{align*}
    \rho_j^{(n)}(\gamma) &= \frac{\gamma \alpha_j^{(n)} P}{1 + \gamma P \beta_j^{(n)}},\\
    \rho_j(\gamma) &= \frac{\gamma \alpha_j P}{1 + \gamma P \beta_j}
\end{align*}
so that $\rho_j^{(n)}(\gamma) \to \rho_j(\gamma)$. Based on $(\ref{firstorderapproxRconstant})$, it can be checked that for every $\gamma \geq 0$ such that $\rho_j(\gamma) \neq 2^R - 1$, 
\begin{align*}
    \lim_{n \to \infty} \mathcal{E}\left(n, R, \rho^{(n)}_j(\gamma) \right) = \mathds{1}(R > \log(1 + \rho_j(\gamma))).
\end{align*}
Since $\mathbb{P}(\rho_j(\gamma) = 2^R - 1) = 0$, we can argue by dominated convergence theorem that 
    \begin{align}
        \lim_{n \to \infty} G_n(\overrightarrow \alpha^{(n)}) &= \sum_{i=1}^K \mathbb{E}_{\gamma}\left [\prod_{j=1}^i \left(1 - \lim_{n \to \infty} \mathcal{E}\left(n, R, \rho_j^{(n)}(\gamma) \right) \right) \right] d_i\\
        &= \sum_{i=1}^K \mathbb{E}_{\gamma}\left [\prod_{j=1}^i \mathds{1}(R < \log(1 + \rho_j(\gamma))) \right] d_i  \label{gnlim}
    \end{align}

Then note that 
\begin{align}
   R &< \log \left(1 + \frac{\gamma \alpha_j P}{1 + \gamma P \beta_j}  \right) \label{6b}  
\end{align}
is equivalent to 
\begin{align}
    \alpha_j -(2^{R} - 1) \beta_j > 0 \quad \text{ and } \quad \gamma > \frac{2^{R} - 1}{P(\alpha_j -(2^{R} - 1) \beta_j)}, \label{m3x}
\end{align}
which in turn is equivalent to $\gamma > \tau_j$. Hence, we obtain 
\begin{align}
        \lim_{n \to \infty} G_n(\overrightarrow \alpha^{(n)}) &= \sum_{i=1}^K \mathbb{E}_{\gamma}\left [\prod_{j=1}^i \mathds{1}\left( \gamma > \tau_j \right) \right] d_i \\
        &= \sum_{i=1}^K \mathbb{E}_{\gamma}\left [\mathds{1}\left(\gamma > \max \{\tau_1, \ldots, \tau_i \} \right) \right] d_i\\
        &= \sum_{i=1}^K \mathbb{P}\left(\gamma > \max \{\tau_1, \ldots, \tau_i \} \right) d_i\\
        &= \sum_{i=1}^K \exp \left( - \frac{\max \{\tau_1, \ldots, \tau_i \}}{\sigma^2} \right) d_i, \label{choosealkar}
    \end{align}
where the last equality above follows since $\gamma = |H|^2$ is exponential with mean $\sigma^2$, given that $H \sim \mathcal{CN}(0, \sigma^2)$. 

Now we let $\overrightarrow \alpha^{(n)} \in \Delta^{K-1}$ be a maximizer (it exists because $\Delta^{K-1}$ is compact and $G_n$ is continuous) in 
\begin{align*}
    \max_{\overrightarrow \alpha \in \Delta^{K - 1}}\,G_n(\overrightarrow \alpha). 
\end{align*}
Since $\Delta^{K-1}$ is compact, we can assume that $\overrightarrow \alpha^{(n_m)} \to \overrightarrow \alpha$ for some $\overrightarrow \alpha \in \Delta^{K-1}$ by passing down to a convergent subsequence $\overrightarrow \alpha^{(n_m)}$ which additionally satisfies
\begin{align*}
    \limsup_{n \to \infty}  \max_{\overrightarrow \alpha \in \Delta^{K - 1}}\,G_n(\overrightarrow \alpha) = \lim_{m \to \infty} G_{n_m}(\overrightarrow \alpha^{(n_m)}).
\end{align*}
Hence, from $(\ref{choosealkar})$, we have 
\begin{align}
    \limsup_{n \to \infty}  \max_{\overrightarrow \alpha \in \Delta^{K - 1}}\,G_n(\overrightarrow \alpha) \leq \max_{\overrightarrow \alpha \in \Delta^{K - 1}}  \sum_{i=1}^K \exp \left( - \frac{\max \{\tau_1, \ldots, \tau_i \}}{\sigma^2} \right) d_i. \label{m1-}
\end{align}

Now choose $\overrightarrow \alpha$ to be a maximizer in 
\begin{align*}
    \max_{\overrightarrow \alpha \in \Delta^{K - 1}} \sum_{i=1}^K \exp \left( - \frac{\max \{\tau_1, \ldots, \tau_i \}}{\sigma^2} \right) d_i. 
\end{align*}
Then assuming a constant sequence $\overrightarrow \alpha^{(n)} = \overrightarrow \alpha$ and invoking $(\ref{choosealkar})$ again, we obtain
\begin{align}
    \liminf_{n \to \infty} \max_{\overrightarrow \alpha \in \Delta^{K - 1}}\,G_n(\overrightarrow \alpha) \geq \max_{\overrightarrow \alpha \in \Delta^{K - 1}} \sum_{i=1}^K \exp \left( - \frac{\max \{\tau_1, \ldots, \tau_i \}}{\sigma^2} \right) d_i. \label{m2-}
\end{align}
Combining $(\ref{m1-})$ and $(\ref{m2-})$ establishes the result.

\section{Proof of $(\ref{opt_prop_1})$ in Theorem \ref{opt_properties_theorem} \label{opt_prop_1_proof}}

Define 
\begin{align*}
   f(\alpha_1, \ldots, \alpha_K) =  \sum_{i=1}^K  \exp \left(- \frac{\tau_i}{\sigma^2}  \right) d_i.
\end{align*}

We prove by contradiction. Let  $(\alpha_1, \ldots, \alpha_K)$ be an optimal solution such that $\tau_i > \tau_{i+1}$ for some $1 \leq i  \leq K - 1$. We divide the proof into two cases: $(1)$ $\tau_i = +\infty$ and $(2)$ $\tau_i < \infty$

\subsection{Case 1}

Since $\tau_{i + 1} < \tau_i = \infty$, we must have $\alpha_{i + 1} > (2^R - 1) \beta_{i + 1}$. Consider a modified feasible point $(\alpha_1', \ldots, \alpha_K')$ with $\alpha_i' = \alpha_i + \alpha_{i + 1}$, $\alpha_{i + 1}' = 0$ and $\alpha_j' = \alpha_j$ for all $j \neq i, i + 1$.  Note that $\tau_j' = \tau_j$ for all $j \neq i$ and $j \neq i + 1$. Also note that $\alpha_i' > (2^R - 1)\beta_i'$ since $\beta_i' = \beta_{i + 1}$ and $\alpha_i \geq 0$. Therefore, $\tau_{i + 1}' = \tau_i = +\infty$,
\begin{align*}
    \tau_{i + 1} &= \frac{2^R - 1}{P (\alpha_{i + 1} -(2^{R} - 1) \beta_{i + 1})}, \\
    \tau_i' &= \frac{2^{R} - 1}{P(\alpha_i' -(2^{R} - 1) \beta_i')}\\
    &= \frac{2^{R} - 1}{P(\alpha_i + \alpha_{i + 1} -(2^{R} - 1) \beta_{i + 1})}.
\end{align*}
Hence, 
\begin{align*}
     &f(\alpha_1, \ldots, \alpha_K) - f(\alpha_1', \ldots, \alpha_K')\\
     &= \left [ \exp \left(- \frac{\tau_i}{\sigma^2}  \right) - \exp \left(- \frac{\tau_i'}{\sigma^2}  \right) \right] d_i + \left [ \exp \left(- \frac{\tau_{i + 1}}{\sigma^2}  \right) - \exp \left(- \frac{\tau_{i + 1}'}{\sigma^2}  \right) \right] d_{i + 1} \\
    &=  - \exp \left(-\frac{2^{R} - 1}{\sigma^2 P(\alpha_i + \alpha_{i + 1} -(2^{R} - 1) \beta_{i + 1})}  \right) d_i +  \exp \left(- \frac{2^R - 1}{\sigma^2 P (\alpha_{i + 1} -(2^{R} - 1) \beta_{i + 1})}  \right)   d_{i + 1}\\ 
    &\leq - \exp \left(-\frac{2^{R} - 1}{\sigma^2 P( \alpha_{i + 1} -(2^{R} - 1) \beta_{i + 1})}  \right) d_i +  \exp \left(- \frac{2^R - 1}{\sigma^2 P (\alpha_{i + 1} -(2^{R} - 1) \beta_{i + 1})}  \right)   d_{i + 1}\\
    &< 0
\end{align*}
since $d_i > d_{i  +1}$. Hence, $(\alpha_1, \ldots, \alpha_K)$ cannot be optimal. 

\subsection{Case 2}

Define 
\begin{align*}
\theta &= \frac{2^R - 1}{\sigma^2 P} > 0,\\
    x &= \alpha_i - (2^R - 1) \beta_i,\\
    y &= \alpha_{i + 1} - (2^R - 1) \beta_{i + 1}.
\end{align*}
If $\tau_{i + 1} < \tau_i < \infty$, then $y > x > 0$. Consider a modified feasible point 
\begin{align*}
     (\alpha_1', \ldots, \alpha_K') = (\alpha_1, \ldots, \alpha_i + \delta, \alpha_{i+1} - \delta, \ldots, \alpha_K),  
\end{align*}
where $|\delta| \in \left(0, \min \{ \alpha_i, \alpha_{i + 1} \}\right)$ is chosen small enough so that $\alpha_{j} - |\delta| > (2^R - 1) \beta_{j}$, for $j = i, i + 1$. Note that $\tau_j' = \tau_j$ for all $j \neq i$ and $j \neq i + 1$. Define 
\begin{align}
    \Delta(\delta) &\coloneqq \frac{\partial f(\alpha_1', \ldots, \alpha_K') }{\partial \delta}.
\end{align}
We have  
\begin{align*}
    \tau_i' &= \frac{2^R - 1}{P(\alpha_i - (2^R - 1) \beta_i + 2^R \delta )} \\
    \tau_{i + 1}' &= \frac{2^R - 1}{P(\alpha_{i  +1}  - (2^R - 1) \beta_{i + 1} - \delta)}\\
    \Delta(\delta) &= \frac{\partial }{\partial \delta} \left [ \exp \left(- \frac{\tau_i'}{\sigma^2}  \right) d_i  + \exp \left(- \frac{\tau_{i + 1}'}{\sigma^2}  \right) d_{i + 1}\right]\\
    &= \frac{\partial }{\partial \delta} \left [ \exp\left(-  \frac{\theta}{\alpha_i - (2^R - 1) \beta_i + 2^R \delta }\right) d_i + \exp \left(- \frac{\theta}{\alpha_{i  +1}  - (2^R - 1) \beta_{i + 1} - \delta}  \right) d_{i+1}  \right]\\
    &= \frac{\theta 2^R}{(\alpha_i - (2^R - 1) \beta_i + 2^R \delta)^2} \exp\left(-  \frac{\theta}{\alpha_i - (2^R - 1) \beta_i + 2^R \delta }\right) d_i - \mbox{} \\
    & \quad \quad \quad \quad \frac{\theta}{(\alpha_{i  +1}  - (2^R - 1) \beta_{i + 1} - \delta)^2}\exp \left(- \frac{\theta}{\alpha_{i  +1}  - (2^R - 1) \beta_{i + 1} - \delta}  \right) d_{i+1}, \\
    &= \theta \left(\frac{2^R}{(x + 2^R \delta)^2} \exp\left(-  \frac{\theta}{x + 2^R \delta }\right) d_i - \frac{1}{(y - \delta)^2}\exp \left(- \frac{\theta}{y - \delta}  \right) d_{i+1} \right),\\
    \Delta(0) &= \theta \left( \frac{ 2^R}{x^2} \exp\left(- \frac{\theta}{x} \right) d_i - \frac{1}{y^2} \exp\left(- \frac{\theta}{y} \right) d_{i + 1} \right).
\end{align*} 
Since $(\alpha_1, \ldots, \alpha_K)$ is optimal with $\alpha_i, \alpha_{i + 1} > 0$, we must have $\Delta(0) = 0$. But we will show that the second derivative $\Delta'(0) > 0$ so that $(\alpha_1, \ldots, \alpha_K)$ cannot be the maximizer. Indeed, 
\begin{align}
    \Delta'(0) &= \theta \left( \left [ \frac{ \theta  - 2x  }{x^4} \right] 2^{2R} \exp\left(-\frac{\theta}{x} \right) d_i+\left [ \frac{\theta - 2y  }{y^4} \right] \exp\left(-\frac{\theta}{y} \right) d_{i + 1} \right). \label{ht}
\end{align}
Given $\Delta(0) = 0$, we have 
\begin{align}
     &\frac{ 2^R}{x^2} \exp\left(- \frac{\theta}{x} \right) d_i = \frac{1}{y^2} \exp\left(- \frac{\theta}{y} \right) d_{i + 1} \notag \\
     \implies & d_{i} =  \frac{x^2}{2^R y^2} \exp\left( \frac{\theta}{x}- \frac{\theta}{y} \right) d_{i + 1}. \label{ratiob}
\end{align}
Substituting this in $(\ref{ht})$, we obtain 
\begin{align}
    \Delta'(0) &= \frac{\theta d_{i + 1}}{y^4} \exp\left(-\frac{\theta}{y} \right) \left( \left (  \theta  - 2x \right)   \frac{2^{R} y^2}{ x^2}  +\theta - 2y    \right)\\
    &= \frac{\theta d_{i + 1}}{y^4} \exp\left(-\frac{\theta}{y} \right) \left( \theta \left(1 + \frac{2^{R} y^2}{ x^2}  \right)     -  \frac{2^{R + 1} y^2}{ x}  - 2y    \right). \label{deltaprime0}
\end{align}
The sign of $\Delta'(0)$ is determined by the sign of the expression inside the parentheses above. Since $d_i > d_{i + 1}$, $(\ref{ratiob})$ implies that 
\begin{align}
    \frac{2^R y^2}{x^2} < \exp\left( \theta \left( \frac{1}{x} - \frac{1}{y} \right)   \right). \label{ygx}
\end{align}
Define $t \coloneqq \frac{y}{x}$ and $u = 2^R t^2$. Note that $t, u > 1$ and $u > t^2$ since $R > 0$ and $y > x > 0$. From $(\ref{ygx})$, we have 
\begin{align}
    &\exp\left( \theta \left( \frac{1}{x} - \frac{1}{y} \right)   \right) > u\\
    \implies \theta &> \frac{\ln(u)}{y - x} x y  \\
    &= \frac{\ln(u)}{t- 1} x t. \label{kappalower}
\end{align}
Using $(\ref{kappalower})$ and $t = y/x > 1$ in $(\ref{deltaprime0})$, we obtain 
\begin{align*}
     \theta \left(1 + \frac{2^{R} y^2}{ x^2}  \right)     -  \frac{2^{R + 1} y^2}{ x}  - 2y   &>  \left(1 + u  \right)  \frac{\ln(u)}{t- 1} x t    - 2 u x   - 2 tx\\
     &\stackrel{(a)}{\geq} \frac{x}{t-1} \left( 2(u - 1) t     - 2 (u + t)(t-1) \right)\\
     &= \frac{x}{t-1} \left( 2tu - 2t     -2tu + 2u - 2t^2 + 2t \right)\\
     &= \frac{2x}{t-1} \left(  u - t^2  \right) > 0. 
\end{align*}
In inequality $(a)$ above, we used the inequality $(1 + u) \ln(u) \geq 2 (u - 1)$ for all $u \geq 1$. 

\section{Proof of  $(\ref{opt_prop_2})$ in Theorem \ref{opt_properties_theorem} \label{opt_prop_2_proof}}

Define 
\begin{align*}
   f(\alpha_1, \ldots, \alpha_K) =  \sum_{i=1}^K  \exp \left(- \frac{\tau_i}{\sigma^2}  \right) d_i.
\end{align*}

If $\alpha_i = 0$, then it is clear from the definition of $\tau_i$ in $(\ref{deftau})$ that $\tau_i = +\infty$. We prove the reverse implication by contradiction. Let $(\alpha_1, \ldots, \alpha_K)$ be an optimal solution. First, note that $\tau_1$ has to be finite; otherwise, by the first point of Theorem \ref{opt_properties_theorem}, $f(\alpha_1, \ldots, \alpha_K) = 0$. If $\tau_i = +\infty$ for some $2 \leq i \leq K$, then Theorem $\ref{opt_properties_theorem}.(\ref{opt_prop_1})$ implies that 
\begin{align*}
    f(\alpha_1, \ldots, \alpha_K) = \sum_{k=1}^{i-1} \exp \left(- \frac{\tau_k}{\sigma^2}  \right) d_k.
\end{align*}
If $\tau_i = +\infty$ but $\alpha_i > 0$, consider a modified feasible point given by 
\begin{align}
    (\alpha_1', \ldots, \alpha_K') = (\alpha_1 + \alpha_{i}, \alpha_2, \ldots, \alpha_{i-1}, 0, \alpha_{i+1}, \ldots, \alpha_K).  
\end{align}
Then $\tau_1' < \tau_1$ and $\tau_j' \leq \tau_j$ for $j = 2, \ldots, i-1$. Therefore, 
\begin{align}
 f(\alpha_1',  \ldots, \alpha_K')   &\geq \sum_{k=1}^{i-1} \exp \left(- \frac{\tau_k'}{\sigma^2}  \right) d_k\\
 &= \exp\left(-\frac{\tau_1'}{\sigma^2} \right) d_1 +  \sum_{k=2}^{i-1} \exp \left(- \frac{\tau_k'}{\sigma^2}  \right) d_k \\
 &> \exp\left(-\frac{\tau_1}{\sigma^2} \right) d_1 +  \sum_{k=2}^{i-1} \exp \left(- \frac{\tau_k'}{\sigma^2}  \right) d_k\\
 &\geq \exp\left(-\frac{\tau_1}{\sigma^2} \right) d_1 +  \sum_{k=2}^{i-1} \exp \left(- \frac{\tau_k}{\sigma^2}  \right) d_k \\
 &= f(\alpha_1, \ldots, \alpha_K). 
\end{align}
Hence, $(\alpha_1, \ldots, \alpha_K)$ cannot be optimal.

\section{Proof of $(\ref{opt_prop_3}), (\ref{opt_prop_4}), (\ref{opt_prop_5})$ and $(\ref{opt_prop_6})$ in Theorem \ref{opt_properties_theorem} \label{3-6_proof}}

\subsection{Proof of $(\ref{opt_prop_3})$}

If $\tau_1 = \infty$, then Theorem $\ref{opt_properties_theorem}.(\ref{opt_prop_1})$ implies that $\tau_i = \infty$ for all $i \geq 1$. This implies that the maximum value is zero. However, choosing $\alpha_1 = 1, \alpha_j = 0$ for $j \geq 2$ gives a strictly positive value of the objective, leading to a contradiction. Hence, we have $\tau_1 < \infty$ which also implies $\alpha_1 > 0$ by Theorem $\ref{opt_properties_theorem}.(\ref{opt_prop_2})$.

\subsection{Proof of $(\ref{opt_prop_4})$}

Suppose $\alpha_i - (2^R - 1) \beta_i < 0$ for some $i$ so that $\beta_i > 0$. Then Theorem $\ref{opt_properties_theorem}.(\ref{opt_prop_3})$ implies that $i \geq 2$. By Theorem $\ref{opt_properties_theorem}.(\ref{opt_prop_1})$, we have $\tau_j = \infty$ for all $j \geq i$. Hence, the maximum value is given by 
\begin{align}
    \sum_{j=1}^{i-1} \exp\left(-\frac{\tau_j}{\sigma^2} \right)d_j. \label{nvf}
\end{align}
Consider a modified feasible point 
\begin{align*}
     (\alpha_1', \ldots, \alpha_K') = (\alpha_1 + \alpha_i + \beta_i, \ldots, \alpha_{i-1}, 0, \ldots, 0),  
\end{align*}
This gives a strictly higher value than $(\ref{nvf})$ because $\beta_i > 0$, leading to a contradiction.

\subsection{Proof of $(\ref{opt_prop_5})$}

If $\beta_{i-1} = 0$, then $\alpha_i = \alpha_{i + 1} = \cdots = \alpha_K = 0$. Then clearly, $\alpha_i - (2^R - 1) \beta_i = 0$. On the other hand, when $\alpha_i - (2^R - 1)\beta_i = 0$, then $\tau_i = \infty$, and $\tau_{j} = \infty$ for all $j \geq i$ by Theorem $\ref{opt_properties_theorem}.(\ref{opt_prop_1}).$ Then Theorem $\ref{opt_properties_theorem}.(\ref{opt_prop_2})$ implies that $\alpha_j = 0$ for all $j \geq i$. This implies that $\beta_{i-1} = 0.$

\subsection{Proof of $(\ref{opt_prop_6})$}

Suppose $\tau_i = \tau_{i + 1}$ for some $i \in \{1, \ldots, \ell - 1 \}$. Since $\tau_i, \tau_{i + 1} < \infty$, this implies that 
\begin{align}
    \frac{2^{R} - 1}{P(\alpha_i -(2^{R} - 1) \beta_i)} &= \frac{2^{R} - 1}{P(\alpha_{i + 1} -(2^{R} - 1) \beta_{i + 1})}\\
    \implies \alpha_{i + 1} -(2^{R} - 1) \beta_{i + 1} &= \alpha_i -(2^{R} - 1) \beta_i \eqqcolon \xi. \label{xideft} 
\end{align}

Then consider a modified feasible point $\overrightarrow \alpha'$ specified as $\alpha_i' = \alpha_i + \delta$, $\alpha_{i + 1}' = \alpha_i - \delta$ and $\alpha_j' = \alpha_j$ for $j \neq i, i + 1$. Since $\tau_i, \tau_{i + 1} < \infty$, we choose $\delta$ small enough so that $\tau_i', \tau_{i + 1}' < \infty$.
Define 
\begin{align*}
   f(\alpha_1, \ldots, \alpha_K) &=  \sum_{i=1}^K  \exp \left(- \frac{\tau_i}{\sigma^2}  \right) d_i.\\
    \Delta(\delta) &\coloneqq \frac{\partial f(\alpha_1', \ldots, \alpha_K') }{\partial \delta}.
\end{align*}
We then have $\tau_j' = \tau_j$ for all $j \neq i, i + 1$ and   
\begin{align*}
    \tau_i' &= \frac{2^R - 1}{P(\alpha_i - (2^R - 1) \beta_i + 2^R \delta )} \\
    \tau_{i + 1}' &= \frac{2^R - 1}{P(\alpha_{i  +1}  - (2^R - 1) \beta_{i + 1} - \delta)}.
\end{align*}
Then 
\begin{align*}
    \Delta(\delta) &= \frac{\partial }{\partial \delta} \left [ \exp \left(- \frac{\tau_i'}{\sigma^2}  \right) d_i  + \exp \left(- \frac{\tau_{i + 1}'}{\sigma^2}  \right) d_{i + 1}\right]\\
    &= \frac{\partial }{\partial \delta} \left [ \exp\left(-  \frac{\theta}{\alpha_i - (2^R - 1) \beta_i + 2^R \delta }\right) d_i + \exp \left(- \frac{\theta}{\alpha_{i  +1}  - (2^R - 1) \beta_{i + 1} - \delta}  \right) d_{i+1}  \right]\\
    &= \frac{\theta 2^R}{(\alpha_i - (2^R - 1) \beta_i + 2^R \delta)^2} \exp\left(-  \frac{\theta}{\alpha_i - (2^R - 1) \beta_i + 2^R \delta }\right) d_i - \mbox{} \\
    & \quad \quad \quad \quad \frac{\theta}{(\alpha_{i  +1}  - (2^R - 1) \beta_{i + 1} - \delta)^2}\exp \left(- \frac{\theta}{\alpha_{i  +1}  - (2^R - 1) \beta_{i + 1} - \delta}  \right) d_{i+1},\\
    \Delta(0) &= \frac{\theta 2^R}{(\alpha_i - (2^R - 1) \beta_i)^2} \exp\left(-  \frac{\theta}{\alpha_i - (2^R - 1) \beta_i  }\right) d_i - \mbox{} \\
    & \quad \quad \quad \quad \frac{\theta}{(\alpha_{i  +1}  - (2^R - 1) \beta_{i + 1} )^2}\exp \left(- \frac{\theta}{\alpha_{i  +1}  - (2^R - 1) \beta_{i + 1} }  \right) d_{i+1}\\
    &\stackrel{(a)}{=} \frac{\theta 2^R}{\xi^2} \exp\left(-  \frac{\theta}{\xi  }\right) d_i - \frac{\theta}{\xi^2}\exp \left(- \frac{\theta}{\xi }  \right) d_{i+1}\\
    &= \frac{\theta}{\xi^2} \exp\left(-  \frac{\theta}{\xi  }\right) \left [ 2^R d_i - d_{i + 1} \right ] > 0. 
\end{align*}
Hence, $\overrightarrow \alpha$ cannot be optimal. In equality $(a)$ above, we used $(\ref{xideft})$. 

\section{Proof of Corollary \ref{corl4} \label{corl4_proof}}

Clearly, we have 
\begin{align}
    \max_{\overrightarrow \alpha \in \Delta^{K - 1}} \sum_{i=1}^K  \exp \left(- \frac{\tau_i}{\sigma^2}  \right) d_i \geq \max_{\overrightarrow \alpha \in \Delta^{K - 1}} \sum_{i=1}^K  \exp \left(- \frac{\max \{\tau_1, \ldots, \tau_i \}}{\sigma^2}  \right) d_i. \label{65a}
\end{align}
To show the reverse inequality, we write
\begin{align}
    \max_{\overrightarrow \alpha \in \Delta^{K - 1}} \sum_{i=1}^K  \exp \left(- \frac{\tau_i}{\sigma^2}  \right) d_i &\stackrel{(a)}{=} \max_{\substack{\overrightarrow \alpha \in \Delta^{K - 1}\\\tau_1 \leq \cdots \leq \tau_K}} \sum_{i=1}^K  \exp \left(- \frac{\tau_i}{\sigma^2}  \right) d_i \notag \\
    &= \max_{\substack{\overrightarrow \alpha \in \Delta^{K - 1}\\\tau_1 \leq \cdots \leq \tau_K}} \sum_{i=1}^K  \exp \left(- \frac{\max \{\tau_1, \ldots, \tau_i \}}{\sigma^2}  \right) d_i\notag \\
    &\leq \max_{\overrightarrow \alpha \in \Delta^{K - 1}} \sum_{i=1}^K  \exp \left(- \frac{\max \{\tau_1, \ldots, \tau_i \}}{\sigma^2}  \right) d_i. \label{65b}
\end{align}
In equality $(a)$ above, we used Theorem $\ref{opt_properties_theorem}.(\ref{opt_prop_1})$.   

Define 
\begin{align*}
    F(\overrightarrow \alpha) &= \sum_{i=1}^K  \exp \left(- \frac{\max \{\tau_1, \ldots, \tau_i \}}{\sigma^2}  \right) d_i,\\
    G(\overrightarrow \alpha) &= \sum_{i=1}^K  \exp \left(- \frac{\tau_i}{\sigma^2}  \right) d_i.
\end{align*}
From $(\ref{65a})$ and $(\ref{65b})$, we have 
\begin{align}
    \max_{\overrightarrow \alpha \in \Delta^{K - 1}} F(\overrightarrow \alpha) = \max_{\overrightarrow \alpha \in \Delta^{K - 1}} G(\overrightarrow \alpha). \label{equality_invals}
\end{align}

Let $\overrightarrow{\alpha}^\star$ be any optimal solution in the LHS of $(\ref{equality_invals})$. Then 
\begin{align}
    \max_{\overrightarrow \alpha \in \Delta^{K - 1}} F(\overrightarrow \alpha) &=  F(\overrightarrow{\alpha}^\star) \leq G(\overrightarrow{\alpha}^\star) \leq    \max_{\overrightarrow \alpha \in \Delta^{K - 1}} G(\overrightarrow \alpha).  \label{ineqgkmt}
\end{align}
But $(\ref{equality_invals})$ implies that all the inequalities in $(\ref{ineqgkmt})$ must be equalities; hence, $\overrightarrow \alpha^\star$ is also an optimal solution in the RHS of $(\ref{equality_invals})$.

Now let $\overrightarrow \alpha^\star$ be any optimal solution to the RHS of $(\ref{equality_invals})$. Theorem \ref{opt_properties_theorem} implies that $\overrightarrow \alpha^\star$ satisfies $\tau_1^\star \leq \cdots \leq \tau_K^\star$. Then 
\begin{align}
    \max_{\overrightarrow \alpha \in \Delta^{K - 1}} G(\overrightarrow \alpha) &=  G(\overrightarrow{\alpha}^\star) = F(\overrightarrow \alpha^\star) \leq \max_{\overrightarrow \alpha \in \Delta^{K - 1}} F(\overrightarrow \alpha).  \label{gdsn}
\end{align}
Then $(\ref{equality_invals})$ again implies that the inequality in $(\ref{gdsn})$ must be an equality. Therefore, $\overrightarrow \alpha^\star$ is also an optimal solution in the LHS of $(\ref{equality_invals})$.

\section{Proof of Theorem \ref{K=2theorem} \label{K=2theorem_proof}}

Let $s_i \in [0, 1]$ be the remaining power budget after $\alpha_1, \ldots, \alpha_{i-1}$ have been chosen, i.e., $s_i = 1 - \sum_{j=1}^{i-1} \alpha_j$. Then $\alpha_i \in [0, s_i ]$. Note also that $\beta_i = s_i - \alpha_i$. We can express $\tau_i$ in $(\ref{deftau})$ as a function of $\alpha_i$ and $s_i$; specifically, $\tau_i = \tau(s_i, \alpha_i)$ where  
\begin{align*}
    \tau(s, \alpha) = \begin{cases}
         \frac{2^{R} - 1}{P(2^R( \alpha-s) + s)} & \text{ if } 2^R( \alpha -s) + s > 0,\\
         +\infty & \text{ otherwise}.  
    \end{cases}
\end{align*}
Define 
\begin{align*}
    r_i(s_i, \alpha_i) &= \exp \left(- \frac{\tau_i}{\sigma^2} \right) d_i.  
\end{align*}
Define the value function 
\begin{align*}
    V_t(s) &= \max_{\substack{\alpha_t, \ldots, \alpha_K \in [0, s]\\
    \sum_{i=t}^K \alpha_i = s}} \sum_{i=t}^K r_i(s_i, \alpha_i),
\end{align*}
where $s_t = s$ and $s_i = s_{i-1} - \alpha_{i-1}$ for $i = t + 1, \ldots, K$. Therefore, the optimization problem $(\ref{c1})$ can be written as  
\begin{align*}
    \max_{\overrightarrow \alpha \in \Delta^{K - 1}} \sum_{i=1}^K  \exp \left(- \frac{\tau_i}{\sigma^2}  \right) d_i = V_1(1).  
\end{align*}
Bellman equation gives us  
\begin{align}
    V_t(s) &= \max_{\alpha \in [0, s]} \,  \left [ r_t(s, \alpha) + V_{t + 1}(s - \alpha) \right] \label{bellmaneq}
\end{align}
for $t = K - 1, \ldots, 1$ and 
\begin{align}
    V_K(s) =  \max_{\alpha \in [0, s]} r_K(s, \alpha) &= \max_{\alpha \in [0, s]} \exp \left(- \frac{\tau_K}{\sigma^2} \right) d_K \notag \\
    &= \begin{cases}
        \exp \left(- \frac{2^R - 1}{s P \sigma^2} \right) d_K & s > 0\\
        0 & s = 0
    \end{cases} \notag \\
    &=  \exp \left(- \frac{2^R - 1}{s P \sigma^2} \right) d_K \label{b23}
\end{align}
for all $s \in [0, 1]$ with the
convention that
\begin{align}
  \exp \left(- \frac{2^R - 1}{s P \sigma^2} \right) d_K \Bigg |_{s = 0} =  \lim_{s \downarrow 0} \exp \left(- \frac{2^R - 1}{s P \sigma^2} \right) d_K = 0. \label{convs=0}
\end{align}
Hence, the optimal power assignment $\alpha_K$ given the remaining power budget $s_K$ at step $K$ is $\alpha_K = s_K$. For $K = 2$, this implies 
\begin{align*}
    V_2(s) &= \exp \left(- \frac{2^R - 1}{s P \sigma^2} \right) d_2. 
\end{align*}
Then using $(\ref{bellmaneq})$ along with the fact that $s_1 = 1$ gives us 
\begin{align*}
    V_{1}(1) &= \max_{\alpha \in [0, s]} \left [   r_{1}(1, \alpha )  + V_{2}\left( 1 - \alpha \right)\right ]\\
    &= \max_{\alpha \in [0, 1]} \left [  \exp \left( -\frac{\tau(1, \alpha)}{\sigma^2} \right) d_1 +  \exp \left(- \frac{2^R - 1}{ P \sigma^2(1 - \alpha)} \right) d_2\right ]\\
    &= \max_{\alpha \in [0, 1] }  w(\alpha),  
\end{align*}
where 
\begin{align}
    w(\alpha) &\coloneqq \exp \left( -\frac{\tau(1, \alpha)}{\sigma^2} \right) d_1 +  \exp \left(- \frac{2^R - 1}{ P \sigma^2(1 - \alpha)} \right) d_2. \label{bm} 
\end{align}
When $\alpha = 0$, 
\begin{align*}
    w(0) &=  \exp \left(- \frac{2^R - 1}{ P \sigma^2 } \right) d_2. 
\end{align*}
As $\alpha$ increases from $0$, the second term in $(\ref{bm})$ decreases and the first term remains constant at zero until $2^R (\alpha -1) + 1 > 0$ or, equivalently, 
\begin{align}
    \alpha > \alpha' = \frac{2^R - 1}{2^R}.
\end{align}
Hence, the function $w(\alpha)$ is decreasing over $[0, \alpha']$. For $\alpha \in (\alpha', 1]$, we have 
\begin{align}
    w(\alpha) &=  \exp \left( -\frac{2^{R} - 1}{\sigma^2 P(2^R( \alpha-1) + 1)} \right) d_1 +  \exp \left(- \frac{2^R - 1}{ P \sigma^2(1 - \alpha)} \right) d_2. 
\end{align}
Over $(\alpha', 1]$, the first term increases and the second term decreases. Since $d_1 > d_2$, 
\begin{align*}
    w(1) = \exp \left( -\frac{2^{R} - 1}{ P \sigma^2} \right) d_1 > w(0) \geq \max_{\alpha \in [0, \alpha']} w(\alpha).
\end{align*}
Since the function $w(\alpha)$ is continuous over $[0,1]$, the maximum value over $[0,1]$ must be attained for some $\alpha \in (\alpha', 1]$. Define 
\begin{align*}
    \theta \coloneqq \frac{2^R - 1}{P \sigma^2} > 0.
\end{align*}
Define a function $\tilde{w} : \mathbb{R} \to (0, \infty)$ as
\begin{align*}
    \tilde{w}(\alpha) &= \exp \left( -\frac{2^{R} - 1}{\sigma^2 P(2^R( \alpha-1) + 1)} \right) d_1 +  \exp \left(- \frac{2^R - 1}{ P \sigma^2(1 - \alpha)} \right) d_2\\
    &=  \exp \left( -\frac{\theta}{2^R( \alpha-1) + 1} \right) d_1 +  \exp \left(- \frac{\theta}{1 - \alpha} \right) d_2. 
\end{align*}
An optimal power allocation is then given by $(\alpha^\star, 1 - \alpha^\star)$, where 
\begin{align*}
    \alpha^\star &= \argmax_{\alpha \in [\alpha', 1]}  \tilde{w}(\alpha). 
\end{align*}
For $\alpha \neq \alpha', 1$,  
\begin{align}
    \tilde{w}'(\alpha) &=  \theta \left [ \frac{2^R   }{ \left(2^R (\alpha -1)+1\right)^2} \exp\left(-\frac{\theta}{2^R (\alpha -1)+1}\right) d_1   - \frac{1}{(1-\alpha )^2} \exp \left(- \frac{\theta}{ 1-\alpha } \right) d_2 \right ]. \notag 
\end{align}
Note that 
\begin{align}
    \lim_{\alpha \uparrow 1} \tilde{w}'(\alpha) &> 0 \label{posd}\\
    \lim_{\alpha \downarrow \alpha'} \tilde{w}'(\alpha) &< 0. \label{negd}
\end{align}
Hence, $\alpha^\star > \alpha'$ since $\tilde{w}(1) > \tilde{w}(\alpha')$. Therefore, either $\alpha^\star = 1$ or $\alpha^\star \in (\alpha', 1)$ such that $\tilde{w}'(\alpha^\star) = 0.$ 

From $(\ref{posd})$, $(\ref{negd})$ and the fact that $\tilde{w}'(\alpha)$ is differentiable over $(\alpha', 1)$, there always exists a local minimizer $\alpha_{1} \in (\alpha', 1)$ such that $\tilde{w}'(\alpha_{1}) = 0$. We will later show that there can be at most $3$ critical points of $\tilde{w}(\alpha)$ over $(\alpha', 1)$. Therefore, a necessary condition for $\alpha^\star \in (\alpha', 1)$ is that there be three critical points of $\tilde{w}(\alpha)$ over $(\alpha', 1)$.

Let $r = \frac{d_2}{d_1} \in (0, 1)$. Setting $\tilde{w}'(\alpha) = 0$ gives us 
\begin{align}
    \frac{2^R   }{ \left(2^R (\alpha -1)+1\right)^2} \exp\left(-\frac{\theta}{2^R (\alpha -1)+1}\right)   &= \frac{1}{(1-\alpha )^2} \exp \left(- \frac{\theta}{ 1-\alpha } \right) r\\
    \ln(2^R) - 2 \ln \left( 2^R (\alpha -1)+1 \right) -\frac{\theta}{2^R (\alpha -1)+1} &= \ln(r) - \frac{\theta}{ 1-\alpha } - 2 \ln(1 - \alpha)\\
    \ln(2^R/r) + 2 \ln \left(\frac{1-\alpha}{ 2^R (\alpha -1)+1} \right) &= \frac{\theta (2^R - 2^R \alpha - \alpha)}{(1-\alpha)(2^R \alpha - 2^R + 1)}. \label{xtoq}
\end{align}
We are only interested in the solutions to $(\ref{xtoq})$ that lie inside the interval $(\alpha', 1)$. To restrict our attention to the interval $(\alpha', 1)$, let 
\begin{align}
    q = q(\alpha) = \frac{ 2^R (\alpha -1)+1}{1-\alpha} \label{alphatoqmapping}    
\end{align}
so that for $\alpha \in (\alpha', 1)$, $q \in (0, \infty)$ with $q \to 0$ as $\alpha \to \alpha'$ and $q \to +\infty$ as $\alpha \uparrow 1$. And $q : (\alpha', 1) \to (0, \infty)$ is strictly increasing and invertible so that $\alpha = \frac{q + 2^R - 1}{q + 2^R}$. Then $(\ref{xtoq})$ is equivalent to 
\begin{align}
    \ln\left( \frac{rq^2}{2^R} \right)  &=  \theta \left(q + 2^R - 1 -  \frac{2^R}{q} \right)\\
     \frac{rq^2}{2^R} &= \exp \left( \theta \left(q + 2^R - 1 -  \frac{2^R}{q} \right) \right). 
\end{align}
Define $h : (0, \infty) \to (0, \infty)$ as 
\begin{align}
    h(q) = \frac{2^R}{rq^2} \exp \left( \theta \left(q + 2^R - 1 -  \frac{2^R}{q} \right) \right). \label{hdef}
\end{align}
Then the solution set of $\tilde{w}'(\alpha) = 0$ for $\alpha \in (\alpha', 1)$ is specified by the solution set of $h(q) = 1$ for $q \in (0, \infty)$. We have $h(0^+) = 0$ and $h(q) \to +\infty$ as $q \to \infty$. 
A sufficient condition for a unique solution to $h(q) = 1$ over $(0, \infty)$ is that $h(q)$ be strictly increasing over $(0, \infty)$. We have 
\begin{align}
    \frac{d}{dq}\left( \ln h(q)\right) &= \theta \frac{q^2 - \frac{2}{\theta} q + 2^R}{q^2} = \frac{h'(q)}{h(q)}. \label{derhq} 
\end{align}
Hence, 
\begin{align*}
    q^2 - \frac{2}{\theta} q + 2^R > 0 \text{ for all } q > 0 \iff \theta > \frac{1}{2^{R/2}}
\end{align*}
is a sufficient condition for 
\begin{itemize}
    \item a unique solution to $h(q) = 1$ for $q \in (0, \infty)$,
    \item the existence of exactly one critical point (a local minimum) of $\tilde{w}(\alpha)$ for $\alpha \in (\alpha', 1)$,  
    \item $\alpha^\star = 1$.
\end{itemize}
For $\theta = \frac{1}{2^{R/2}}$, $h'(q) > 0$ for all $q > 0, q \neq \frac{1}{\theta}$. Hence, we have the following lemma.  
\begin{lemma}
    If $$\frac{2^{R} - 1}{P \sigma^2} \geq 2^{-R/2},$$ then $\alpha^\star = 1$ for all $d_1 > d_2 > 0$.  
\end{lemma}

If $0 < \theta < \frac{1}{2^{R/2}}$, then let $q_- = \frac{1}{\theta} - \sqrt{\frac{1}{\theta^2} - 2^R}$ and $q_+ = \frac{1}{\theta} + \sqrt{\frac{1}{\theta^2} - 2^R}$. For $q \in (0, q_-)$, $h(q)$ is increasing. For $q \in (q_-, q_+)$, $h(q)$ is decreasing. For $q > q_+$, $h(q)$ is increasing. Hence, $h(q) = 1$ can have at most $3$ solutions for $q \in (0, \infty)$. Hence, $\theta < \frac{1}{2^{R/2}}$ is a necessary condition for the existence of three solutions of $\tilde{w}'(\alpha) = 0$ over $(\alpha', 1)$. A necessary and sufficient condition for the existence of three solutions of $h(q) = 1$ over $(0, \infty)$ is described after stating the following lemma.  

\begin{lemma}
 For all $R > 0$, $\theta < \frac{1}{2^{R/2}}$ and $r \in (0, 1)$, $h(q_-) \geq \frac{2^R}{r} > 1$.
 \label{hq->1}
\end{lemma}
\textit{Proof:} The proof of Lemma \ref{hq->1} is given in Appendix \ref{hq->1proof}. 
 
Since $h(q_-)$ is a local maximum and $h(q_+)$ is a local minimum, Lemma \ref{hq->1} implies that  
\begin{align*}
     \# \{\alpha \in (\alpha', 1): \tilde{w}'(\alpha) = 0 \} = \begin{cases}
        1 & \text{ if } \theta \geq \frac{1}{2^{R/2}}\\
        1 & \text{ if } h(q_+) > 1, \theta < \frac{1}{2^{R/2}}\\
        2 & \text{ if } h(q_+) = 1, \theta < \frac{1}{2^{R/2}}\\
        3 & \text{ if } h(q_+) < 1, \theta < \frac{1}{2^{R/2}}. 
    \end{cases}
\end{align*}

Define 
\begin{align*}
    \xi(P, \sigma^2, R) = \left( \frac{\theta 2^{R/2}}{1 +  \sqrt{1-2^R\theta^2 }}\right)^2  \exp \left( \theta(2^R - 1) +  2 \sqrt{1 - 2^R \theta^2} \right).
\end{align*}
\begin{lemma}
   Let $0 < \theta < 2^{-R/2}$. Then the condition $h(q_+) > 1$ is equivalent to 
   \begin{align}
       \frac{d_2}{d_1} < \xi(P, \sigma^2, R). \label{b-}
   \end{align}
   \label{rcond}
\end{lemma}
\textit{Proof:} The proof of Lemma \ref{rcond} is given in Appendix \ref{rcond_proof}.

Note that that $\theta < 2^{-R/2}, h(q_+) < 1$ is a necessary but not sufficient condition for $\alpha^\star \in (\alpha', 1)$. This is because $\theta < \frac{1}{2^{R/2}}$ only implies that $h(q)$ has a local maximizer $q_-$ and a local minimizer $q_+$ over $(0, \infty)$. With $h(q_+) < 1$, we only have that $\tilde{w}(\alpha)$ has two local minima and one local maximum, denoted as $\tilde{w}(\alpha_{0})$, in the interval $(\alpha', 1)$. A necessary and sufficient condition therefore is $\tilde{w}(\alpha_0) > w(1).$ 

From $(\ref{posd})$ and $(\ref{negd})$, it is clear that the order of the critical points of $\tilde{w}$ is min - max - min. Given $\theta < 2^{-R/2}, h(q_+) < 1$, there exist three distinct solutions $0 < q_1 < q_0 < q_3 < \infty$ to    
\begin{align}
    \frac{2^R}{rq^2} \exp \left( \theta \left(q + 2^R - 1 -  \frac{2^R}{q} \right) \right) = 1, \label{gi}
\end{align}
where $q_0 \in (q_-, q_+)$ corresponds to the local maximizer $\alpha_0$ of $\tilde{w}(\alpha)$ over $(\alpha', 1)$ because $(\ref{alphatoqmapping})$ preserves the order of the solutions of $h(q) = 1$ and those of $\tilde{w}'(\alpha) = 0$. Then \begin{align*}
    \tilde{w}\left(\alpha_0 = \frac{q_0 + 2^R - 1}{q_0 + 2^R}\right) &= d_1 \exp\left(-\frac{\theta \left(q_0+2^R\right)}{q_0} \right)  + r d_1 \exp \left( -\theta \left(q_0+2^R\right) \right)\\
    &\stackrel{(a)}{=} d_1 \exp\left(-\frac{\theta \left(q_0+2^R\right)}{q_0} \right) +  d_1 \frac{2^R}{q_0^2} \exp \left(-\theta\left( \frac{q_0 + 2^R}{q_0} \right) \right)\\
    &= d_1 \exp\left(-\frac{\theta \left(q_0+2^R\right)}{q_0} \right) \left( 1 + \frac{2^R}{q_0^2} \right).
\end{align*}
In equality $(a)$ above, we used the fact that $q_0$ satisfies $(\ref{gi})$. Hence, $\tilde{w}(1) < \tilde{w}(\alpha_0)$ can be written as  
\begin{align*}
     1  < \exp\left(-\frac{\theta 2^R}{q_0} \right) \left( 1 + \frac{2^R}{q_0^2} \right).
\end{align*}

\section{Proof of Lemma \ref{atox} \label{atox_proof}}

For any $\overrightarrow \alpha \in \Delta_+^{K-1}$, we show that $\overrightarrow x = M_B( \overrightarrow \alpha)$ be specified by
\begin{align}
    x_i = \alpha_i - (2^R -1)\beta_i \text{ for all } 1 \leq i \leq K. \label{atoxmapping}   
\end{align}
We now show that $\overrightarrow x \in \mathcal{S}_K$. For $1 \leq i \leq K$, 
\begin{align}
    \alpha_i &= x_i + (2^R -1)\beta_i \label{recur1} \\
    \beta_{i-1} - \beta_i &=  x_i + (2^R -1)\beta_i \notag \\
    \beta_{i-1} &= x_i + 2^R\beta_i. \label{recurbeta} 
\end{align}
The fact that $(\ref{recurbeta})$ holds for all $1 \leq i \leq K$ implies that 
\begin{align}
    \beta_i &= \sum_{j = i + 1}^K (2^R)^{j - i - 1} x_j \text{ for } i  = 0, \ldots, K, \label{recur4}
\end{align}
where $\beta_K = 0$ by convention. Hence, 
\begin{align}
    &\sum_{j=1}^{K} (2^R)^{j-1} x_j = \beta_0 = \sum_{i=1}^K \alpha_i.  \label{xcon} 
\end{align}
The last expression above on the RHS is equal to $1$ since $\overrightarrow \alpha \in \Delta_+^{K-1}$. Also, each $x_i \geq 0$ so we conclude that $\overrightarrow x = M_B(\overrightarrow \alpha) \in \mathcal{S}_K$ for all $\overrightarrow \alpha \in \Delta_+^{K-1}$. 

We will now show that $M_B : \Delta_+^{K-1} \to \mathcal{S}_K$ is a bijection by constructing a mapping $\Psi : \mathcal{S}_K \to \Delta_+^{K-1}$ such that
\begin{itemize}
    \item $M_B(\Psi(\overrightarrow x)) = \overrightarrow x$ for all $\overrightarrow x \in \mathcal{S}_K$
    \item $\Psi(M_B(\overrightarrow \alpha)) = \overrightarrow \alpha$ for all $\overrightarrow \alpha \in \Delta_+^{K-1}$.
\end{itemize}
\vspace{0.1cm}
Let $\overrightarrow{\tilde{\alpha}} = \Psi(\overrightarrow x)$ be specified as 
\begin{align}
    \tilde{\alpha}_{i} &= x_{i} + (2^R - 1)\sum_{j=i + 1}^{K} 2^{R(j-i-1)} x_{j} \label{psispecification}
\end{align}
for all $1 \leq i \leq K$. Then it can be checked that 
\begin{align}
    \tilde{\beta}_i &= \sum_{j=i + 1}^K \tilde{\alpha}_j = \sum_{j=i + 1}^{K} 2^{R(j-i - 1)} x_{j}. \label{psispecification2} 
\end{align}
Clearly, $\tilde{\alpha}_i \geq 0$ and $\tilde{\alpha}_i - (2^R-1) \tilde{\beta}_i \geq 0$. Moreover,  
\begin{align*}
    \sum_{i=1}^K \tilde{\alpha}_i = \tilde{\beta}_0 =  \sum_{j=1}^{K} (2^R)^{j-1} x_j = 1, 
\end{align*}
where the last equality above uses the fact that $\overrightarrow x \in \mathcal{S}_K$. Hence, $\overrightarrow{\tilde{\alpha}} \in \Delta_+^{K-1}$.

\subsection{$M_B(\Psi(\overrightarrow x)) = \overrightarrow x$}

In this subsection, we show that $M_B(\Psi(\overrightarrow x)) = \overrightarrow x$. Let $\overrightarrow{\tilde{\alpha}} = \Psi(\overrightarrow x)$, where $\overrightarrow{\tilde{\alpha}}$ is specified in $(\ref{psispecification})$ and $M_B(\cdot)$ is specified in $(\ref{atoxmapping})$. For clarity, we write $\overrightarrow{\tilde{x}} = M_B(\overrightarrow{\tilde{\alpha}})$ and proceed to show that $\overrightarrow{\tilde{x}} = \overrightarrow{x}$. Indeed,
\begin{align*}
    \tilde{x}_K &\stackrel{(a)}{=} \tilde{\alpha}_K \stackrel{(b)}{=} x_K\\
    \tilde{x}_{K-1} &\stackrel{(a)}{=} \tilde{\alpha}_{K-1} - (2^R - 1)\tilde{\beta}_{K-1} \stackrel{(b)}{=} x_{K-1}\\
    &\,\,\,\vdots\\
    \tilde{x}_i &\stackrel{(a)}{=} \tilde{\alpha}_i - (2^R-1) \tilde{\beta}_i \stackrel{(b)}{=} x_i
\end{align*}
for $i = K - 1, \ldots, 1$, where the equalities $(a)$ use the specification of $M_B(\cdot)$ given in  $(\ref{atoxmapping})$ and the equalities $(b)$ use the specification of $\Psi(\cdot)$ given in $(\ref{psispecification})$ and $(\ref{psispecification2})$. 

\subsection{$\Psi(M_B(\overrightarrow \alpha)) = \overrightarrow \alpha$}

In this subsection, we show that $\Psi(M_B(\overrightarrow\alpha)) = \overrightarrow\alpha$ for all $\overrightarrow\alpha \in \Delta_+^{K-1}$. Let $\overrightarrow x = M_B(\overrightarrow\alpha)$ as specified in $(\ref{atoxmapping})$ and let $\overrightarrow{\tilde{\alpha}} = \Psi(\overrightarrow x)$ as specified in $(\ref{psispecification}) - (\ref{psispecification2})$. Starting from $(\ref{psispecification})$, we have
\begin{align*}
    \tilde{\alpha}_{i} &= x_{i} + (2^R - 1)\sum_{j=i + 1}^{K} 2^{R(j-i-1)} x_{j} \\
    &\stackrel{(a)}{=} \alpha_i - (2^R -1)\beta_i + (2^R - 1) \sum_{j=i + 1}^{K} 2^{R(j-i-1)} x_{j}\\
    &\stackrel{(b)}{=} \alpha_i 
\end{align*}
for all $1 \leq i \leq K$, where equality $(a)$ above uses $(\ref{atoxmapping})$ and equality $(b)$ above uses $(\ref{recur4})$.

\section{Proof of Theorem \ref{orasiml} \label{orasiml_proof}}

First, consider any sequence $\overrightarrow w^{(n)} \in \Delta_{n}^{K-1}$ that converges to some $\overrightarrow v \in \Delta^{K-1}$. Based on $(\ref{firstorderapproxRconstant})$, it can be checked that for every $\gamma \geq 0$ such that $\gamma P \neq 2^{R/v_i} - 1$, 
\begin{align*}
    \lim_{n \to \infty} \mathcal{E}\left(w_i^{(n)} n, \frac{R}{w_i^{(n)}}, \gamma P \right) &= \mathds{1}\left( R > v_i \log(1 + \gamma P) \right).
\end{align*}
Since $\mathbb{P}(\gamma P = 2^{R/v_i} - 1) = 0$, we can argue by dominated convergence theorem that 
\begin{align}
        \lim_{n \to \infty} T_n(\overrightarrow w^{(n)}) &= \sum_{i=1}^K \mathbb{E}_{\gamma}\left [1 - \lim_{n \to \infty} \mathcal{E}\left(w_i^{(n)} n, \frac{R}{w_i^{(n)}}, \gamma P \right)  \right] d_i\\
        &= \sum_{i=1}^K \mathbb{E}_{\gamma}\left [ \mathds{1}\left( R < v_i \log(1 + \gamma P) \right)  \right] d_i\\
        &= \sum_{i=1}^K \mathbb{P}\left(\gamma > \frac{2^{R/v_i}-1}{P} \right) d_i\\
    &= \sum_{i=1}^K \exp\left(- \frac{2^{R/v_i}-1}{P\sigma^2}\right) d_i\\
    &= \sum_{i=1}^K \exp\left(- \frac{2^{R/v_i}-1}{2^R - 1} \theta\right) d_i, \label{.,.q}
\end{align}
where the last equality follows by 
substituting the definition of $\theta$ from $(\ref{thetadef})$. 

Now we let $\overrightarrow w^{(n)} \in \Delta_n^{K-1} \subset \Delta^{K-1}$ be a maximizer in 
\begin{align*}
    \max_{\overrightarrow w \in \Delta_n^{K - 1}}\,T_n(\overrightarrow w). 
\end{align*}
Since $\Delta^{K-1}$ is compact, we can assume that $\overrightarrow w^{(n_m)} \to \overrightarrow v$ for some $\overrightarrow v \in \Delta^{K-1}$ by passing down to a convergent subsequence $\overrightarrow w^{(n_m)}$ which additionally satisfies
\begin{align*}
    \limsup_{n \to \infty}  \max_{\overrightarrow w \in \Delta_n^{K - 1}}\,T_n(\overrightarrow w) = \lim_{m \to \infty} T_{n_m}(\overrightarrow w^{(n_m)}).
\end{align*}
Hence, from $(\ref{.,.q})$, we have 
\begin{align}
    \limsup_{n \to \infty}  \max_{\overrightarrow w \in \Delta_n^{K - 1}}\,T_n(\overrightarrow w) \leq \max_{\overrightarrow v \in \Delta^{K - 1}}  \sum_{i=1}^K \exp\left(- \frac{2^{R/v_i}-1}{2^R - 1} \theta\right) d_i. \label{m/1-}
\end{align}

Now choose $\overrightarrow v$ to be a maximizer in 
\begin{align*}
    \max_{\overrightarrow v \in \Delta^{K - 1}} \sum_{i=1}^K \exp\left(- \frac{2^{R/v_i}-1}{2^R - 1} \theta\right) d_i. 
\end{align*}
Given this $\overrightarrow v$ and for each integer $n \geq 1$, construct $\overrightarrow w^{(n)}$ as outlined in Remark $\ref{vtowmapping}$ so that $\overrightarrow w^{(n)} \to \overrightarrow v$. Then invoking $(\ref{.,.q})$ again, we obtain
\begin{align}
    \liminf_{n \to \infty} \max_{\overrightarrow w \in \Delta_n^{K - 1}}\,T_n(\overrightarrow w) \geq \max_{\overrightarrow v \in \Delta^{K - 1}} \sum_{i=1}^K \exp\left(- \frac{2^{R/v_i}-1}{2^R - 1} \theta\right) d_i. \label{m/2-}
\end{align}
Combining $(\ref{m/1-})$ and $(\ref{m/2-})$ establishes the result.

\section{Proof of Lemma \ref{hq->1} \label{hq->1proof}}

For $\theta < \frac{1}{2^{R/2}}$, let $t = \frac{1}{\theta} - \sqrt{\frac{1}{\theta^2} - 2^R}$ so that $t \in (0,2^{R/2})$ and $\theta = \frac{2t}{t^2 + 2^R}$. Then  
\begin{align*}
    h(t) &= \frac{2^R}{r t^2} \exp \left( \frac{2t}{t^2 + 2^R} \left(t + 2^R - 1 -  \frac{2^R}{t} \right) \right).
\end{align*}
So it suffices to show that
$$F(t) = \frac{1}{t^2} \exp \left( \frac{2t}{t^2 + 2^R} \left(t + 2^R - 1 -  \frac{2^R}{t} \right) \right) \geq 1$$
for all $t \in (0, 2^{R/2})$, since $2^R/r > 1$. Indeed, the minimum of $F(t)$ over $(0, 2^{R/2})$ is attained at $t = 1$ and the minimum value is $1$. 

\section{Proof of Lemma \ref{rcond} \label{rcond_proof}}

We first write 
    \begin{align*}
        h(q) &= \frac{e^{ \theta \left(2^R - 1 \right)}}{r}  \frac{2^R}{q^2} \exp \left( \theta \left(q -  \frac{2^R}{q} \right) \right). 
    \end{align*}
    Then using 
    \begin{align*}
        \frac{1}{q_+^2} &= \left( \frac{\theta}{1 +  \sqrt{1-2^R\theta^2 }}\right)^2
    \end{align*}
    and 
    \begin{align*}
         \theta \left(q_+ -  \frac{2^R}{q_+} \right) &= 2 \sqrt{1 - 2^R \theta^2},  
    \end{align*}
    we have 
    \begin{align*}
        h(q_+) &= \frac{2^R}{r} \left( \frac{\theta}{1 +  \sqrt{1-2^R\theta^2 }}\right)^2  \exp \left( \theta(2^R - 1) +  2 \sqrt{1 - 2^R \theta^2} \right). 
    \end{align*}

\ifCLASSOPTIONcaptionsoff
  \newpage
\fi


\begin{thebibliography}{10}
\providecommand{\url}[1]{#1}
\csname url@samestyle\endcsname
\providecommand{\newblock}{\relax}
\providecommand{\bibinfo}[2]{#2}
\providecommand{\BIBentrySTDinterwordspacing}{\spaceskip=0pt\relax}
\providecommand{\BIBentryALTinterwordstretchfactor}{4}
\providecommand{\BIBentryALTinterwordspacing}{\spaceskip=\fontdimen2\font plus
\BIBentryALTinterwordstretchfactor\fontdimen3\font minus \fontdimen4\font\relax}
\providecommand{\BIBforeignlanguage}[2]{{%
\expandafter\ifx\csname l@#1\endcsname\relax
\typeout{** WARNING: IEEEtran.bst: No hyphenation pattern has been}%
\typeout{** loaded for the language `#1'. Using the pattern for}%
\typeout{** the default language instead.}%
\else
\language=\csname l@#1\endcsname
\fi
#2}}
\providecommand{\BIBdecl}{\relax}
\BIBdecl

\bibitem{kostina_JSCC}
V.~Kostina and S.~Verdú, ``Lossy joint source-channel coding in the finite blocklength regime,'' \emph{IEEE Transactions on Information Theory}, vol.~59, no.~5, pp. 2545--2575, 2013.

\bibitem{gunduz2022transmittingbitscontextsemantics}
\BIBentryALTinterwordspacing
D.~Gunduz, Z.~Qin, I.~E. Aguerri, H.~S. Dhillon, Z.~Yang, A.~Yener, K.~K. Wong, and C.-B. Chae, ``Beyond transmitting bits: Context, semantics, and task-oriented communications,'' 2022. [Online]. Available: \url{https://arxiv.org/abs/2207.09353}
\BIBentrySTDinterwordspacing

\bibitem{deepJSCC}
\BIBentryALTinterwordspacing
E.~Bourtsoulatze, D.~Burth~Kurka, and D.~Gunduz, ``Deep joint source-channel coding for wireless image transmission,'' \emph{IEEE Transactions on Cognitive Communications and Networking}, vol.~5, no.~3, p. 567–579, Sep. 2019. [Online]. Available: \url{http://dx.doi.org/10.1109/TCCN.2019.2919300}
\BIBentrySTDinterwordspacing

\bibitem{homaspaper1}
T.-Y. Tung, H.~Esfahanizadeh, J.~Du, and H.~Viswanathan, ``Multi-level reliability interface for semantic communications over wireless networks,'' \emph{IEEE Transactions on Communications}, vol.~73, no.~8, pp. 6023--6035, 2025.

\bibitem{homapaper2}
\BIBentryALTinterwordspacing
H.~Esfahanizadeh, N.~Fayaz, J.~Du, and H.~Viswanathan, ``Block erasure-aware semantic multimedia compression via jscc autoencoder,'' 2026. [Online]. Available: \url{https://arxiv.org/abs/2601.20707}
\BIBentrySTDinterwordspacing

\bibitem{UEP_paper}
S.~Borade, B.~Nakiboğlu, and L.~Zheng, ``Unequal error protection: An information-theoretic perspective,'' \emph{IEEE Transactions on Information Theory}, vol.~55, no.~12, pp. 5511--5539, 2009.

\bibitem{SR_broadcast_andrea}
C.~T.~K. Ng, D.~Gunduz, A.~J. Goldsmith, and E.~Erkip, ``Distortion minimization in {G}aussian layered broadcast coding with successive refinement,'' \emph{IEEE Transactions on Information Theory}, vol.~55, no.~11, pp. 5074--5086, 2009.

\bibitem{SR_via_broadcast}
C.~Tian, A.~Steiner, S.~Shamai, and S.~N. Diggavi, ``Successive refinement via broadcast: Optimizing expected distortion of a {G}aussian source over a {G}aussian fading channel,'' \emph{IEEE Transactions on Information Theory}, vol.~54, no.~7, pp. 2903--2918, 2008.

\bibitem{utility_max}
M.~Shaqfeh, W.~Mesbah, and H.~Alnuweiri, ``Utility maximization for layered transmission using the broadcast approach,'' \emph{IEEE Transactions on Wireless Communications}, vol.~11, no.~3, pp. 1228--1238, 2012.

\bibitem{1055184}
P.~Bergmans, ``A simple converse for broadcast channels with additive white {G}aussian noise (corresp.),'' \emph{IEEE Transactions on Information Theory}, vol.~20, no.~2, pp. 279--280, 1974.

\bibitem{1237140}
S.~Shamai and A.~Steiner, ``A broadcast approach for a single-user slowly fading mimo channel,'' \emph{IEEE Transactions on Information Theory}, vol.~49, no.~10, pp. 2617--2635, 2003.

\bibitem{7156144}
W.~Yang, G.~Caire, G.~Durisi, and Y.~Polyanskiy, ``Optimum power control at finite blocklength,'' \emph{IEEE Trans.\ Inf.\ Theory}, vol.~61, no.~9, pp. 4598--4615, 2015.

\bibitem{NISTHandbook}
F.~W.~J. Olver, D.~W. Lozier, R.~F. Boisvert, and C.~W. Clark, Eds., \emph{NIST Handbook of Mathematical Functions}.\hskip 1em plus 0.5em minus 0.4em\relax New York: Cambridge University Press, 2010.

\bibitem{timidboldadeel}
A.~Mahmood and A.~B. Wagner, ``Timid/bold coding for channels with cost constraints,'' in \emph{2023 IEEE International Symposium on Information Theory (ISIT)}, 2023, pp. 1442--1447.

\bibitem{gallager_paper_65}
R.~Gallager, ``A simple derivation of the coding theorem and some applications,'' \emph{IEEE Transactions on Information Theory}, vol.~11, no.~1, pp. 3--18, 1965.

\bibitem{gallager1968}
R.~G. Gallager, \emph{Information Theory and Reliable Communication}.\hskip 1em plus 0.5em minus 0.4em\relax New York, NY, USA: Wiley, 1968.

\bibitem{7056434}
V.~Y.~F. Tan and M.~Tomamichel, ``The third-order term in the normal approximation for the awgn channel,'' \emph{IEEE Trans.\ Inf.\ Theory}, vol.~61, no.~5, pp. 2430--2438, 2015.

\bibitem{ppv}
Y.~Polyanskiy, H.~V. Poor, and S.~Verdu, ``Channel coding rate in the finite blocklength regime,'' \emph{IEEE Transactions on Information Theory}, vol.~56, no.~5, pp. 2307--2359, 2010.

\bibitem{6802432}
W.~Yang, G.~Durisi, T.~Koch, and Y.~Polyanskiy, ``Quasi-static multiple-antenna fading channels at finite blocklength,'' \emph{IEEE Transactions on Information Theory}, vol.~60, no.~7, pp. 4232--4265, 2014.

\bibitem{nonlinopt}
J.~Nocedal and S.~J. Wright, \emph{Numerical Optimization}, 2nd~ed., ser. Springer Series in Operations Research and Financial Engineering.\hskip 1em plus 0.5em minus 0.4em\relax New York: Springer, 2006.

\bibitem{SHANNON19576}
\BIBentryALTinterwordspacing
C.~E. Shannon, ``Certain results in coding theory for noisy channels,'' \emph{Information and Control}, vol.~1, no.~1, pp. 6--25, 1957. [Online]. Available: \url{https://www.sciencedirect.com/science/article/pii/S0019995857900396}
\BIBentrySTDinterwordspacing

\bibitem{9099482}
A.~B. Wagner, N.~V. Shende, and Y.~Altuğ, ``A new method for employing feedback to improve coding performance,'' \emph{IEEE Trans.\ Inf.\ Theory}, vol.~66, no.~11, pp. 6660--6681, 2020.

\bibitem{Shevtsova2013}
I.~G. Shevtsova, ``On the absolute constants in the {Berry}--{Esseen} inequality and its structural and nonuniform improvements,'' \emph{Inform. Primen.}, vol.~7, no.~1, pp. 124--125, 2013.

\end{thebibliography}
\end{document}